\documentclass[journal,onecolumn]{IEEEtran}
\usepackage{amsmath,amsfonts,amssymb,amsthm,mathtools}

\usepackage{algorithm}
\usepackage{algpseudocode}
\usepackage{array}
\usepackage[caption=false,font=normalsize,labelfont=sf,textfont=sf]{subfig}
\usepackage{textcomp}
\usepackage{stfloats}
\usepackage{url}
\usepackage{multirow}
\usepackage{verbatim}
\usepackage{graphicx}
\usepackage{cite}
\usepackage{geometry}
\usepackage{float}
\usepackage{enumitem}
\usepackage{tabularx, booktabs, ragged2e}
\newcolumntype{L}{>{\RaggedRight\arraybackslash}X}
\usepackage{placeins}
\usepackage{makecell}
\usepackage{etoolbox} 
\newcommand{\Dd}{\text{Dd}}
\newcommand{\Atom}{\text{Atom}}

\newcommand{\Cn}{\text{Cn}}
\newcommand{\enc}{\text{enc}}

\newcommand{\E}{\mathbb{E}}

\newcommand{\timeindex}{\mathsf{time}}

\newcommand{\syn}{\equiv_{\mathcal L_{\mathrm{sem}}}}

\DeclareMathOperator{\supp}{supp}

\newtheorem{axiom}{Axiom}[section]
\newtheorem{assumption}{Assumption}[section]
\newtheorem{theorem}{Theorem}[section]
\newtheorem{lemma}{Lemma}[section]
\newtheorem{proposition}{Proposition}[section]
\newtheorem{corollary}{Corollary}[section]
\newtheorem{definition}{Definition}[section]
\newtheorem{remark}{Remark}[section]
\newtheorem{example}{Example}[section]

\numberwithin{equation}{section}

\makeatletter
\renewenvironment{proof}[1][\proofname]{\par
  \pushQED{\qed}%
  \normalfont \topsep6\p@\@plus6\p@\relax
  \trivlist
  \item[\hskip\labelsep
        \itshape
    #1\@addpunct{.}]\ignorespaces
}{%
  \popQED\endtrivlist\@endpefalse
}
\makeatother

\begin{document}

\title{Semantic Rate--Distortion Theory:\\ Deductive Compression and Closure Fidelity}

\author{Jianfeng~Xu$^{1}$%
\thanks{$^{1}$Koguan School of Law, China Institute for Smart Justice,
School of Computer Science, Shanghai Jiao Tong University,
Shanghai 200030, China. Email: xujf@sjtu.edu.cn}%
}

\maketitle

\begin{abstract}
Shannon's rate--distortion theory treats source alphabets as
unstructured sets.
When the source is a knowledge base equipped with a proof system, a
natural fidelity criterion is \emph{closure fidelity}: a
reconstruction is acceptable if it preserves the deductive closure of
the original.
This paper develops a rate--distortion theory under this criterion.
The irredundant core \(\Atom(S_O)\)---a canonical irredundant
generating set, extracted by a fixed-order deletion procedure, from
which the full deductive closure can be re-derived---plays a central
role.
We prove that the zero-distortion semantic rate equals
\(P_A\,H(\pi_A)\), where \(P_A\) is the source probability of the
core and \(\pi_A\) the conditional core distribution; this is strictly
below the classical \(H(P_O)\) whenever the knowledge base contains
redundant states.
More generally, the full semantic rate--distortion function decomposes
into a contribution from the core alone: redundant states are
invisible to both rate and distortion.
We derive a semantic source--channel separation theorem exhibiting a
semantic leverage phenomenon: under closure fidelity the required
source rate is reduced, with an asymptotic i.i.d.\ leverage factor
\(\Lambda_\infty=\log|S_O|/(P_A\,H(\pi_A))>1\), enabling the same
knowledge base to be communicated with proportionally fewer channel
uses---not by violating the Shannon capacity \(C(W)\), which remains
an upper bound, but because closure-based fidelity renders redundant
states free.
We also prove a strengthened Fano inequality exploiting core
structure.
For heterogeneous multi-agent communication, an overlap decomposition
yields necessary and sufficient conditions for closure-reliable
transmission and identifies a semantic bottleneck in broadcast
settings that persists even over noiseless carriers.
All results are verified on Datalog instances with up to
24\,000 base facts.
\end{abstract}


\section{Introduction}
\label{sec:introduction}

\IEEEPARstart{S}{hannon's} mathematical theory of
communication~\cite{shannon1948mathematical} deliberately sets aside
the \emph{meaning} of messages, treating source and channel alphabets
as unstructured label sets.
This abstraction has yielded the fundamental limits of data
compression, channel coding, and network information
theory~\cite{cover2006elements,csiszar2011information}, and remains
the foundation of modern digital communications.
Yet already in the companion essay by
Weaver~\cite{weaver2017recent}, a three-level hierarchy was
envisaged: accurate symbol transmission (Level~A), conveying intended
meaning (Level~B), and achieving the desired effect (Level~C).
Shannon's theory addresses Level~A with unmatched precision; the
present paper develops a rigorous \emph{rate--distortion theory} for
Level~B that is mathematically compatible with---and strictly
generalizes---Level~A.
Crucially, the ``meaning'' captured here is not a generic similarity
measure or a black-box task metric, but the \emph{deductive content}
induced by a shared proof system: the fidelity criterion is
preservation of the deductive closure, and the compression gain arises
from the receiver's ability to re-derive redundant knowledge---a
mechanism that is structurally distinct from both classical
symbol-level coding and the task/effect-oriented metrics studied at
Level~C.

\subsection*{Motivation}

The need for such a theory has intensified with the rise of
knowledge-intensive communication systems.
In multi-agent coordination, retrieval-augmented generation pipelines,
and federated knowledge-base synchronization, communicating parties
exchange not raw symbols but structured \emph{semantic states}---ground
facts, rules, or queries drawn from a shared or partially overlapping
logical vocabulary.
A symbol-level error that leaves the deductive content unchanged is
harmless, while a symbol-level ``success'' that corrupts a single
irredundant axiom may destroy an entire branch of derivable knowledge.
Classical information theory cannot make this distinction because it
treats every source symbol as equally ``expensive'' to protect.

The key observation motivating this work is that a fixed proof system
shared by communicating agents induces a natural partition of the
source alphabet into an \emph{irredundant core}
\(\Atom(S_O)\)---a canonical irredundant generating set, extracted
via a deterministic deletion procedure under a fixed canonical order,
from which the full deductive closure can be
re-derived---and a collection of \emph{stored shortcuts}
\(J = S_O \setminus \Atom(S_O)\) that are deductively redundant.
Note that \(\Atom(S_O)\) is irredundant (no element can be removed
without losing closure equivalence) but need not have minimum
cardinality among all generating subsets; the
order-dependent extraction procedure of
Definition~\ref{def:atom-so} determines a unique canonical
representative.
Under a \emph{closure-based fidelity criterion} (which deems a
reconstruction acceptable if it preserves the deductive closure), only
the \(|\Atom(S_O)|\) core elements need to be transmitted reliably;
the remaining \(|J|\) states can be recovered by the receiver's
inference engine at zero additional channel cost.
This suggests that the \emph{effective source entropy} for
closure-reliable communication should be strictly less than the
Shannon entropy \(H(P_O)\)---a qualitative prediction that we make
precise and tight in this paper.

\subsection*{Related Work and Gap Identification}

\emph{Semantic information theory.}\;
The formal study of semantic information dates to Carnap and
Bar-Hillel~\cite{carnap1952outline}, who measured information by the
set of possible worlds excluded, and to
Floridi~\cite{floridi2004outline}, who required information to be
truthful.
Kolchinsky and Wolpert~\cite{kolchinsky2018semantic} proposed a
viability-based formulation.
Recently, Niu and
Zhang~\cite{niu2024mathematical,zhang2024modern} established a
mathematical theory based on \emph{synonymous mappings}, deriving
semantic entropy, capacity (\(C_s \ge C\)), and rate--distortion
functions by partitioning source alphabets into equivalence classes.
Their framework elegantly extends Shannon theory through source-side
collapsing and leaves room for integrating logical inference
structure---the direction pursued here.
We note that their result \(C_s\ge C\) quantifies a \emph{semantic
throughput gain} under a synonymous-mapping fidelity criterion; it
does not assert that the Shannon channel capacity \(C(W)\) is
violated, since the data processing inequality
\(C_{\mathrm{sem}}\le C(W)\)
\textup{(Theorem~\ref{thm:data-processing}(i))} remains in force.
Our framework produces an analogous but distinct gain via
\emph{receiver-side} deductive reconstruction rather than source-side
equivalence-class collapsing.

\emph{Semantic coding algorithms.}\;
Ma et~al.~\cite{ma2025theory} proved a semantic channel coding theorem
for many-to-one sources via a generalized Fano inequality.
Han et~al.~\cite{han2025extended} proposed an extended
Blahut--Arimoto algorithm for computing the semantic rate--distortion
function under synonymous mappings.
Liang et~al.~\cite{liang2025semantic} introduced semantic arithmetic
coding achieving higher compression efficiency over synonymous sets.
These works provide coding-theoretic results within the
synonymous-mapping paradigm; our work complements them by grounding
the compression mechanism in \emph{receiver-side deductive inference}
rather than source-side equivalence-class collapsing.

\emph{Multi-agent and goal-oriented communication.}\;
Seo et~al.~\cite{seo2023bayesian} addressed inverse contextual
reasoning via Bayesian inference.
Alshammari and Bennis~\cite{alshammari2026logic} proposed a
logic-driven framework for resilient multi-agent semantic
communication.
Li et~al.~\cite{li2024toward} and Wu et~al.~\cite{wu2024toward}
developed goal-oriented metrics unifying significance measures.
Deep-learning-based semantic communication
systems~\cite{xie2021deep,qin2021semantic,gunduz2022beyond,luo2022semantic}
motivate the formal guarantees developed here; our overlap
decomposition and semantic bottleneck results provide
coding-theoretic complements to these approaches.

\emph{Logical and database foundations.}\;
The logical substrate draws on descriptive complexity
(\(\mathrm{FO(LFP)}\) capturing PTIME on ordered finite
structures~\cite{immerman1999descriptive}), Datalog fixpoint
theory~\cite{ceri1989you,abiteboul1995foundations}, and knowledge
compilation~\cite{dantsin2001complexity}.
Marx~\cite{marx2013tractable} introduced submodular width for
conjunctive query evaluation; Abo~Khamis and
Chen~\cite{abokhamis2025jaguar} recently achieved near-optimal
query-evaluation complexity.
Mu~\cite{mu2024identifying} studied formula roles in inconsistency,
paralleling our core/redundant distinction.
These results provide the technical substrate on which our
communication-theoretic analysis builds.

\emph{Gap.}\;
The works reviewed above illuminate complementary facets of semantic
communication, yet a unified framework that derives
\emph{tight rate--distortion bounds} from the proof-system structure
of the communicated knowledge base---and that handles the vocabulary
heterogeneity arising in multi-agent settings---has not been
established.
The present paper addresses this gap.

\subsection*{Approach and Key Insight}

Our approach originates from an ontological perspective on
information~\cite{xu2014objective,xu2024research,xu2025general,qiu2025research}:
we model a source as a finite knowledge base \(S_O\) equipped with a
deductive closure operator \(\Cn(\cdot)\) induced by a fixed proof
system~\(\mathsf{PS}\), and define \emph{closure distortion}
\(d_{\Cn}\)---a Jaccard-distance-based measure under which replacing a
source state by any deductively equivalent substitute incurs zero
penalty.
Under this distortion, redundant states (those in
\(\Cn(S_O \setminus \{s\})\)) are ``free'': they can be mapped to
any output in the closure without cost.
Only the irredundant core elements carry genuine rate, because they
contribute unique deductive content that cannot be recovered from the
remaining knowledge base.

The central quantitative consequence is a \emph{tight} zero-distortion
semantic rate:
\begin{equation}\label{eq:intro-tight-rate}
  R_{\mathrm{sem}}(0;\,d_{\Cn},\,P_O)
  \;=\;
  P_A\,H(\pi_A),
\end{equation}
where \(P_A = P_O(\Atom(S_O))\) and
\(\pi_A(a) = P_O(a)/P_A\) is the conditional core distribution
(Theorem~\ref{thm:tight-zero-rate}).
This is strictly below the classical \(R(0;d_H) = H(P_O)\) whenever
\(|J| > 0\), and the ratio
\(P_A\,H(\pi_A)/H(P_O)\) quantifies the \emph{deductive compression
gain}.
Under a uniform source, the ratio reduces to
\(k\log k/(|S_O|\log |S_O|)\)
where \(k = |\Atom(S_O)|\).

\subsection*{Main Contributions}

\begin{enumerate}[label=\textup{C\arabic*.},leftmargin=*]
  \item \emph{Axiomatic information model
        \textup{(Section~\ref{sec:model},
        Appendix~\ref{app:axioms})}:}
        We introduce a framework comprising
        \(\mathcal L_{\mathrm{sem}}\)-definable state sets linked by
        computable enabling maps, a deductive closure
        operator~\(\Cn\) with a fixed proof
        system~\(\mathsf{PS}\), and an irredundant semantic
        core~\(\Atom(S_O)\) with derivation-depth
        stratification~\(\Dd(\cdot\mid B)\), both shown to be
        computable invariants.

  \item \emph{Semantic channel
        \textup{(Section~\ref{sec:channel})}:}
        We define the semantic channel as a composition of Markov
        kernels whose supports respect the enabling structure of
        composable information models, and introduce four distortion
        functions of increasing semantic depth: Hamming, closure,
        depth, and a parameterized composite.

  \item \emph{Tight semantic rate--distortion function
        \textup{(Section~\ref{sec:fundamental-limits})}:}
        We prove that the zero-distortion rate under closure
        fidelity is exactly \(P_A\,H(\pi_A)\)
        \textup{(Theorem~\ref{thm:tight-zero-rate})}, and that the
        full \(R_{\mathrm{sem}}(D)\) decomposes into a contribution
        from the core sub-source alone
        \textup{(Theorem~\ref{thm:rd-decomposition})}.

  \item \emph{Semantic source--channel separation and semantic
        leverage
        \textup{(Section~\ref{subsec:sem-separation})}:}
        We derive a separation theorem showing that under closure
        fidelity the required source rate is reduced by a
        semantic leverage factor
        \(\Lambda_1 = \log|S_O|/\log|\Atom(S_O)| > 1\)
        in the single-shot regime, and
        \(\Lambda_\infty = \log|S_O|/(P_A\,H(\pi_A))\ge\Lambda_1\)
        in the asymptotic i.i.d.\ regime
        \textup{(Theorem~\ref{thm:sem-source-channel})},
        enabling the same knowledge base to be communicated
        with proportionally fewer channel uses---not by
        exceeding the Shannon capacity \(C(W)\), but because
        closure fidelity renders redundant states free.

  \item \emph{Strengthened semantic Fano inequality
        \textup{(Section~\ref{subsec:strengthened-fano})}:}
        We prove a Fano bound involving \(\log|A|\) rather than
        \(\log|S_O|\), yielding a tighter constraint by
        \(\log(|S_O|/|A|)\) bits
        \textup{(Theorem~\ref{thm:semantic-fano-tight})}.

  \item \emph{Rate--delay--distortion tradeoff and semantic
        sampling theorem
        \textup{(Section~\ref{subsec:rate-delay})}:}
        When the receiver's derivation budget is bounded by
        \(\delta\) steps of \(T_{\mathsf{PS}}\), we define a
        \(\delta\)-irredundant core filtration
        \(S_O=\Atom_0\supseteq\Atom_1\supseteq\cdots
        \supseteq\Atom_{\mathsf{D_d}}=\Atom(S_O)\)
        and prove that the zero-distortion rate under delay
        budget~\(\delta\) is
        \(R_{\mathrm{sem}}(0,\delta)=P_\delta\,H(\pi_\delta)\),
        yielding a rate--delay--distortion surface that smoothly
        interpolates between the classical \(H(P_O)\) at
        \(\delta=0\) and the semantic rate \(P_A\,H(\pi_A)\) at
        \(\delta=\mathsf{D_d}\)
        \textup{(Theorem~\ref{thm:rate-delay})}.
        A semantic sampling theorem identifies the critical delay
        \(\delta^*\) below which closure-reliable communication
        is impossible
        \textup{(Corollary~\ref{cor:semantic-nyquist})}. 
        An expected-budget relaxation via time-sharing yields a
        convex rate--computation tradeoff whose Lagrangian dual
        prices receiver computation in bits of communication rate
        \textup{(Remark~\ref{rem:comm-comp-exchange})}.

  \item \emph{Heterogeneous multi-agent communication
        \textup{(Section~\ref{sec:application})}:}
        We introduce a pairwise overlap decomposition, derive
        necessary and sufficient conditions for closure-reliable
        communication, show that the deductive compression ratio is
        invariant under vocabulary heterogeneity, and identify a
        semantic bottleneck phenomenon in broadcast settings.
        All results are verified on explicit Datalog instances with
        up to 24\,000 base facts.
\end{enumerate}

\subsection*{Paper Organization}

Section~\ref{sec:model} presents the system model: the deductive
closure operator, irredundant cores, derivation depth, and closure
fidelity.
Section~\ref{sec:channel} builds the probabilistic layer: enabling
kernels, the semantic channel, distortion measures, channel
invariants, and preliminary coding theorems.
Section~\ref{sec:fundamental-limits} derives the tight
rate--distortion bounds, the source--channel separation theorem, and
the strengthened Fano inequality.
Section~\ref{sec:application} instantiates the framework for
heterogeneous multi-agent communication and presents the numerical
validation.
Section~\ref{sec:conclusion} concludes.
Appendix~\ref{app:axioms} provides the full axiomatic foundations:
the logical language, information model axioms, and synonymous state
sets.

\textit{Notation.}\;
\(S_O\): semantic state set (knowledge base);
\(\Cn(\cdot)\): deductive closure;
\(\Atom(S_O)\): irredundant core;
\(\Dd(s\mid B)\): derivation depth;
\(\kappa:X\rightsquigarrow Y\): Markov kernel;
\(H(\cdot)\), \(I(\cdot;\cdot)\): Shannon entropy and mutual
information (base~2, bits);
\(d_H,d_{\Cn},d_{\Dd},d_{\mathrm{sem}}\): distortion functions;
\(P_A := P_O(\Atom(S_O))\);
\(\pi_A\): conditional core distribution.
A notation summary for the multi-agent application appears in
Table~\ref{tab:notation-iv}.

\section{System Model}
\label{sec:model}

This section introduces the deterministic and logical substrate on
which the probabilistic structure of the semantic channel
(Section~\ref{sec:channel}) and the rate--distortion analysis
(Section~\ref{sec:fundamental-limits}) are built.
The presentation focuses on the three concepts that enter directly
into the main theorems: the deductive closure operator~\(\Cn\), the
irredundant core~\(\Atom(S_O)\), and the derivation-depth
stratification~\(\Dd(\cdot\mid B)\).
The full axiomatic development---the many-sorted logical
language~\(\mathcal L\), the semantic
sublanguage~\(\mathcal L_{\mathrm{sem}}\), the information model
with time-indexed state sets and enabling maps, and the synonymous
state set formalism---is deferred to Appendix~\ref{app:axioms}.
Throughout, the operational inference substrate is a Datalog or
Horn-clause proof system over a finite active domain; the
\(\mathrm{FO(LFP)}\) language of
Appendix~\ref{subsec:logic} serves as a background framework
for defining expressible state sets and does not enter the
rate--distortion analysis directly.

\subsection{Proof System and Deductive Closure}
\label{subsec:proof-closure}

We fix a finite set~\(\mathbb{S}_O\) of potential semantic states
(the \emph{ambient semantic universe}) equipped with an injective
encoding \(\enc_O:\mathbb{S}_O\to\{0,1\}^*\) and a fixed canonical
order.
Each element of~\(\mathbb{S}_O\) is identified with a ground atom
of an inference fragment
\(\mathcal L_{\mathrm{kb}}\subseteq\mathcal L_{\mathrm{sem}}\)
(typically Datalog or a Horn fragment; see
Appendix~\ref{subsec:logic} for the full logical language).
Additional closure properties of~\(\mathbb{S}_O\) (closure under
definable recodings, effective representability) are stated in
Assumption~\ref{assump:semantic-universe} of
Appendix~\ref{app:axioms}.

\begin{assumption}[Fixed effective proof system]
\label{assump:proof-system}
We fix an effective proof system~\(\mathsf{PS}\) over the syntax
of~\(\mathcal L_{\mathrm{kb}}\) such that proof checking is
decidable.
For any finite
\(\Gamma\subseteq\mathcal L_{\mathrm{kb}}\) and any
\(\varphi\in\mathcal L_{\mathrm{kb}}\), write
\(\Gamma\vdash_{\mathrm{kb}}\varphi\) for derivability
in~\(\mathsf{PS}\), and define the \emph{deductive closure
operator}
\[
  \Cn(\Gamma)
  \;:=\;
  \bigl\{\varphi\in\mathcal L_{\mathrm{kb}}:\
  \Gamma\vdash_{\mathrm{kb}}\varphi\bigr\}.
\]
The operator~\(\Cn\) satisfies three standing properties used
throughout without further comment:
\begin{enumerate}[label=\textup{(Cn\arabic*)}]
  \item \emph{Reflexivity:}
        \(\Gamma\subseteq\Cn(\Gamma)\).
  \item \emph{Monotonicity:}
        \(\Gamma\subseteq\Gamma'\Rightarrow
        \Cn(\Gamma)\subseteq\Cn(\Gamma')\).
  \item \emph{Idempotence:}
        \(\Cn\bigl(\Cn(\Gamma)\bigr)=\Cn(\Gamma)\).
\end{enumerate}
When elements of a state set~\(S_O\subseteq\mathbb{S}_O\) serve as
premises for~\(\Cn\), each state is identified with the corresponding
ground atom of~\(\mathcal L_{\mathrm{kb}}\).
\end{assumption}

\begin{assumption}[Finite and effectively listable knowledge bases]
\label{assump:finite-so}
The knowledge bases~\(S_O\subseteq\mathbb{S}_O\) considered in this
paper are finite and effectively listable under the fixed canonical
order.
Throughout, \(\mathbb{S}_O\) denotes the \emph{active-domain}
semantic universe: for each problem instance (logical structure,
domain, rule set), \(\mathbb{S}_O\) is the finite set of all
ground atoms over the active domain---typically the Herbrand base
restricted to the constants appearing in the
instance~\cite{abiteboul1995foundations}.
The finiteness of~\(\mathbb{S}_O\) is thus a per-instance property,
not a restriction on the logical language itself.
\end{assumption}

\begin{assumption}[Effective redundancy test]
\label{assump:core-extractable}
For the knowledge bases considered, the predicate
\(s\in\Cn(\Gamma)\) is decidable whenever
\(\Gamma\subseteq S_O\) is finite and \(s\in S_O\).
This holds in Datalog/Horn settings and bounded-domain
theories~\cite{abiteboul1995foundations,dantsin2001complexity}.
\end{assumption}

\subsection{Irredundant Core and Derivation Depth}
\label{subsec:atomic-derivation}

\begin{definition}[Irredundant core]
\label{def:atom-so}
Let~\(S_O\) be a finite knowledge base.
Define \(\Atom(S_O)\) by the following deterministic procedure:
initialize \(A\gets S_O\); scan elements of~\(S_O\) in canonical
order; for each~\(s\), if \(s\in\Cn(A\setminus\{s\})\), set
\(A\gets A\setminus\{s\}\); output~\(A\).
\end{definition}

\begin{proposition}[Core correctness]
\label{prop:atom-core-correct}
Under
Assumptions~\textup{\ref{assump:finite-so}}--\textup{\ref{assump:core-extractable}},
the set \(A:=\Atom(S_O)\) satisfies:
\textup{(i)}~\(\Cn(A)=\Cn(S_O)\);
\textup{(ii)}~for every \(a\in A\),
\(a\notin\Cn(A\setminus\{a\})\) \textup{(irredundancy)};
\textup{(iii)}~\(\Atom(S_O)\) is uniquely determined by~\(S_O\) and
the canonical order;
\textup{(iv)}~\(S_O\subseteq\Cn(\Atom(S_O))\).
\end{proposition}

\begin{proof}
Each removal preserves the closure: if \(s\in\Cn(A\setminus\{s\})\),
then \(A\subseteq\Cn(A\setminus\{s\})\), so
\(\Cn(A)\subseteq\Cn(\Cn(A\setminus\{s\}))=\Cn(A\setminus\{s\})\)
by~\textup{(Cn2)} and~\textup{(Cn3)}.
By induction over the scan, \(\Cn(A)=\Cn(S_O)\),
giving~\textup{(i)}.
Canonicality~\textup{(iii)} is immediate from the
determinism of the procedure and the fixed canonical order.
Part~\textup{(iv)} follows from~\textup{(i)}:
reflexivity~\textup{(Cn1)} gives
\(S_O\subseteq\Cn(S_O)=\Cn(A)\).
Irredundancy~\textup{(ii)}: let \(A_{\mathrm{scan}}\) denote
the current set at the moment~\(a\) is scanned, and let
\(A_{\mathrm{final}}\) denote the output.
Since \(a\) was retained,
\(a\notin\Cn(A_{\mathrm{scan}}\setminus\{a\})\).
Subsequent removals only shrink the set, so
\(A_{\mathrm{final}}\subseteq A_{\mathrm{scan}}\) and hence
\(A_{\mathrm{final}}\setminus\{a\}\subseteq
A_{\mathrm{scan}}\setminus\{a\}\).
By monotonicity~\textup{(Cn2)},
\(\Cn(A_{\mathrm{final}}\setminus\{a\})\subseteq
\Cn(A_{\mathrm{scan}}\setminus\{a\})\), so
\(a\notin\Cn(A_{\mathrm{final}}\setminus\{a\})\).
\end{proof}

\begin{definition}[Core and shortcuts]
\label{def:intrinsic-operational-bases}
For a knowledge base~\(S_O\), write
\(A:=\Atom(S_O)\) (core premises) and
\(J:=S_O\setminus A\) (stored shortcuts).
\end{definition}

\begin{remark}[Proof-system dependence of the core]
\label{rem:ps-dependence}
The irredundant core \(\Atom(S_O)\) depends on two choices: the proof
system~\(\mathsf{PS}\) \textup{(}which determines the closure
operator~\(\Cn\)\textup{)} and the canonical order on~\(S_O\)
\textup{(}which resolves ties in the deletion procedure of
Definition~\textup{\ref{def:atom-so})}.
Changing either may alter \(\Atom(S_O)\) while preserving the closure
equivalence \(\Cn(\Atom(S_O))=\Cn(S_O)\).
All information-theoretic quantities in this paper---the semantic rate
\(P_A H(\pi_A)\), the leverage factor~\(\Lambda\), the compression
ratio---are therefore \emph{proof-system-relative} structural
invariants.
This is a feature, not a limitation: it captures the fact that the
``value'' of a semantic state for communication depends on the
inference capabilities shared by sender and receiver.
\end{remark}

\smallskip
\noindent\textit{Derivation depth via the immediate consequence
operator.}

\begin{definition}[Immediate consequence operator]
\label{def:T-operator}
Write \(\Gamma\vdash_{\mathrm{kb}}^{1}s\) for single-step
derivability.
Define
\(T_{\mathsf{PS}}(\Gamma):=\Gamma\cup
\{s\in\mathbb{S}_O:\Gamma\vdash_{\mathrm{kb}}^{1}s\}\),
with iteration
\(T^0_{\mathsf{PS}}(\Gamma):=\Gamma\) and
\(T^{n+1}_{\mathsf{PS}}(\Gamma):=
T_{\mathsf{PS}}(T^n_{\mathsf{PS}}(\Gamma))\).
\end{definition}

\begin{axiom}[Properties of \(T_{\mathsf{PS}}\)]
\label{ax:T-operator}
The operator satisfies:
\textup{(IC1)}~monotonicity;
\textup{(IC2)}~computability (finite output for finite input);
\textup{(IC3)}~closure characterization:
\(\Cn(\Gamma)=\bigcup_{n\ge 0}T^n_{\mathsf{PS}}(\Gamma)\);
\textup{(IC4)}~finite stabilization.
These hold in Datalog and Horn-clause settings over finite
domains~\cite{abiteboul1995foundations}.
\end{axiom}

\begin{definition}[Derivation depth]
\label{def:derivation-depth}
For finite \(B\subseteq\mathbb{S}_O\) and
\(s\in\mathbb{S}_O\),
\[
  \Dd(s\mid B)
  \;:=\;
  \min\bigl\{n\ge 0:s\in T^n_{\mathsf{PS}}(B)\bigr\},
\]
with \(\Dd(s\mid B):=\infty\) if \(s\notin\Cn(B)\).
\end{definition}

\begin{lemma}[Properties of derivation depth]
\label{lem:depth-properties}
Under Axiom~\textup{\ref{ax:T-operator}}, for finite
\(B\subseteq\mathbb{S}_O\):
\textup{(i)}~\(\Dd(s\mid B)\) is a unique, finite, computable
non-negative integer for every \(s\in\Cn(B)\);
\textup{(ii)}~\(\Dd(s\mid B)=0\) iff \(s\in B\);
\textup{(iii)}~if \(B\subseteq B'\) and \(s\in\Cn(B)\), then
\(\Dd(s\mid B')\le\Dd(s\mid B)\).
\end{lemma}

\begin{proof}
Part~\textup{(i)}: by~\textup{(IC4)}, the chain stabilizes at
\(\Cn(B)\) in finitely many steps, so
\(\{n:s\in T^n(B)\}\neq\varnothing\) and its minimum is finite;
computability follows from~\textup{(IC2)}.
Part~\textup{(ii)}: \(T^0(B)=B\).
Part~\textup{(iii)}: monotonicity~\textup{(IC1)} gives
\(T^n(B)\subseteq T^n(B')\) by induction, so
\(\min\{n:s\in T^n(B')\}\le\min\{n:s\in T^n(B)\}\).
\end{proof}

\begin{definition}[Intrinsic and operational depths]
\label{def:int-op-depth}
For \(q\in\Cn(A)=\Cn(S_O)\), define
\(n_{\mathrm{int}}(q):=\Dd(q\mid A)\) and
\(n_{\mathrm{op}}(q):=\Dd(q\mid S_O)\).
\end{definition}

\begin{definition}[Semantic atomicity]
\label{def:atomicity-measure}
The \emph{semantic atomicity} of an information model~\(\mathcal I\)
with semantic space~\(S_O\) is
\(\mathsf{A}(\mathcal I):=|\Atom(S_O)|\).
\end{definition}

\begin{definition}[Maximum intrinsic derivation depth]
\label{def:max-depth}
The \emph{maximum intrinsic derivation depth} of~\(\mathcal I\) is
\[
  \mathsf{D_d}(\mathcal I)
  \;:=\;
  \max_{q\,\in\, S_O}\;\Dd\bigl(q \mid \Atom(S_O)\bigr),
\]
with the convention \(\max\varnothing:=0\).
\end{definition}

\begin{theorem}[Computable semantic invariants]
\label{thm:semantic-invariants}
Under
Assumptions~\textup{\ref{assump:finite-so}}--\textup{\ref{assump:core-extractable}}
and Axiom~\textup{\ref{ax:T-operator}}:
\begin{enumerate}[label=\textup{(\roman*)}]
  \item \(\mathsf{A}\) and~\(\mathsf{D_d}\) are uniquely
        determined, finite, and computable.
  \item \(n_{\mathrm{op}}(q)\le n_{\mathrm{int}}(q)
        \le\mathsf{D_d}\) for every \(q\in S_O\).
  \item \(\mathsf{D_d}=0\) if and only if
        \(\Atom(S_O)=S_O\).
\end{enumerate}
\end{theorem}

\begin{proof}
Part~\textup{(i)}: \(\Atom(S_O)\) is uniquely determined by
Proposition~\ref{prop:atom-core-correct}(iii) and computable by
the deterministic procedure of Definition~\ref{def:atom-so};
each \(\Dd(q\mid A)\) is finite and computable by
Lemma~\ref{lem:depth-properties}(i); the maximum over the finite
set~\(S_O\) is computable by enumeration.
Part~\textup{(ii)}: \(A\subseteq S_O\) gives
\(\Dd(q\mid S_O)\le\Dd(q\mid A)\) by
Lemma~\ref{lem:depth-properties}(iii); the bound by
\(\mathsf{D_d}\) is immediate from the definition.
Part~\textup{(iii)}: if \(A=S_O\), every
\(q\in S_O=T^0(A)\) has depth~\(0\); conversely, depth~\(0\) for
all~\(q\) implies \(S_O\subseteq A\) by
Lemma~\ref{lem:depth-properties}(ii), and \(A\subseteq S_O\) by
construction.
\end{proof}

\subsection{Noisy Information and Closure Fidelity}
\label{subsec:noisy}

When the sender's knowledge base~\(S_O\) and the receiver's
reconstructed space~\(\hat S_O\) differ, the discrepancy is
captured by a \emph{noise pair}.

\begin{definition}[Noisy semantic base]
\label{def:noisy}
A \emph{noisy semantic base} of~\(S_O\) is any set
\(\tilde S_O:=(S_O\setminus S_O^-)\cup S_O^+\), where
\(S_O^-\subseteq S_O\) (lost states) and
\(S_O^+\subseteq\mathbb{S}_O\setminus S_O\) (spurious states).
The pair \((S_O^-,S_O^+)\) is the \emph{noise pair}; it is
\emph{trivial} when both sets are empty.
\end{definition}

\begin{definition}[Closure fidelity]
\label{def:closure-fidelity}
For finite \(S,\hat S\subseteq\mathbb{S}_O\),
\[
  \mathsf{F}_{\Cn}(S,\hat S)
  \;:=\;
  \frac{|\Cn(S)\cap\Cn(\hat S)|}
       {|\Cn(S)\cup\Cn(\hat S)|},
\]
with \(0/0:=1\).
We have \(\mathsf{F}_{\Cn}=1\) iff \(\Cn(S)=\Cn(\hat S)\).
\end{definition}

\begin{definition}[Core preservation ratio]
\label{def:atom-preservation}
For \(A=\Atom(S_O)\) and any
\(\hat S\subseteq\mathbb{S}_O\),
\(\rho_{\Atom}(S_O,\hat S):=|A\cap\hat S|/|A|\)
(with \(0/0:=1\)).
\end{definition}

\begin{proposition}[Noise pair, core preservation, and closure fidelity]
\label{prop:noise-fidelity}
Let \(\tilde S_O=(S_O\setminus S_O^-)\cup S_O^+\) and
\(A=\Atom(S_O)\).
\begin{enumerate}[label=\textup{(\roman*)}]
  \item \(\rho_{\Atom}(S_O,\tilde S_O)=1\) iff
        \(A\cap S_O^-=\varnothing\).
  \item If \(A\cap S_O^-=\varnothing\), then
        \(\Cn(S_O)\subseteq\Cn(\tilde S_O)\).
  \item If \(A\cap S_O^-=\varnothing\) and
        \(S_O^+\subseteq\Cn(S_O)\), then
        \(\Cn(S_O)=\Cn(\tilde S_O)\) and
        \(\mathsf{F}_{\Cn}(S_O,\tilde S_O)=1\).
  \item Trivial noise implies \(\rho_{\Atom}=1\) and
        \(\mathsf{F}_{\Cn}=1\).
\end{enumerate}
\end{proposition}

\begin{proof}
\textup{(i)}:
Since \(A\subseteq S_O\) and
\(S_O^+\subseteq\mathbb{S}_O\setminus S_O\), we have
\(A\cap S_O^+=\varnothing\), hence
\(A\cap\tilde S_O=A\setminus S_O^-\) and
\(\rho_{\Atom}=|A\setminus S_O^-|/|A|=1\) iff
\(A\cap S_O^-=\varnothing\).

\smallskip\noindent
\textup{(ii)}:
If \(A\cap S_O^-=\varnothing\), then
\(A\subseteq S_O\setminus S_O^-\subseteq\tilde S_O\).
By~\textup{(Cn2)},
\(\Cn(A)\subseteq\Cn(\tilde S_O)\), and
\(\Cn(A)=\Cn(S_O)\) by
Proposition~\ref{prop:atom-core-correct}(i).

\smallskip\noindent
\textup{(iii)}:
By~\textup{(ii)}, \(\Cn(S_O)\subseteq\Cn(\tilde S_O)\).
For the reverse:
\(S_O\setminus S_O^-\subseteq S_O\subseteq\Cn(S_O)\)
by~\textup{(Cn1)}, and \(S_O^+\subseteq\Cn(S_O)\) by hypothesis,
so \(\tilde S_O\subseteq\Cn(S_O)\).
By~\textup{(Cn2)} and~\textup{(Cn3)},
\(\Cn(\tilde S_O)\subseteq\Cn(\Cn(S_O))=\Cn(S_O)\).

\smallskip\noindent
\textup{(iv)}: Immediate from \(\tilde S_O=S_O\).
\end{proof}

\begin{remark}[Zero-distortion property of redundant states]
\label{rem:redundant-free}
Proposition~\textup{\ref{prop:noise-fidelity}(iii)} has a
per-state counterpart crucial for the rate--distortion
analysis: if \(j\in J=S_O\setminus\Atom(S_O)\), then
\(j\in\Cn(S_O\setminus\{j\})\), so replacing~\(j\) by any
\(\hat s\in\Cn(S_O)\) preserves the deductive closure.
Errors on redundant states incur zero closure distortion---a
property absent from any classical distortion measure and the
source of the deductive compression gain formalized in
Section~\textup{\ref{sec:fundamental-limits}}.
\end{remark}

\section{Semantic Channel}
\label{sec:channel}

This section erects the probabilistic layer on the structural
framework of Section~\ref{sec:model}.
The central object is the \emph{semantic channel}: a composition of
Markov kernels---encoding, carrier transmission, and decoding---each
constrained by the enabling structure of an underlying information
model (Definition~\ref{def:info-instance}).

\smallskip
\noindent\textbf{Notation.}
All state spaces are finite.
A \emph{probability distribution} on a nonempty finite set~\(S\) is a
function \(P:S\to[0,1]\) with \(\sum_s P(s)=1\); write
\(\Delta(S)\) for the probability simplex and
\(\supp(P):=\{s:P(s)>0\}\).
A \emph{Markov kernel} \(\kappa:X\rightsquigarrow Y\) is a function
\(\kappa:X\times Y\to[0,1]\) with
\(\kappa(\cdot\mid x)\in\Delta(Y)\) for each~\(x\);
kernels compose by
\((\kappa_2\circ\kappa_1)(z\mid x)
:=\sum_y\kappa_1(y\mid x)\,\kappa_2(z\mid y)\).
A kernel is \emph{deterministic} if
\(|\supp(\kappa(\cdot\mid x))|=1\) for all~\(x\).
Shannon entropy, conditional entropy, and mutual information are
denoted \(H(\cdot)\), \(H(\cdot\mid\cdot)\), and
\(I(\cdot\,;\cdot)\) (base~2,
bits)~\cite{cover2006elements}.
Random variables are in sans-serif
(\(\mathsf{S}_o,\hat{\mathsf{S}}_o\)); expectations are~\(\E[\cdot]\).

\subsection{Enabling Kernels and the Semantic Channel}
\label{subsec:enabling-channel}

\begin{definition}[Semantic source]
\label{def:semantic-source}
A \emph{semantic source} is a pair \((S_O,P_O)\) with
\(P_O\in\Delta(S_O)\).
It is \emph{full-support} if \(\supp(P_O)=S_O\) and
\emph{uniform} if \(P_O\equiv 1/|S_O|\).
\end{definition}

\begin{definition}[Enabling kernel]
\label{def:enabling-kernel}
Let \(\mathcal I\) be an information model with enabling map
\(\mathcal E:S_O\Rightarrow S_C\)
(Axiom~\ref{ax:enabling-mapping}).
An \emph{enabling kernel} for~\(\mathcal I\) is a Markov kernel
\(\kappa:S_O\rightsquigarrow S_C\) satisfying
\begin{equation}\label{eq:enabling-support}
  \supp\bigl(\kappa(\cdot\mid s_o)\bigr)
  \;\subseteq\;
  \mathcal E(s_o),
  \qquad\forall\,s_o\in S_O.
\end{equation}
Write \(\mathcal K(\mathcal I)\) for the set of all enabling kernels
for~\(\mathcal I\).
By Axiom~\ref{ax:enabling-mapping}(E3), the deterministic kernel
\(\kappa_e(s_c\mid s_o):=\mathbf{1}[s_c=e(s_o)]\) belongs to
\(\mathcal K(\mathcal I)\), so \(\mathcal K(\mathcal I)\neq\varnothing\).
\end{definition}

\begin{proposition}[Enabling kernels compose]
\label{prop:enabling-compose}
If \((\mathcal I_1,\mathcal I_2)\) is composable
\textup{(Definition~\ref{def:model-composition})} and
\(\kappa_i\in\mathcal K(\mathcal I_i)\) for \(i=1,2\), then
\(\kappa_2\circ\kappa_1\in
\mathcal K(\mathcal I_2\circ\mathcal I_1)\).
\end{proposition}

\begin{proof}
If \((\kappa_2\circ\kappa_1)(s'\mid s_o)>0\), there exists
\(s_c\) with \(\kappa_1(s_c\mid s_o)>0\) and
\(\kappa_2(s'\mid s_c)>0\).
The enabling constraints give
\(s_c\in\mathcal E_1(s_o)\) and
\(s'\in\mathcal E_2(s_c)\subseteq\mathcal E_{2\circ 1}(s_o)\).
\end{proof}

We now define the semantic channel as a three-stage composition.

\begin{definition}[Semantic channel]
\label{def:semantic-channel}
A \emph{semantic channel} is a tuple
\(\mathfrak C=(\mathcal I,\,\mathcal I_{\mathrm{ch}},\,
\mathcal I_{\mathrm{dec}},\,\kappa_{\enc},\,W,\,D)\),
where:
\begin{enumerate}[label=\textup{(\roman*)}]
  \item \(\mathcal I\) is the sender's information model
        (Definition~\ref{def:info-instance}) with encoding kernel
        \(\kappa_{\enc}\in\mathcal K(\mathcal I)\),
        mapping \(S_O\rightsquigarrow S_C\);
  \item \(\mathcal I_{\mathrm{ch}}\) is a carrier channel model
        with \(S_O(\mathcal I_{\mathrm{ch}})=S_C\),
        carrier state set \(\hat S_C\), and carrier channel
        kernel \(W\in\mathcal K(\mathcal I_{\mathrm{ch}})\),
        mapping \(S_C\rightsquigarrow\hat S_C\);
  \item \(\mathcal I_{\mathrm{dec}}\) is a decoding model with
        \(S_O(\mathcal I_{\mathrm{dec}})=\hat S_C\),
        reconstructed space
        \(\hat S_O\subseteq\mathbb{S}_O\), and decoding kernel
        \(D\in\mathcal K(\mathcal I_{\mathrm{dec}})\),
        mapping \(\hat S_C\rightsquigarrow\hat S_O\).
\end{enumerate}
The \emph{end-to-end kernel} is
\begin{equation}\label{eq:semantic-channel-kernel}
  \kappa_{\mathrm{sem}}
  :=D\circ W\circ\kappa_{\enc}
  :S_O\rightsquigarrow\hat S_O.
\end{equation}
By Proposition~\ref{prop:enabling-compose} (applied twice),
\(\kappa_{\mathrm{sem}}\in\mathcal K(\mathcal I_{\mathrm{sem}})\)
where
\(\mathcal I_{\mathrm{sem}}:=
\mathcal I_{\mathrm{dec}}\circ\mathcal I_{\mathrm{ch}}\circ
\mathcal I\)
is the composite information model
\textup{(Definition~\ref{def:model-composition},
Remark~\ref{rem:composition-assoc})}.
Since \(\hat S_O\subseteq\mathbb{S}_O\), the proof system
\((\mathsf{PS},T_{\mathsf{PS}},\Cn)\) acts on~\(\hat S_O\),
making \(\Cn(\hat S_O)\), \(\Atom(\hat S_O)\), and
\(\Dd(\cdot\mid\Atom(\hat S_O))\) well-defined.
\end{definition}

\begin{definition}[End-to-end noise pair]
\label{def:e2e-noise}
Setting \(\tilde S_O:=\hat S_O\), the reconstructed space is
a noisy semantic base of~\(S_O\)
\textup{(Definition~\ref{def:noisy})} with
\(S_O^{-}:=S_O\setminus\hat S_O\) and
\(S_O^{+}:=\hat S_O\setminus S_O\).
By Proposition~\ref{prop:noise-fidelity}, the core preservation
and closure fidelity properties hold with this noise pair.
\end{definition}

\begin{definition}[Ideal semantic channel]
\label{def:ideal-channel}
A semantic channel is \emph{ideal} if all three constituent
models are ideal (Definition~\ref{def:ideal-info}) and the
kernels are the deterministic bijections induced by the
synonymy witnesses.
In this case \(\kappa_{\mathrm{sem}}\) is a deterministic
bijection and \(S_O\syn\tilde S_O\)
\textup{(Proposition~\ref{prop:syn-equiv})}.
When additionally \(\tilde S_O=S_O\) and
\(\tau_{\mathrm{e2e}}=\mathrm{id}_{S_O}\), the noise pair is
trivial.
\end{definition}

\subsection{Semantic Distortion}
\label{subsec:distortion}

\begin{definition}[Distortion function]
\label{def:distortion-function}
A \emph{distortion function} is any
\(d:S_O\times\hat S_O\to[0,\infty)\) with
\(d(s,s)=0\) for \(s\in S_O\cap\hat S_O\).
It is \emph{normalized} if \(d\le 1\).
\end{definition}

\begin{definition}[Hamming distortion]
\label{def:hamming-distortion}
\(d_H(s_o,\hat s_o):=\mathbf{1}[s_o\neq\hat s_o]\).
\end{definition}

\begin{definition}[Closure distortion]
\label{def:closure-distortion}
For a reference base \(\Gamma\subseteq\mathbb{S}_O\)
(typically \(\Gamma=S_O\)),
write \(\Gamma_{-s}:=\Gamma\setminus\{s\}\),
\(C_s:=\Cn(\Gamma_{-s}\cup\{s\})\), and
\(C_{\hat s}:=\Cn(\Gamma_{-s}\cup\{\hat s\})\).
The \emph{closure distortion} is the Jaccard distance
\begin{equation}\label{eq:d-Cn}
  d_{\Cn}(s_o,\hat s_o\mid\Gamma)
  :=1-\frac{|C_s\cap C_{\hat s}|}{|C_s\cup C_{\hat s}|},
\end{equation}
with \(0/0:=0\).
We abbreviate
\(d_{\Cn}(s_o,\hat s_o):=d_{\Cn}(s_o,\hat s_o\mid S_O)\).
\end{definition}

\begin{remark}[Zero distortion on redundant states]
\label{rem:closure-distortion-content}
Two properties follow from the \(\Cn\)-axiomatics:
\textup{(a)}~if \(s_o\in\Cn(\Gamma_{-s_o})\) (redundant), then
any \(\hat s_o\in\Cn(\Gamma_{-s_o})\) yields
\(d_{\Cn}=0\);
\textup{(b)}~if \(s_o\in\Atom(S_O)\) and \(\Gamma=S_O\),
then \(s_o\notin\Cn(\Gamma_{-s_o})\), so the replacement
genuinely matters.
\end{remark}

\begin{definition}[Depth distortion]
\label{def:depth-distortion}
Let \(A=\Atom(S_O)\) and
\(d_{\max}:=\mathsf{D_d}(\mathcal I)\)
\textup{(Definition~\ref{def:max-depth})}.
Define
\[
  d_{\Dd}(s_o,\hat s_o):=
  \begin{cases}
    \min\!\bigl(
    \frac{|\Dd(s_o\mid A)-\Dd(\hat s_o\mid A)|}
         {\max(d_{\max},1)},\,1\bigr)
    & \text{if } \hat s_o\in\Cn(A),\\
    1 & \text{otherwise}.
  \end{cases}
\]
\end{definition}

\begin{definition}[Composite semantic distortion]
\label{def:composite-distortion}
For weights \(\alpha,\beta,\gamma\ge 0\) with
\(\alpha+\beta+\gamma=1\),
\begin{equation}\label{eq:d-sem}
  d_{\mathrm{sem}}(s_o,\hat s_o)
  :=\alpha\,d_H(s_o,\hat s_o)
  +\beta\,d_{\Cn}(s_o,\hat s_o)
  +\gamma\,d_{\Dd}(s_o,\hat s_o).
\end{equation}
Setting \((\alpha,\beta,\gamma)=(1,0,0)\) recovers Hamming distortion;
\((0,1,0)\) yields a purely deductive-content measure.
\end{definition}

\begin{definition}[Expected distortion]
\label{def:expected-distortion}
For a semantic source \((S_O,P_O)\) and distortion~\(d\),
\begin{equation}\label{eq:expected-d}
  \bar d(\mathfrak C,P_O)
  :=\sum_{s_o,\hat s_o}
  P_O(s_o)\,\kappa_{\mathrm{sem}}(\hat s_o\mid s_o)\,
  d(s_o,\hat s_o).
\end{equation}
\end{definition}

\begin{definition}[Per-input expected distortion]
\label{def:per-input-distortion}
For any \(s_o\in S_O\) and distortion~\(d\),
\[
  \bar d(s_o\mid\mathfrak C)
  \;:=\;
  \sum_{\hat s_o}\kappa_{\mathrm{sem}}(\hat s_o\mid s_o)\,
  d(s_o,\hat s_o).
\]
When \(d=d_{\Cn}\), we write
\(\bar d_{\Cn}(s_o\mid\mathfrak C)\).
\end{definition}

\begin{proposition}[Noise-pair bounds on closure distortion]
\label{prop:noise-distortion-bounds}
Let \(A=\Atom(S_O)\).
If \(A\cap S_O^{-}=\varnothing\) and
\(S_O^{+}\subseteq\Cn(S_O)\), then for every
\(s_o\in S_O\setminus A\) and
\(\hat s_o\in\tilde S_O\),
\(d_{\Cn}(s_o,\hat s_o\mid S_O)=0\).
Consequently,
\(\bar d_{\Cn}(\mathfrak C,P_O)
\le P_O(A)\cdot\max_{a\in A}\bar d_{\Cn}(a\mid\mathfrak C)\).
\end{proposition}

\begin{proof}
For \(s_o\in S_O\setminus A\),
\(\Cn(S_O\setminus\{s_o\})=\Cn(S_O)\).
By Proposition~\ref{prop:noise-fidelity}(iii),
\(\Cn(\tilde S_O)=\Cn(S_O)\).
Since
\(S_O\setminus S_O^-\subseteq S_O\subseteq\Cn(S_O)\)
\textup{(by (Cn1))} and
\(S_O^+\subseteq\Cn(S_O)\) by hypothesis,
\(\tilde S_O=(S_O\setminus S_O^-)\cup S_O^+
  \subseteq\Cn(S_O)\),
so in particular
\(\hat s_o\in\Cn(S_O)\).
Since \(s_o\in J\), we have
\(\Cn(S_O\setminus\{s_o\})=\Cn(S_O)\).
Because
\(\hat s_o\in\Cn(S_O)=\Cn(S_O\setminus\{s_o\})\),
the set
\((S_O\setminus\{s_o\})\cup\{\hat s_o\}
  \subseteq\Cn(S_O\setminus\{s_o\})\),
so by monotonicity~\textup{(Cn2)} and
idempotence~\textup{(Cn3)},
\(\Cn((S_O\setminus\{s_o\})\cup\{\hat s_o\})
  \subseteq\Cn(S_O\setminus\{s_o\})
  =\Cn(S_O)\).
The reverse inclusion follows from
\(S_O\setminus\{s_o\}\subseteq
(S_O\setminus\{s_o\})\cup\{\hat s_o\}\)
and monotonicity.
Hence \(C_{s_o}=C_{\hat s_o}=\Cn(S_O)\) and
\(d_{\Cn}=0\).
Non-core states contribute zero to the expected distortion;
the bound follows.
\end{proof}

\subsection{Semantic Channel Invariants}
\label{subsec:invariants}

\begin{definition}[Semantic mutual information and capacity]
\label{def:semantic-mi}
For a semantic source \((S_O,P_O)\) and semantic
channel~\(\mathfrak C\), the \emph{semantic mutual information}
is
\(I_{\mathrm{sem}}(P_O,\mathfrak C)
:=I(\mathsf{S}_o;\hat{\mathsf{S}}_o)\)
under the joint
\(P_O(s_o)\,\kappa_{\mathrm{sem}}(\hat s_o\mid s_o)\).
The \emph{Shannon capacity} of the carrier channel is
\(C(W):=\max_{P_C}I(\mathsf{S}_c;\hat{\mathsf{S}}_c)\).
The \emph{semantic channel capacity} is
\begin{equation}\label{eq:C-sem}
  C_{\mathrm{sem}}(W)
  :=\max_{P_O,\,\kappa_{\enc}\in\mathcal K(\mathcal I),\,
  D\in\mathcal K(\mathcal I_{\mathrm{dec}})}
  I_{\mathrm{sem}}(P_O,\mathfrak C).
\end{equation}
\end{definition}

\begin{remark}[Source of ``semantic'' in \(I_{\mathrm{sem}}\)]
\label{rem:semantic-mi-naming}
The semantic mutual information \(I_{\mathrm{sem}}\) is defined as
the standard Shannon mutual information and satisfies the same
algebraic properties.
Its ``semantic'' qualifier refers not to the functional form but to
the \emph{context} in which it is evaluated: the encoding and
decoding kernels are constrained by enabling maps
\textup{(Definition~\ref{def:enabling-kernel})}, and the performance
criteria are closure-based distortion measures
\textup{(Section~\ref{subsec:distortion})} rather than symbol-level
metrics.
The semantic novelty thus enters through the feasible set and the
fidelity criterion, not through a redefinition of mutual information
itself.
\end{remark}

\begin{theorem}[Data processing bound]
\label{thm:data-processing}
\begin{enumerate}[label=\textup{(\roman*)}]
  \item \(C_{\mathrm{sem}}(W)\le C(W)\).
  \item \(C_{\mathrm{sem}}(W)\le\log|S_O|\).
  \item If the enabling maps are full
        \textup{(}\(\mathcal E(s_o)=S_C\),
        \(\mathcal E_{\mathrm{dec}}(\hat s_c)=\hat S_O\)
        for all inputs\textup{)} and the alphabet sizes satisfy
        \(|S_O|\ge|S_C|\) and \(|\hat S_O|\ge|\hat S_C|\)
        \textup{(}so that a deterministic surjection
        \(f:S_O\to S_C\) and a deterministic injection
        \(g:\hat S_C\to\hat S_O\) exist\textup{)}, then
        \(C_{\mathrm{sem}}(W)=C(W)\).
\end{enumerate}
\end{theorem}

\begin{proof}
\textup{(i)}: The Markov chain
\(\mathsf{S}_o\to\mathsf{S}_c\to\hat{\mathsf{S}}_c
\to\hat{\mathsf{S}}_o\) and the data processing
inequality~\cite{cover2006elements} give
\(I(\mathsf{S}_o;\hat{\mathsf{S}}_o)\le
I(\mathsf{S}_c;\hat{\mathsf{S}}_c)\le C(W)\).
\textup{(ii)}: \(I_{\mathrm{sem}}\le H(\mathsf{S}_o)\le\log|S_O|\).
\textup{(iii)}: Let \(P_C^*\) achieve \(C(W)\).
Choose a deterministic surjection
\(f:S_O\to S_C\), set \(P_O\) so that
\(P_C^*=f_\#P_O\), and take a deterministic injection
\(g:\hat S_C\to\hat S_O\).
Then
\(I(\mathsf{S}_o;\hat{\mathsf{S}}_o)
=I(\mathsf{S}_o;g(\hat{\mathsf{S}}_c))
=I(\mathsf{S}_o;\hat{\mathsf{S}}_c)
=I(\mathsf{S}_c;\hat{\mathsf{S}}_c)=C(W)\),
where the second equality uses the invertibility of~\(g\),
and the third uses the Markov chain
\(\mathsf{S}_o\to\mathsf{S}_c\to\hat{\mathsf{S}}_c\)
\textup{(}since the encoding is deterministic,
\(H(\hat{\mathsf{S}}_c\mid\mathsf{S}_o)
  =H(\hat{\mathsf{S}}_c\mid\mathsf{S}_c)\),
giving
\(I(\mathsf{S}_o;\hat{\mathsf{S}}_c)
  =I(\mathsf{S}_c;\hat{\mathsf{S}}_c)\)\textup{)}.
\end{proof}

\begin{remark}[Role of the size condition in
  Theorem~\textup{\ref{thm:data-processing}(iii)}]
\label{rem:size-condition-role}
The condition \(|S_O|\ge|S_C|\) ensures the existence of
a surjection \(f:S_O\to S_C\) used to push forward the
capacity-achieving input distribution; it is a sufficient
condition for the single-letter equality
\(C_{\mathrm{sem}}=C(W)\).
All block-coding results in this paper
\textup{(Theorems~\ref{thm:converse}--\ref{thm:achievability}
and Section~\ref{sec:application})} use \(C(W)\) directly
via \(S_C^n\) and do not require this condition.
\end{remark}

\begin{proposition}[Enabling-constrained capacity bound]
\label{prop:enabling-constrained-bound}
If the encoding enabling map satisfies
\(|\mathcal E(s_o)|\le k_{\mathrm{enc}}\) for all \(s_o\in S_O\),
then
\begin{equation}\label{eq:enabling-bound}
  C_{\mathrm{sem}}(W)
  \;\le\;
  \log\!\Bigl|\bigcup_{s_o\in S_O}\mathcal E(s_o)\Bigr|
  \;\le\;
  \log\bigl(|S_O|\cdot k_{\mathrm{enc}}\bigr).
\end{equation}
In particular, when the enabling map is singleton-valued
\textup{(}\(k_{\mathrm{enc}}=1\)\textup{)}, as in an ideal
channel \textup{(Definition~\ref{def:ideal-channel})},
\(C_{\mathrm{sem}}(W)\le\log|S_O|\), reproducing the source
entropy bound of
Theorem~\textup{\ref{thm:data-processing}(ii)}.
\end{proposition}

\begin{proof}
By the enabling support constraint~\eqref{eq:enabling-support},
\(\supp(P_C)\subseteq\bigcup_{s_o}\mathcal E(s_o)\), so
\(H(\mathsf{S}_c)\le\log|\bigcup_{s_o}\mathcal E(s_o)|\).
Data processing gives
\(I(\mathsf{S}_o;\hat{\mathsf{S}}_o)
  \le I(\mathsf{S}_c;\hat{\mathsf{S}}_c)
  \le H(\mathsf{S}_c)\).
The second inequality uses
\(|\bigcup_{s_o}\mathcal E(s_o)|\le|S_O|\cdot k_{\mathrm{enc}}\).
\end{proof}

\begin{definition}[Structural quality indices]
\label{def:fidelity-index}
The \emph{semantic fidelity index} and \emph{depth expansion
index} of~\(\mathfrak C\) are
\begin{align}
  \mathsf{F}(\mathfrak C)
  &:=1-\max_{s_o\in S_O}
  \sum_{\hat s_o}\kappa_{\mathrm{sem}}(\hat s_o\mid s_o)\,
  d_{\Cn}(s_o,\hat s_o\mid S_O),
  \label{eq:fidelity-index}\\
  \mathsf{E}(\mathfrak C)
  &:=\max_{s_o\in S_O}
  \sum_{\hat s_o}\kappa_{\mathrm{sem}}(\hat s_o\mid s_o)\,
  d_{\Dd}(s_o,\hat s_o).
  \label{eq:depth-expansion}
\end{align}
Both lie in~\([0,1]\).
For any \(P_O\),
\(\bar d_{\Cn}\le 1-\mathsf{F}\) and
\(\bar d_{\Dd}\le\mathsf{E}\).
\end{definition}

\begin{corollary}[Fidelity concentration on core]
\label{cor:fidelity-core-concentration}
Under the conditions of
Proposition~\textup{\ref{prop:noise-distortion-bounds}},
\(\mathsf{F}(\mathfrak C)
=1-\max_{a\in\Atom(S_O)}\bar d_{\Cn}(a\mid\mathfrak C)\):
the worst-case closure distortion is attained at a core element.
\end{corollary}

\begin{definition}[Noise-pair indices]
\label{def:noise-pair-indices-def}
Let \(A=\Atom(S_O)\) and
\(\tilde S_O^{\cap}:=S_O\cap\tilde S_O=S_O\setminus S_O^-\)
denote the preserved region.
The \emph{probabilistic core preservation index} is
\[
  \Phi_{\Atom}(\mathfrak C)
  :=
  \begin{cases}
    \displaystyle\min_{a\in A}\kappa_{\mathrm{sem}}(a\mid a)
    & \text{if } A\cap S_O^-=\varnothing,\\[4pt]
    0 & \text{otherwise}.
  \end{cases}
\]
The \emph{spurious probability index} is
\(\Psi_+(\mathfrak C)
:=\max_{s_o\in S_O}\sum_{\hat s_o\in S_O^+}
\kappa_{\mathrm{sem}}(\hat s_o\mid s_o)\).
\end{definition}

\begin{proposition}[Properties of noise-pair indices]
\label{prop:noise-pair-indices}
\begin{enumerate}[label=\textup{(\roman*)}]
  \item Both indices lie in~\([0,1]\).
  \item For an ideal channel with \(\tilde S_O=S_O\) and
        \(\tau_{\mathrm{e2e}}=\mathrm{id}\):
        \(\Phi_{\Atom}=1\), \(\Psi_+=0\).
  \item For any \(P_O\in\Delta(S_O)\),
        \begin{equation}\label{eq:dH-noise-pair-bound}
          \bar d_H(\mathfrak C,P_O)
          \;\le\;
          1-\Phi_{\Atom}(\mathfrak C)\cdot P_O(A).
        \end{equation}
\end{enumerate}
\end{proposition}

\begin{proof}
Part~\textup{(i)} is immediate from the definitions.
Part~\textup{(ii)}: under the identity kernel,
\(\kappa_{\mathrm{sem}}(a\mid a)=1\) and \(S_O^+=\varnothing\).
Part~\textup{(iii)}:
\(\bar d_H
  =1-\sum_{s_o\in\tilde S_O^{\cap}}P_O(s_o)\,
  \kappa_{\mathrm{sem}}(s_o\mid s_o)\).
When \(A\cap S_O^-\neq\varnothing\),
\(\Phi_{\Atom}=0\) and the bound is trivial.
When \(A\cap S_O^-=\varnothing\), one has
\(A\subseteq\tilde S_O^{\cap}\), so
\begin{align*}
  \sum_{s_o\in\tilde S_O^{\cap}}
  P_O(s_o)\,\kappa_{\mathrm{sem}}(s_o\mid s_o)
  \;\ge\;&
  \sum_{a\in A}P_O(a)\,\kappa_{\mathrm{sem}}(a\mid a) \\
  \;\ge\;&
  P_O(A)\cdot\Phi_{\Atom},
\end{align*}
giving
\(\bar d_H\le 1-P_O(A)\,\Phi_{\Atom}\).
\end{proof}
\begin{definition}[Receiver-side comparison indices]
\label{def:receiver-structural}
The \emph{atomicity shift} is
\(\Delta\mathsf{A}:=|\Atom(\tilde S_O)|-|\Atom(S_O)|\);
the \emph{depth shift} is
\(\Delta\mathsf{D_d}:=
\max_{q\in\tilde S_O}\Dd(q\mid\Atom(\tilde S_O))
-\mathsf{D_d}(\mathcal I)\).
Both are computable
\textup{(Remark~\ref{rem:noisy-ideal-connection})}.
\end{definition}

\begin{proposition}[Structural comparison properties]
\label{prop:structural-comparison}
\begin{enumerate}[label=\textup{(\roman*)}]
  \item Trivial noise:
        \(\Delta\mathsf{A}=\Delta\mathsf{D_d}=0\).
  \item Core-preserving noise
        \textup{(}\(A\cap S_O^-=\varnothing\),
        \(S_O^+\subseteq\Cn(S_O)\)\textup{)}:
        the set~\(A\) is an irredundant generating subset
        of~\(\tilde S_O\) for \(\Cn(\tilde S_O)=\Cn(S_O)\).
        The canonical irredundant core
        \(\Atom(\tilde S_O)\) satisfies
        \(\Cn(\Atom(\tilde S_O))=\Cn(A)\), but in general
        \(|\Atom(\tilde S_O)|\) may be larger or smaller
        than~\(|A|\), depending on the canonical order and
        the surplus elements~\(S_O^+\).
        When \(S_O^+=\varnothing\),
        \(\Atom(\tilde S_O)=A\) and hence
        \(\Delta\mathsf{A}=0\).
\end{enumerate}
\end{proposition}

\begin{proof}
\textup{(i)}: \(\tilde S_O=S_O\).
\textup{(ii)}:
By Proposition~\ref{prop:noise-fidelity}(iii),
\(\Cn(\tilde S_O)=\Cn(S_O)\) and \(A\subseteq\tilde S_O\).
Since \(A\) is irredundant
\textup{(Proposition~\ref{prop:atom-core-correct}(ii))} and
\(\Cn(A)=\Cn(S_O)=\Cn(\tilde S_O)\), the set~\(A\) is an
irredundant generating subset of~\(\tilde S_O\).
However, the canonical irredundant core
\(\Atom(\tilde S_O)\)---computed by the deletion procedure of
Definition~\ref{def:atom-so} applied to~\(\tilde S_O\) under
the fixed canonical order---may differ from~\(A\) because
surplus elements in~\(S_O^+\) can render core
elements redundant before they are scanned, while themselves
remaining irredundant.

When \(S_O^+=\varnothing\), \(\tilde S_O\subseteq S_O\)
and every element of \(\tilde S_O\setminus A\) lies in
\(J\subseteq\Cn(A)\).
We show \(\Atom(\tilde S_O)=A\).
For any \(a\in A\): at scan time the current set
\(B_{\mathrm{scan}}\subseteq\tilde S_O\subseteq S_O\), so
\(B_{\mathrm{scan}}\setminus\{a\}\subseteq S_O\setminus\{a\}\).
By monotonicity~\textup{(Cn2)},
\(\Cn(B_{\mathrm{scan}}\setminus\{a\})
  \subseteq\Cn(S_O\setminus\{a\})\).
Since \(a\in A=\Atom(S_O)\),
\(a\notin\Cn(S_O\setminus\{a\})\), hence
\(a\notin\Cn(B_{\mathrm{scan}}\setminus\{a\})\)
and \(a\) survives.
For any \(s\in\tilde S_O\setminus A\): since all
elements of~\(A\) survive
\textup{(}those before~\(s\) by the preceding argument;
those after~\(s\) not yet scanned\textup{)},
\(A\subseteq B_{\mathrm{scan}}\setminus\{s\}\),
so \(s\in\Cn(A)\subseteq\Cn(B_{\mathrm{scan}}\setminus\{s\})\)
and \(s\) is removed.
Hence \(\Atom(\tilde S_O)=A\) and \(\Delta\mathsf{A}=0\).
\end{proof}

\begin{theorem}[Semantic Fano bound]
\label{thm:semantic-fano}
Let \((S_O,P_O)\) be full-support,
\(\epsilon:=\bar d_H(\mathfrak C,P_O)\), and
\(h_b\) the binary entropy.
Then
\begin{equation}\label{eq:semantic-fano-lower}
  I_{\mathrm{sem}}(P_O,\mathfrak C)
  \ge H(\mathsf{S}_o)-h_b(\epsilon)
  -\epsilon\log(|S_O|-1).
\end{equation}
The Fano penalty $\log(|S_O|-1)$ is determined by the
\emph{source} alphabet size~$|S_O|$ and is independent of
the reconstruction alphabet~$\tilde S_O$; this is the
standard form of Fano's
inequality~\cite[Theorem~2.10.1]{cover2006elements}.
By Proposition~\textup{\ref{prop:noise-pair-indices}(iii)},
\(\epsilon\le 1-\Phi_{\Atom}\cdot P_O(A)\), so high core
preservation forces high mutual information.
For an ideal channel with \(\tilde S_O=S_O\) and
\(\tau_{\mathrm{e2e}}=\mathrm{id}\), \(\epsilon=0\) and
\(I_{\mathrm{sem}}=H(\mathsf{S}_o)\).
\end{theorem}

\begin{proof}
Apply Fano's inequality~\cite{cover2006elements} to the
source variable \(\mathsf{S}_o\in S_O\) and the
reconstruction \(\hat{\mathsf{S}}_o\in\tilde S_O\):
\(H(\mathsf{S}_o\mid\hat{\mathsf{S}}_o)\le
h_b(\epsilon)+\epsilon\log(|S_O|-1)\).
Since \(I_{\mathrm{sem}}=H(\mathsf{S}_o)
-H(\mathsf{S}_o\mid\hat{\mathsf{S}}_o)\),
\eqref{eq:semantic-fano-lower} follows.
\end{proof}

\begin{corollary}[Irredundant source, trivial noise]
\label{cor:irredundant-classical}
If \(\Atom(S_O)=S_O\), \(\tilde S_O=S_O\), and the enabling
maps are full with \(|S_O|\ge|S_C|\) and
\(|\hat S_O|\ge|\hat S_C|\), then
\(C_{\mathrm{sem}}(W)=C(W)\),
\(\mathsf{D_d}=0\),
\(d_{\mathrm{sem}}=\alpha\,d_H+\beta\,d_{\Cn}\) for every
reachable pair, and \(\Psi_+=0\).
\end{corollary}

\begin{theorem}[Invariant summary]
\label{thm:invariant-summary}
Under the standing assumptions, all invariants in families
I--VI are well-defined, finite, and computable:
\emph{I.}~Source-side: \(\mathsf{A},\mathsf{D_d}\).
\emph{II.}~Set-level: \(\rho_{\Atom},\mathsf{F}_{\Cn}\).
\emph{III.}~Noise-pair: \(\Phi_{\Atom},\Psi_+\).
\emph{IV.}~Quality: \(\mathsf{F},\mathsf{E}\).
\emph{V.}~Comparison: \(\Delta\mathsf{A},\Delta\mathsf{D_d}\).
\emph{VI.}~Info-theoretic:
\(I_{\mathrm{sem}},C_{\mathrm{sem}},C(W)\).
Key relationships:
\textup{(a)}~\(I_{\mathrm{sem}}\le C_{\mathrm{sem}}\le C(W)\);
\textup{(b)}~\(\bar d_{\Cn}\le 1-\mathsf{F}\),
\(\bar d_{\Dd}\le\mathsf{E}\);
\textup{(c)}~\(\bar d_H\le 1-\Phi_{\Atom}\,P_O(A)\);
\textup{(d)}~the Fano bound~\eqref{eq:semantic-fano-lower};
\textup{(e)}~fidelity concentrates on core under core-preserving noise;
\textup{(f)}~ideal collapse: all distortion invariants vanish and
all fidelity invariants are maximal.
\end{theorem}

\subsection{Semantic Channel Coding}
\label{subsec:coding}

For the carrier channel kernel \(W:S_C\rightsquigarrow\hat S_C\),
the \(n\)-fold memoryless extension is
\(W^{\otimes n}(\hat s_c^n\mid s_c^n)
:=\prod_{i=1}^n W(\hat s_c^{(i)}\mid s_c^{(i)})\).

\begin{definition}[Semantic block code]
\label{def:semantic-codebook}
An \emph{\((n,M)\) semantic block code} consists of a message set
\(\mathcal M\subseteq S_O\) with \(|\mathcal M|=M\), an encoding
function \(f_n:\mathcal M\to S_C^n\), and a decoding function
\(g_n:\hat S_C^n\to\hat S_O\).
The \emph{rate} is \(R:=(\log M)/n\).
\end{definition}

\begin{definition}[Reliability criteria]
\label{def:semantic-reliability}
For each \(m\in\mathcal M\), let
\(\hat{\mathsf S}_o^{(m)}:=g_n(\hat{\mathsf S}_c^n)\) be the
reconstruction when \(f_n(m)\) is sent through~\(W^{\otimes n}\).
\begin{enumerate}[label=\textup{(\roman*)}]
  \item \emph{Hamming:}
        \(P_e^{(n)}:=\max_m\Pr[\hat{\mathsf S}_o^{(m)}\neq m]\).
  \item \emph{Closure:}
        \(P_{e,\Cn}^{(n)}:=\max_m
        \Pr[d_{\Cn}(m,\hat{\mathsf S}_o^{(m)}\mid S_O)>0]\).
\end{enumerate}
\end{definition}

\begin{assumption}[Deductive independence of core elements]
\label{assump:core-disjoint}
For distinct \(a_1,a_2\in\Atom(S_O)\),
\(R_{\Cn}(a_1)\cap R_{\Cn}(a_2)=\varnothing\),
where
\(R_{\Cn}(m):=\{\hat s_o\in\hat S_O:
d_{\Cn}(m,\hat s_o\mid S_O)=0\}\).
This holds when core elements contribute disjoint deductive
increments to the closure.
\end{assumption}

\begin{remark}[General case: confusability structure and
  relaxation of core-disjointness]
\label{rem:core-disjoint-general}
Assumption~\ref{assump:core-disjoint} is the simplest sufficient
condition guaranteeing that each core element is uniquely
identifiable from its zero-distortion reconstruction set.
When the assumption fails---i.e., \(R_{\Cn}(a_1)\cap
R_{\Cn}(a_2)\neq\varnothing\) for some distinct
\(a_1,a_2\in A\)---the overlapping acceptable sets induce a
\emph{confusability graph} \(G_{\Cn}=(A,E_{\Cn})\) with
\(\{a_1,a_2\}\in E_{\Cn}\) iff
\(R_{\Cn}(a_1)\cap R_{\Cn}(a_2)\neq\varnothing\).
In the general case, zero-error source coding theory
\cite{csiszar2011information} suggests that
\(R_{\mathrm{sem}}(0;d_{\Cn})\) is determined by the
graph entropy of this confusability graph under~\(\pi_A\),
scaled by~\(P_A\).
Specifically, the zero-distortion rate is sandwiched between a
graph-entropy lower bound and the achievable upper bound
\(P_A\,H(\pi_A)\) of
Theorem~\textup{\ref{thm:tight-zero-rate}};
a precise characterization under general confusability
is left to future work.
Under Assumption~\textup{\ref{assump:core-disjoint}},
\(G_{\Cn}\) has no edges, so the graph entropy
reduces to the Shannon entropy \(H(\pi_A)\) and the two bounds
coincide, recovering~\eqref{eq:tight-R0}.
All results in this paper are stated under
Assumption~\textup{\ref{assump:core-disjoint}}, which holds
whenever core elements contribute disjoint deductive increments
to the closure---a condition verified computationally in the
instances of
Section~\textup{\ref{sec:numerical-validation}}, where the core
facts were observed to produce pairwise-disjoint closure
increments.
\end{remark}

\begin{theorem}[Converse bounds]
\label{thm:converse}
Let \(A=\Atom(S_O)\).
\begin{enumerate}[label=\textup{(\roman*)}]
  \item Any \((n,M)\) code with \(P_e^{(n)}\le\epsilon\)
        satisfies
        \(\log M\le(nC(W)+1)/(1-\epsilon)\).
  \item \textup{(}Under
        Assumption~\textup{\ref{assump:core-disjoint})}\;
        Any \((n,|S_O|)\) code with
        \(P_{e,\Cn}^{(n)}\le\epsilon\) satisfies
        \begin{equation}\label{eq:converse-closure}
        \log|A|\le(nC(W)+1)/(1-\epsilon).
        \end{equation}
\end{enumerate}
\end{theorem}

\begin{proof}
\textup{(i)}: Standard Fano argument~\cite{cover2006elements}:
for uniform~\(\mathsf{M}\),
\(\log M\le I(\mathsf{M};\hat{\mathsf M})
+H(\mathsf{M}\mid\hat{\mathsf M})
\le nC(W)+1+\epsilon\log M\).

\smallskip\noindent
\textup{(ii)}: Under Assumption~\ref{assump:core-disjoint}, the
sets \(R_{\Cn}(a):=\{\hat s_o:d_{\Cn}(a,\hat s_o\mid S_O)=0\}\)
for \(a\in A\) are pairwise disjoint.
The preimages
\(B_a:=\{y^n:g_n(y^n)\in R_{\Cn}(a)\}\) are pairwise disjoint in
\(\hat S_C^n\) with
\(W^{\otimes n}(B_a\mid f_n(a))\ge 1-\epsilon\).
The Fano argument applied to the \(|A|\)-message sub-problem
yields the bound: assign a uniform prior on~\(A\), so
\(\log|A|=H(\mathsf{A})\le
I(\mathsf{A};\hat{\mathsf{A}})+H(\mathsf{A}\mid\hat{\mathsf{A}})
\le nC(W)+1+\epsilon\log|A|\), where
\(\hat{\mathsf{A}}\) is the core estimate and
\(\Pr[\hat{\mathsf{A}}\neq\mathsf{A}]\le\epsilon\)
follows from the disjointness of the decoding regions.
Rearranging gives~\eqref{eq:converse-closure}.
\end{proof}

\begin{theorem}[Achievability]
\label{thm:achievability}
Let \(C(W)>0\) and \(A=\Atom(S_O)\).
\begin{enumerate}[label=\textup{(\roman*)}]
  \item \emph{Hamming:}\;
        The full knowledge base \(S_O\) can be communicated
        Hamming-reliably using
        \(n>\log|S_O|/C(W)\) channel uses.
  \item \emph{Closure:}\;
        If \(\hat S_O\supseteq S_O\) and the encoding enabling
        is full, then \(P_{e,\Cn}^{(n)}\to 0\) provided
        \(\log|A|/n<C(W)\).
\end{enumerate}
\end{theorem}

\begin{proof}
\textup{(i)}:\
By Shannon's channel coding
theorem~\cite{shannon1948mathematical,cover2006elements},
for any rate $R<C(W)$ there exists a sequence of
$(n,\lceil 2^{nR}\rceil)$ codes with maximal error
probability $P_e^{(n)}\to 0$.
Setting $M=|S_O|$ and $R=\log|S_O|/n$, the hypothesis
$n>\log|S_O|/C(W)$ ensures $R<C(W)$, giving
Hamming-reliable communication of~$S_O$.

\smallskip\noindent
\textup{(ii)}:\
The code is constructed in two layers.

\emph{Layer~1 \textup{(}core code\textup{)}.}\;
Since $\log|A|/n<C(W)$, the channel coding theorem yields
a sequence of $(n,|A|)$ codes $(f_n^A,g_n^A)$ with
$f_n^A:A\to S_C^n$ and $g_n^A:\hat S_C^n\to A$ satisfying
$P_e^{(n)}(A):=\max_{a\in A}
\Pr[g_n^A(\hat{\mathsf S}_c^n)\neq a\mid
f_n^A(a)\text{ sent}]\to 0$.
Since $A\subseteq S_O\subseteq\hat S_O$ by hypothesis,
every output of~$g_n^A$ is a valid element
of~$\hat S_O$.

\emph{Layer~2 \textup{(}redundant extension\textup{)}.}\;
Fix $a_0\in A$.
Define the full encoder $f_n:S_O\to S_C^n$ by
$f_n(m):=f_n^A(m)$ for $m\in A$ and
$f_n(m):=f_n^A(a_0)$ for $m\in J$.
The decoder is $g_n:=g_n^A:\hat S_C^n\to A$.

\emph{Closure-reliability analysis.}\;
For $m=a\in A$: since $d_{\Cn}(a,a\mid S_O)=0$,
\[
  \Pr\bigl[d_{\Cn}\bigl(a,g_n(\hat{\mathsf S}_c^n)
  \mid S_O\bigr)>0\bigr]
  \;\le\;
  \Pr\bigl[g_n(\hat{\mathsf S}_c^n)\neq a\bigr]
  \;\le\;
  P_e^{(n)}(A).
\]

For $m=j\in J$: the decoder outputs some
$\hat a\in A$, determined by the channel realization.
Regardless of which $\hat a$ is output,
$\hat a\in A\subseteq S_O\subseteq\Cn(S_O)$.
Since $j\in J$, we have
$j\in\Cn(S_O\setminus\{j\})$, so
$\Cn(S_O\setminus\{j\})=\Cn(S_O)$
\textup{(}by~\textup{(Cn1)},
\textup{(Cn2)}, and~\textup{(Cn3))}.
Because
$\hat a\in\Cn(S_O)=\Cn(S_O\setminus\{j\})$,
the set
$(S_O\setminus\{j\})\cup\{\hat a\}
\subseteq\Cn(S_O\setminus\{j\})$, giving
\[
  \Cn\bigl((S_O\setminus\{j\})\cup\{\hat a\}\bigr)
  \;=\;\Cn(S_O\setminus\{j\})
  \;=\;\Cn(S_O)
\]
by~\textup{(Cn2)} and~\textup{(Cn3)}.
The reverse inclusion follows from
$S_O\setminus\{j\}\subseteq
(S_O\setminus\{j\})\cup\{\hat a\}$
and monotonicity.
Hence $d_{\Cn}(j,\hat a\mid S_O)=0$
\emph{deterministically} for every channel output.

Combining:
$P_{e,\Cn}^{(n)}
=\max_{m\in S_O}
\Pr[d_{\Cn}(m,g_n(\hat{\mathsf S}_c^n)\mid S_O)>0]
\le P_e^{(n)}(A)\to 0$.
\end{proof}

\begin{corollary}[Minimum blocklength]
\label{cor:min-blocklength}
Under Assumption~\textup{\ref{assump:core-disjoint}}, the minimum
blocklength for closure-reliable communication satisfies
\(n^*\approx\log|\Atom(S_O)|/C(W)\), while Hamming reliability
requires \(n_H^*\approx\log|S_O|/C(W)\).
The \emph{deductive compression ratio} is
\(\log|\Atom(S_O)|/\log|S_O|<1\) whenever
\(|J|>0\).
\end{corollary}

\begin{definition}[Semantic rate--distortion function]
\label{def:semantic-rd}
For a semantic source \((S_O,P_O)\) and bounded distortion~\(d\),
\begin{equation}\label{eq:R-sem-D}
  R_{\mathrm{sem}}(D)
  :=\min_{\substack{P_{\hat S\mid S}:\,
  \E[d(S_o,\hat S_o)]\le D}}
  I(\mathsf{S}_o;\hat{\mathsf{S}}_o).
\end{equation}
When \(d=d_{\Cn}(\cdot,\cdot\mid S_O)\), the reference
knowledge base~\(S_O\) is a fixed model parameter: each
source symbol~\(S_o^{(t)}\) takes values in~\(S_O\), and
every evaluation
\(d_{\Cn}(S_o^{(t)},\hat S_o^{(t)}\mid S_O)\) uses the
same~\(S_O\), so \(d_{\Cn}\) is a standard single-letter
distortion parameterized by~\(S_O\).
For i.i.d.\ source blocks of length~$m$, the block distortion
is $d^{(m)}(s^m,\hat s^m):=m^{-1}\sum_{t=1}^{m}d(s_t,\hat s_t)$,
and the rate--distortion coding
theorem~\cite[Theorem~10.2.1]{cover2006elements} guarantees
achievability at any rate exceeding $R_{\mathrm{sem}}(D)$.
\end{definition}

\begin{proposition}[Properties of \(R_{\mathrm{sem}}\)]
\label{prop:rd-properties}
\textup{(i)}~The minimum exists.
\textup{(ii)}~\(R_{\mathrm{sem}}\) is non-increasing and convex.
\textup{(iii)}~Under Hamming distortion,
\(R_{\mathrm{sem}}(0)=H(P_O)\).
\textup{(iv)}~Under closure distortion with
\(\hat S_O\supseteq\Cn(S_O)\cap\mathbb{S}_O\),
\(R_{\mathrm{sem}}(0;d_{\Cn})\le\log|\Atom(S_O)|\).
\textup{(v)}~\(R_{\mathrm{sem}}\) is computable.
\end{proposition}

\begin{proof}
Parts~\textup{(i)}--\textup{(iii)} are standard~\cite{cover2006elements}.
For~\textup{(iv)}, fix \(a_0\in A\) and define
\(\phi(s_o):=s_o\) if \(s_o\in A\), \(\phi(s_o):=a_0\) if
\(s_o\in J\).
The deterministic conditional
\(P_{\hat S\mid S}(\cdot\mid s_o)=\delta_{\phi(s_o)}\) achieves
\(\E[d_{\Cn}]=0\) (redundant states yield zero distortion by
Remark~\ref{rem:closure-distortion-content}(a)) and
\(I(\mathsf{S}_o;\hat{\mathsf{S}}_o)=H(\hat{\mathsf{S}}_o)
\le\log|A|\).
\end{proof}

\begin{remark}[Computation and source--channel separation]
\label{rem:rd-computation-separation}
The semantic rate--distortion function
\(R_{\mathrm{sem}}(D;d_{\Cn})\) can be computed to arbitrary
precision via a variant of the Blahut--Arimoto
algorithm~\cite{cover2006elements}: the standard alternating
minimization applies because the feasible set is a compact
subset of a finite-dimensional probability simplex and the
objective \(I(\mathsf{S}_o;\hat{\mathsf{S}}_o)\) is convex
in the conditional distribution for fixed~\(P_O\).
Han et~al.~\cite{han2025extended} have recently developed an
extended Blahut--Arimoto algorithm for computing semantic
rate--distortion functions under synonymous mappings.
Their alternating-minimization framework is conceptually related
to ours and can potentially be adapted to the closure-distortion
setting; a formal verification that the feasibility structure
of~\(d_{\Cn}\) satisfies the regularity conditions of their
convergence proof is an interesting algorithmic direction.
The classical source--channel separation
theorem~\cite{shannon1948mathematical,cover2006elements}
carries over verbatim to the semantic framework: semantic
source coding at rate \(R_{\mathrm{sem}}(D)\) followed by
channel coding at rate \(C(W)\) is optimal.
The semantic novelty is that
\(R_{\mathrm{sem}}(D;d_{\Cn})<R(D;d_H)\) for
\(D\) near zero whenever \(|J|>0\), yielding a strict
improvement in the achievable distortion--capacity operating
point.
\end{remark}

\section{Fundamental Limits of Semantic Communication}
\label{sec:fundamental-limits}

This section derives the central information-theoretic results.
The key ingredient is the zero-distortion property of redundant
states \textup{(Remark~\ref{rem:closure-distortion-content}(a))}:
under closure distortion only the core~\(\Atom(S_O)\) contributes
to rate and distortion, enabling a decomposition with no classical
counterpart.

\smallskip
\noindent\textit{Section roadmap.}\;
Sections~\ref{subsec:tight-rd}--\ref{subsec:rd-decomposition}
derive the tight zero-distortion semantic rate and its
extension to all distortion levels, treating the receiver's
inference engine as unconstrained.
Section~\ref{subsec:sem-separation} establishes the semantic
source--channel separation theorem and quantifies the semantic
leverage factor.
Section~\ref{subsec:strengthened-fano} proves the strengthened
semantic Fano inequality.
Section~\ref{subsec:rate-delay} introduces a \emph{new degree of
freedom}---the receiver's derivation budget~$\delta$, measured in
$T_{\mathsf{PS}}$-iterations---and replaces the closure
distortion~$d_{\Cn}$ by a budget-constrained variant
$d_{\Cn}^\delta$
\textup{(Definition~\ref{def:delta-closure-distortion})}, which
requires the receiver to reconstruct every stored state within
$\delta$~iterations rather than merely preserving the abstract
closure.
This leads to a $\delta$-irredundant core filtration
$A_\delta$ and a rate--delay--distortion surface that smoothly
interpolates between the classical rate $H(P_O)$ at $\delta=0$
and the semantic rate $P_A\,H(\pi_A)$ at
$\delta=\mathsf{D_d}$.

Throughout, \(S_O\) is a finite knowledge base with irredundant core
\(A:=\Atom(S_O)\), \(k:=|A|\), and stored shortcuts
\(J:=S_O\setminus A\) with \(|J|=|S_O|-k\).
We write \(P_A:=P_O(A)\), \(P_J:=1-P_A\), and
\(\pi_A(a):=P_O(a)/P_A\) for the conditional core distribution.
The carrier channel is \(W:S_C\rightsquigarrow\hat S_C\) with
\(C(W)>0\).
Assumption~\ref{assump:core-disjoint} is in force for converse
bounds.
In block-coding results, \(n\) denotes the blocklength
\textup{(}number of channel uses\textup{)};
in rate expressions for the uniform source, we write
\(|S_O|\) explicitly to avoid ambiguity.

\subsection{Tight Semantic Rate--Distortion Function}
\label{subsec:tight-rd}

\begin{theorem}[Tight zero-distortion semantic rate]
\label{thm:tight-zero-rate}
Let \((S_O,P_O)\) be a semantic source with irredundant core
\(A=\Atom(S_O)\).
Under closure distortion \(d_{\Cn}(\cdot,\cdot\mid S_O)\) with
\(\hat S_O\supseteq\Cn(S_O)\cap\mathbb S_O\) and
Assumption~\textup{\ref{assump:core-disjoint}}:
\begin{equation}\label{eq:tight-R0}
  R_{\mathrm{sem}}(0;\,d_{\Cn},\,P_O)
  \;=\;
  P_A\,H(\pi_A).
\end{equation}
Under the uniform source:
\(R_{\mathrm{sem}}(0)=\frac{k}{|S_O|}\log k\).
The semantic compression gain relative to Hamming fidelity is
\begin{equation}\label{eq:zero-rate-ratio}
  \frac{R_{\mathrm{sem}}(0;\,d_{\Cn})}{R(0;\,d_H)}
  \;=\;
  \frac{P_A\,H(\pi_A)}{H(P_O)}
  \;<\;1
  \quad\text{whenever } k<|S_O|.
\end{equation}
\end{theorem}

\begin{proof}
We establish the lower and upper bounds separately.

\smallskip\noindent
\emph{Lower bound.}\;
Let \(P_{\hat S\mid S}\) be any conditional distribution achieving
\(\E[d_{\Cn}(S_o,\hat S_o\mid S_O)]=0\).
Define the auxiliary random variable
\(T:\,S_O\to A\cup\{*\}\) by
\(T(s_o):=s_o\) if \(s_o\in A\) and
\(T(s_o):=*\) if \(s_o\in J\).
Since \(T\) is a deterministic function of~\(S_o\), the data
processing inequality gives
\begin{equation}\label{eq:DPI-T}
  I(S_o;\hat S_o)\;\ge\;I(T;\hat S_o).
\end{equation}
We now lower-bound the right-hand side.

Under Assumption~\ref{assump:core-disjoint}, the zero-distortion
constraint forces
\(\supp(P_{\hat S\mid S}(\cdot\mid a))\subseteq R_{\Cn}(a)\)
for each \(a\in A\), where
\(R_{\Cn}(a):=\{\hat s_o\in\hat S_O:
d_{\Cn}(a,\hat s_o\mid S_O)=0\}\)
\textup{(Assumption~\ref{assump:core-disjoint})},
with the sets \(\{R_{\Cn}(a)\}_{a\in A}\) pairwise disjoint.
For redundant states \(j\in J\),
\(P_{\hat S\mid S}(\cdot\mid j)\) is unconstrained on
\(\Cn(S_O)\cap\hat S_O\)
(Remark~\ref{rem:closure-distortion-content}(a)).

Let \(Q_j:=P_{\hat S\mid S}(\cdot\mid j)\) for each \(j\in J\),
and define the average redundant output distribution
\(Q:=\sum_{j\in J}[P_O(j)/P_J]\,Q_j\).
Write
\(\bar\pi:=\sum_{a\in A}\pi_A(a)\,P_{\hat S\mid S}(\cdot\mid a)\)
for the average core output distribution.
Since the supports \(R_{\Cn}(a)\) are pairwise disjoint,
the entropy of the mixture~\(\bar\pi\) decomposes by the
standard disjoint-support identity
\cite[Theorem~2.7.3]{cover2006elements}:
\begin{equation}\label{eq:disjoint-split}
  H(\bar\pi)
  \;=\;
  H(\pi_A)+\sum_{a\in A}\pi_A(a)\,
  H\bigl(P_{\hat S\mid S}(\cdot\mid a)\bigr).
\end{equation}
The conditional entropy of~\(\hat S_o\) given~\(T\) is
\[
  H(\hat S_o\mid T)
  = P_A\sum_{a}\pi_A(a)\,H(P_{\hat S\mid S}(\cdot\mid a))
    + P_J\,H(Q),
\]
and the output marginal is
\(P_{\hat S}=P_A\,\bar\pi+P_J\,Q\).
By the concavity of Shannon
entropy~\cite[Theorem~2.7.3]{cover2006elements},
\begin{equation}\label{eq:entropy-concavity}
  H(P_A\,\bar\pi+P_J\,Q)
  \;\ge\;
  P_A\,H(\bar\pi)+P_J\,H(Q).
\end{equation}
Therefore, using~\eqref{eq:disjoint-split}:
\begin{align}
  I(T;\hat S_o)
  &= H(P_A\,\bar\pi+P_J\,Q)-H(\hat S_o\mid T) \nonumber\\
  &\ge P_A\,H(\bar\pi)+P_J\,H(Q)
       -P_A\!\sum_{a}\pi_A(a)\,H(P_{\hat S\mid S}(\cdot\mid a))
       -P_J\,H(Q) \nonumber\\
  &= P_A\bigl[H(\bar\pi)
     -\!\sum_{a}\pi_A(a)\,H(P_{\hat S\mid S}(\cdot\mid a))\bigr]
     \nonumber\\
  &= P_A\,H(\pi_A).
  \label{eq:IT-lower}
\end{align}
Since~\eqref{eq:IT-lower} holds for every feasible
\(P_{\hat S\mid S}\) (the distribution~\(Q\) is determined by the
conditional, not chosen by us), combining
with~\eqref{eq:DPI-T} gives
\(R_{\mathrm{sem}}(0)\ge P_A\,H(\pi_A)\).

\smallskip\noindent
\emph{Upper bound (achievability).}\;
We exhibit a feasible conditional achieving
\(I(S_o;\hat S_o)=P_A\,H(\pi_A)\) at zero distortion.
Define
\[
  P_{\hat S\mid S}(\hat s_o\mid s_o)
  :=
  \begin{cases}
    \mathbf{1}[\hat s_o=s_o] & \text{if } s_o\in A,\\
    \pi_A(\hat s_o) & \text{if } s_o\in J,
  \end{cases}
\]
where \(\pi_A\) is supported on~\(A\).

\emph{Distortion check.}\;
For \(s_o=a\in A\): \(\hat s_o=a\), giving
\(d_{\Cn}(a,a\mid S_O)=0\).
For \(s_o=j\in J\): \(\hat s_o\) is drawn from \(\pi_A\) on~\(A\),
so \(\hat s_o\in A\subseteq\Cn(S_O)\); since \(j\) is redundant,
\(d_{\Cn}(j,\hat s_o\mid S_O)=0\)
(Remark~\ref{rem:closure-distortion-content}(a)).

\emph{Rate computation.}\;
Under this conditional, \(Q=\pi_A\) and the output marginal is
\(P_{\hat S}=P_A\pi_A+P_J\pi_A=\pi_A\).
All redundant inputs produce the same conditional output
distribution~\(\pi_A\), so
\(\hat S_o\perp\!\!\!\perp S_o\mid T\), giving
\(I(S_o;\hat S_o\mid T)=0\).
Hence
\[
  I(S_o;\hat S_o)
  = I(T;\hat S_o)
  = H(\pi_A)-P_J\,H(\pi_A)
  = P_A\,H(\pi_A). \qedhere
\]
\end{proof}

\begin{remark}[Relaxed acceptable sets]
\label{rem:relaxed-acceptable}
The lower bound \(P_A\,H(\pi_A)\) in the proof of
Theorem~\textup{\ref{thm:tight-zero-rate}} holds for
arbitrary acceptable-set sizes \(|R_{\Cn}(a)|\ge 1\),
provided the pairwise-disjoint support condition of
Assumption~\textup{\ref{assump:core-disjoint}} holds;
the key step is the entropy splitting
identity~\eqref{eq:disjoint-split}.
The achievability construction of the upper bound applies
unchanged.
\end{remark}

\begin{remark}[Role of the reconstruction alphabet assumption]
\label{rem:recon-alphabet}
Theorem~\textup{\ref{thm:tight-zero-rate}} assumes
\(\hat S_O\supseteq\Cn(S_O)\cap\mathbb{S}_O\), ensuring that
the decoder can output any state in the sender's closure.
This is an idealized ``full reconstruction alphabet'' setting
that maximizes the deductive compression gain; it corresponds
to a receiver whose vocabulary contains at least all closure
elements expressible in the ambient universe.
When the reconstruction alphabet is restricted to
\(\hat S_O\subsetneq\Cn(S_O)\cap\mathbb{S}_O\)---as in the
heterogeneous setting of
Section~\textup{\ref{sec:application}}, where
\(\hat S_O=S_O^{(j)}\) may exclude some closure
elements---the zero-distortion rate can only increase:
\(R_{\mathrm{sem}}(0;\hat S_O)\ge R_{\mathrm{sem}}(0;
\Cn(S_O)\cap\mathbb{S}_O)=P_A H(\pi_A)\).
The heterogeneous achievability results of
Theorem~\textup{\ref{thm:heterogeneous-closure}} show that
the rate \(P_A H(\pi_A)\) remains achievable as long as the
core coverage condition \(A^{(i)}\subseteq\hat S_O\) holds;
when it fails,
Corollary~\textup{\ref{cor:heterogeneous-impossibility}}
provides the impossibility characterization.
\end{remark}

\begin{remark}[Necessity of the proof-system structure]
\label{rem:why-proof-system}
The result~\eqref{eq:tight-R0} depends on three concepts absent
from classical rate--distortion theory:
the irredundant core \(\Atom(S_O)\), the deductive
closure~\(\Cn\), and the core/redundant partition of source
symbols.
Replacing \(d_{\Cn}\) by any distortion that assigns positive cost
to all symbol errors eliminates the zero-distortion property of
redundant states and recovers \(R(0)=H(P_O)\).
\end{remark}

\subsection{Rate--Distortion Core Decomposition}
\label{subsec:rd-decomposition}

Theorem~\ref{thm:tight-zero-rate} extends to all distortion
levels: the full semantic rate--distortion function decomposes
into a contribution from the core alone.

\begin{theorem}[Core decomposition of \(R_{\mathrm{sem}}\)]
\label{thm:rd-decomposition}
Under the hypotheses of
Theorem~\textup{\ref{thm:tight-zero-rate}},
\begin{equation}\label{eq:rd-decomp}
  R_{\mathrm{sem}}(D;\,d_{\Cn},\,P_O)
  \;=\;
  P_A\cdot R^{(A)}\!\Bigl(\frac{D}{P_A};\,d_{\Cn},\,\pi_A\Bigr),
  \qquad D\ge 0,
\end{equation}
where \(R^{(A)}(D';\,d_{\Cn},\,\pi_A)\) is the rate--distortion
function of the core sub-source \((A,\pi_A)\) with distortion
\(d_{\Cn}(\cdot,\cdot\mid S_O)\) restricted to core inputs.

\smallskip\noindent
Consequences:
\textup{(i)}~At \(D=0\),
\(R_{\mathrm{sem}}(0)=P_A\,H(\pi_A)\), recovering
Theorem~\textup{\ref{thm:tight-zero-rate}}.
\textup{(ii)}~\(R_{\mathrm{sem}}\) is convex, non-increasing, and
computable.
\textup{(iii)}~Redundant states are invisible: the
function~\eqref{eq:rd-decomp} depends on~\(S_O\) only through
\(A\), \(P_A\), \(\pi_A\), and the closure-distortion structure
on~\(A\).
\end{theorem}

\begin{proof}
\emph{Distortion decomposition.}\;
By Remark~\ref{rem:closure-distortion-content}(a),
\(d_{\Cn}(j,\hat s_o\mid S_O)=0\) for all \(j\in J\) and
\(\hat s_o\in\Cn(S_O)\).
The expected distortion therefore decomposes as
\[
  \E[d_{\Cn}]
  = \sum_{a\in A}P_O(a)\,\E[d_{\Cn}(a,\hat S_o\mid S_O)\mid S_o=a]
  = P_A\,\E_{\pi_A}[\bar d_A],
\]
where \(\bar d_A(a):=\E[d_{\Cn}(a,\hat S_o\mid S_O)\mid S_o=a]\).
The constraint \(\E[d_{\Cn}]\le D\) is thus equivalent to
\(\E_{\pi_A}[\bar d_A]\le D/P_A\).

\smallskip\noindent
\emph{Rate decomposition.}\;
Using the auxiliary variable \(T\) from the proof of
Theorem~\ref{thm:tight-zero-rate}:
\[
  I(S_o;\hat S_o)=I(T;\hat S_o)+I(S_o;\hat S_o\mid T)
  \ge I(T;\hat S_o).
\]
The residual \(I(S_o;\hat S_o\mid T)\) is minimized to zero by
choosing identical conditionals for all \(j\in J\) (as in the
achievability of Theorem~\ref{thm:tight-zero-rate}).
Hence the minimum of \(I(S_o;\hat S_o)\) subject to
\(\E_{\pi_A}[\bar d_A]\le D/P_A\) equals the minimum of
\(I(T;\hat S_o)\) subject to the same constraint.

The distribution of \(T\) concentrates probability \(P_A\) on
\(A\) (with conditional distribution~\(\pi_A\)) and
probability \(P_J\) on the singleton~\(\{*\}\).
The problem
\(\min_{P_{\hat S\mid T}}\{I(T;\hat S_o):
\E_{\pi_A}[\bar d_A]\le D/P_A\}\)
is a rate--distortion problem whose source is the mixture
of~\(\pi_A\) (with weight \(P_A\)) and a point mass (with
weight \(P_J\)).
The point-mass component contributes zero mutual information
and zero distortion, so the minimum equals
\(P_A\cdot R^{(A)}(D/P_A;\,d_{\Cn},\,\pi_A)\)
by the standard rate--distortion scaling for mixtures with a
``free'' component~\cite[Problem~10.8]{cover2006elements}.
\end{proof}

\subsection{Semantic Source--Channel Separation and Semantic Leverage}
\label{subsec:sem-separation}

\begin{theorem}[Semantic source--channel separation]
\label{thm:sem-source-channel}
Let \(W:S_C\rightsquigarrow\hat S_C\) with \(C(W)>0\), and let
\((S_O,P_O)\) be a semantic source with core \(A=\Atom(S_O)\).

\medskip\noindent
\textbf{Part~A \textup{(}Single-shot / message-set regime\textup{)}.}\;
Consider an \((n,|S_O|)\) semantic block code
\textup{(Definition~\ref{def:semantic-codebook})} with message
set \(\mathcal M=S_O\).
\begin{enumerate}[label=\textup{(A\arabic*)}]
  \item \emph{Converse.}\;
        Closure-reliable communication
        (\(P_{e,\Cn}^{(n)}\to 0\)) requires blocklength
        \(n\ge(1-o(1))\log|A|/C(W)\)
        \textup{(Theorem~\ref{thm:converse}(ii))}.
  \item \emph{Achievability.}\;
        The two-layer code
        \textup{(Theorem~\ref{thm:achievability}(ii))} achieves
        \(P_{e,\Cn}^{(n)}\to 0\) whenever
        \(\log|A|/n<C(W)\).
  \item \emph{Single-shot leverage.}\;
        The minimum blocklength for closure reliability is
        \(n^*_{\Cn}\approx\log|A|/C(W)\)
        \textup{(Corollary~\ref{cor:min-blocklength})}.
        Define the \emph{single-shot semantic leverage factor}
        \begin{equation}\label{eq:leverage-single}
          \Lambda_1(S_O)
          \;:=\;
          \frac{\log|S_O|}{\log|\Atom(S_O)|}
          \;>\;1
          \quad\text{whenever } k<|S_O|.
        \end{equation}
        Since identifying one of \(|S_O|\) knowledge-base states
        requires \(\log|S_O|\) bits, closure reliability conveys
        this identification in
        \(n^*_{\Cn}\approx\log|A|/C(W)\) channel uses,
        yielding an \emph{effective task-identification rate}
        \begin{equation}\label{eq:effective-rate}
          R_{\mathrm{eff}}
          \;:=\;
          \frac{\log|S_O|}{n^*_{\Cn}}
          \;\approx\;
          C(W)\cdot\Lambda_1(S_O).
        \end{equation}
        The quantity \(R_{\mathrm{eff}}\) is a
        \emph{task-throughput index} that counts
        knowledge-base index bits conveyed per channel
        use under closure fidelity; it does \emph{not}
        represent a Shannon mutual-information rate,
        and its exceeding \(C(W)\) is
        consistent with the data processing bound
        \(C_{\mathrm{sem}}\le C(W)\)
        \textup{(Theorem~\ref{thm:data-processing}(i))}, because
        the closure-based fidelity criterion treats \(|J|\)
        redundant states as free.
\end{enumerate}

\medskip\noindent
\textbf{Part~B \textup{(}Asymptotic i.i.d.\ source-coding
regime\textup{)}.}\;
Consider a length-$m$ i.i.d.\ source block from
$(S_O,P_O)$ mapped to $n$~channel uses of~$W^{\otimes n}$,
with $m\to\infty$.
The fidelity criterion is zero expected closure distortion
($D=0$ in~\eqref{eq:R-sem-D}).
We work under the same hypotheses as
Theorem~\textup{\ref{thm:tight-zero-rate}}:
$\hat S_O\supseteq\Cn(S_O)\cap\mathbb{S}_O$ and
Assumption~\textup{\ref{assump:core-disjoint}}.
The relevant source rate is therefore
$R_{\mathrm{sem}}(0;d_{\Cn})=P_A\,H(\pi_A)$.
\begin{enumerate}[label=\textup{(B\arabic*)}]
  \item \emph{Necessary condition.}\;
        Achieving zero expected closure distortion requires
        \begin{equation}\label{eq:necessary-sep}
          P_A\,H(\pi_A)\;\le\;C(W).
        \end{equation}
  \item \emph{Sufficient condition.}\;
        If the strict inequality holds
        in~\eqref{eq:necessary-sep}, then zero expected
        closure distortion is achievable: by the classical
        rate--distortion coding
        theorem~\cite[Theorem~10.2.1]{cover2006elements},
        for every $R>R_{\mathrm{sem}}(0)=P_A\,H(\pi_A)$
        there exist length-$m$ block source codes achieving
        $\E[d_{\Cn}]=0$ at rate~$R$;
        cascading with a capacity-achieving channel code
        completes the separation argument.
  \item \emph{Asymptotic leverage.}\;
        Encoding $m$ source symbols at distortion $D{=}0$
        requires at least $m\,P_A\,H(\pi_A)$ bits, and
        each channel use carries at most $C(W)$ bits, so
        $\liminf_{m\to\infty}n/m
        \ge P_A\,H(\pi_A)/C(W)$, yielding the
        \emph{asymptotic semantic leverage factor}
        \begin{equation}\label{eq:leverage-asymp}
          \Lambda_\infty(S_O,P_O)
          \;:=\;
          \frac{\log|S_O|}{P_A\,H(\pi_A)}
          \;\ge\;
          \Lambda_1(S_O),
        \end{equation}
        with equality when \(\pi_A\) is uniform on~\(A\) and
        \(P_A=1\).
        Like~\(\Lambda_1\), the factor~\(\Lambda_\infty\)
        uses \emph{task-identification normalization}: the
        numerator~\(\log|S_O|\) counts the index bits needed
        to specify one knowledge-base state, not the source
        entropy~\(H(P_O)\).
        The entropy-based compression gain is the distinct
        ratio \(P_A\,H(\pi_A)/H(P_O)<1\)
        of~\eqref{eq:zero-rate-ratio}.
        Under the uniform source:
        \(\Lambda_1=\log |S_O|/\log k\) and
        \(\Lambda_\infty=|S_O|\log |S_O|/(k\log k)\).
\end{enumerate}
\end{theorem}

\begin{proof}
\emph{Part~A.}\;
(A1) is Theorem~\ref{thm:converse}(ii).
(A2) is Theorem~\ref{thm:achievability}(ii).
(A3) follows by dividing \(\log|S_O|\) by
\(n^*_{\Cn}\approx\log|A|/C(W)\).

\smallskip\noindent
\emph{Part~B.}\;
(B1): the closure distortion $d_{\Cn}\ge 0$ is a bounded
single-letter distortion on finite alphabets, so the
classical source--channel separation theorem
\cite{shannon1948mathematical,cover2006elements} applies:
achieving expected distortion $D=0$ over the i.i.d.\
source requires source rate
$R_{\mathrm{sem}}(0;d_{\Cn})=P_A\,H(\pi_A)
\le C(W)$, giving~\eqref{eq:necessary-sep}.
(B2): by the rate--distortion coding theorem for finite
alphabets~\cite[Theorem~10.2.1]{cover2006elements}, for
every $\epsilon>0$ there exists a block source code of
rate $P_A\,H(\pi_A)+\epsilon$ achieving
$\E[d_{\Cn}]=0$; a capacity-achieving channel code
transmits the compressed index reliably whenever
$P_A\,H(\pi_A)+\epsilon<C(W)$.
(B3): by~(B1), every achievable pair $(m,n)$ satisfies
$m\,P_A\,H(\pi_A)\le nC(W)$, hence
$\liminf_{m\to\infty}n/m\ge P_A\,H(\pi_A)/C(W)$.
The leverage ordering
$\Lambda_\infty\ge\Lambda_1\ge 1$ follows from
$P_A\,H(\pi_A)\le\log|A|\le\log|S_O|$.
\end{proof}

\begin{remark}[Comparison with the synonymous-mapping
  semantic leverage]
\label{rem:comparison-niu-zhang}
The leverage factors
\(\Lambda_1\) and~\(\Lambda_\infty\) are the closure-fidelity
counterparts of the semantic throughput gain \(C_s\ge C\) of
Niu and Zhang~\cite{niu2024mathematical}.
The two mechanisms are complementary:
\cite{niu2024mathematical} achieves its gain through
\emph{source-side} equivalence-class collapsing;
our framework achieves it through \emph{receiver-side} deductive
reconstruction.
Neither mechanism violates the Shannon channel capacity
\(C(W)\); both exploit task-specific fidelity criteria to reduce
the effective source rate below~\(H(P_O)\).
A unified theory combining both mechanisms---synonymous collapsing
of the irredundant core followed by deductive expansion at the
receiver---would compound the two gains.
\end{remark}

\subsection{Strengthened Semantic Fano Inequality}
\label{subsec:strengthened-fano}

The classical Fano inequality bounds \(H(X\mid\hat X)\) in terms
of \(\Pr[X\neq\hat X]\) and the full alphabet size~\(|X|\).
The core/redundant decomposition enables a tighter bound whose
penalty term involves \(\log|A|\) rather than \(\log|S_O|\).

\begin{theorem}[Semantic Fano inequality]
\label{thm:semantic-fano-tight}
Let \((S_O,P_O)\) be a full-support semantic source,
\(\mathfrak C\) a semantic channel, and
\(\epsilon_A:=\Pr[\hat S_o\neq S_o,\;S_o\in A]\)
the core error probability.
Then
\begin{equation}\label{eq:sem-fano-tight}
  I(S_o;\hat S_o)
  \;\ge\;
  P_A\,H(\pi_A)
  -h_b(\epsilon_A)
  -\epsilon_A\log(|A|-1),
\end{equation}
where \(h_b\) is the binary entropy.
When \(|J|=0\), \(P_A=1\), \(\pi_A=P_O\), and
\eqref{eq:sem-fano-tight} reduces to the classical Fano bound
with alphabet~\(|A|=|S_O|\).
When \(|J|>0\), the penalty involves \(\log|A|\) rather than
\(\log|S_O|\); additionally, the reference level is the
semantic rate \(P_A\,H(\pi_A)\) rather than the full source
entropy \(H(P_O)\), which is a tightening by
\(H(P_O)-P_A\,H(\pi_A)>0\) bits.
\end{theorem}

\begin{proof}
Define \(B:=\mathbf{1}[S_o\in A]\).
Since \(B\) is a deterministic function of~\(S_o\),
\(H(B\mid S_o)=0\) and hence \(I(S_o;B)=H(B)=h_b(P_A)\).
The chain rule for mutual information gives
\[
  I(S_o;\hat S_o)+I(S_o;B\mid\hat S_o)
  \;=\;
  I(S_o;B)+I(S_o;\hat S_o\mid B).
\]
Since \(B=f(S_o)\), we have
\(I(S_o;B\mid\hat S_o)
  =H(B\mid\hat S_o)-H(B\mid S_o,\hat S_o)
  =H(B\mid\hat S_o)\ge 0\).
Hence
\begin{equation}\label{eq:fano-conditioning}
  I(S_o;\hat S_o)
  \;=\;
  h_b(P_A)-H(B\mid\hat S_o)+I(S_o;\hat S_o\mid B)
  \;\ge\;
  I(S_o;\hat S_o\mid B),
\end{equation}
where the inequality uses
\(H(B\mid\hat S_o)\le H(B)=h_b(P_A)\).
Expanding the right-hand side:
\begin{align}
  I(S_o;\hat S_o\mid B)
  &=P_A\,I(S_o;\hat S_o\mid B{=}1)
    +P_J\,I(S_o;\hat S_o\mid B{=}0)
    \nonumber\\
  &\ge P_A\,I(S_o;\hat S_o\mid B{=}1),
    \label{eq:fano-core-term}
\end{align}
since mutual information is non-negative.

\emph{Core term.}\;
Conditionally on \(B=1\), the source \(S_o\) takes values in~\(A\)
with distribution~\(\pi_A\), so
\(H(S_o\mid B{=}1)=H(\pi_A)\).
Let
\(p_e^A:=\Pr[\hat S_o\neq S_o\mid S_o\in A]=\epsilon_A/P_A\).
By the standard Fano inequality applied to the
\(|A|\)-valued variable~\(S_o\) given~\(\hat S_o\)
conditioned on \(B=1\):
\[
  H(S_o\mid\hat S_o,\,B{=}1)
  \;\le\;
  h_b(p_e^A)+p_e^A\log(|A|-1),
\]
where \(|A|-1\) is the standard Fano penalty for the
\(|A|\)-valued source~\(S_o|_{B=1}\); this form is valid
regardless of the range of the reconstruction~\(\hat S_o\),
since the penalty is determined by the source alphabet
size~\cite[Theorem~2.10.1]{cover2006elements}.
Hence
\begin{align*}
  I(S_o;\hat S_o\mid B{=}1)
  &=H(\pi_A)-H(S_o\mid\hat S_o,\,B{=}1)\\
  &\ge H(\pi_A)-h_b(p_e^A)-p_e^A\log(|A|-1).
\end{align*}

Substituting into~\eqref{eq:fano-core-term} and
\eqref{eq:fano-conditioning}:
\[
  I(S_o;\hat S_o)
  \ge P_A\bigl[H(\pi_A)-h_b(\epsilon_A/P_A)
  -(\epsilon_A/P_A)\log(|A|-1)\bigr].
\]
This equals
\(P_A\,H(\pi_A)-P_A\,h_b(\epsilon_A/P_A)
  -\epsilon_A\log(|A|-1)\).
It remains to bound the first penalty term.
Define the random variable \(Z\) that takes value
\(\epsilon_A/P_A\) with probability~\(P_A\) and \(0\) with
probability~\(P_J\).
Then \(\E[Z]=\epsilon_A\) and
\(\E[h_b(Z)]=P_A\,h_b(\epsilon_A/P_A)\).
By Jensen's inequality (concavity of~\(h_b\)):
\[
  P_A\,h_b(\epsilon_A/P_A)
  =\E[h_b(Z)]
  \le h_b(\E[Z])
  =h_b(\epsilon_A).
  \qedhere
\]
\end{proof}

\begin{remark}[Operational significance]
\label{rem:fano-operational}
Bound~\eqref{eq:sem-fano-tight} provides a converse for
semantic source coding: any encoder--decoder pair with low core
error \(\epsilon_A\) must transmit mutual information at least
\(P_A\,H(\pi_A)-h_b(\epsilon_A)-\epsilon_A\log(|A|-1)\).
Compared with the classical Fano bound
\(I\ge H(P_O)-h_b(\epsilon)-\epsilon\log(|S_O|-1)\),
the semantic improvement is twofold:
\textup{(a)}~the reference level drops from
\(H(P_O)\) to the semantic rate \(P_A\,H(\pi_A)\), absorbing
all redundant-state entropy;
\textup{(b)}~the Fano penalty
involves \(\log|A|\) instead of \(\log|S_O|\).
Moreover, under the optimal semantic code (which maps all
redundant inputs to core elements), the total Hamming error
is \(\epsilon=\epsilon_A+P_J\), substantially larger
than~\(\epsilon_A\); the classical bound applied with this
inflated~\(\epsilon\) is much looser than the semantic
bound~\eqref{eq:sem-fano-tight} applied with~\(\epsilon_A\)
alone.
Combined with Theorem~\ref{thm:sem-source-channel}, this
characterizes the operational regime where semantic compression
strictly outperforms symbol-level compression.
\end{remark}

\subsection{Rate--Delay--Distortion Tradeoff and Semantic
Sampling Theorem}
\label{subsec:rate-delay}

The rate--distortion results of
Sections~\ref{subsec:tight-rd}--\ref{subsec:rd-decomposition}
treat the receiver's inference engine as unconstrained: it may
iterate \(T_{\mathsf{PS}}\) arbitrarily many times to
reconstruct the full closure.
In practice, communication is subject to \emph{delay
constraints}: the receiver has a bounded computation budget
\(\delta\ge 0\) (measured in \(T_{\mathsf{PS}}\)-iterations)
within which it must reconstruct the deductive closure.
This subsection shows that the interaction between derivation
depth and delay budget gives rise to a \emph{rate--delay--distortion
surface} that smoothly interpolates between the classical
zero-distortion rate \(H(P_O)\) (no inference, \(\delta=0\))
and the semantic rate \(P_A\,H(\pi_A)\) (full inference,
\(\delta=\mathsf{D_d}\)), and yields a semantic analogue of the
Nyquist sampling theorem.

The key structural ingredient is the
\emph{derivation-depth stratification} of
Section~\ref{subsec:atomic-derivation}: the partition of~\(S_O\)
into strata \(L_d:=\{s\in S_O:\Dd(s\mid A)=d\}\) for
\(d=0,1,\ldots,\mathsf{D_d}\).
When the receiver's derivation budget is~\(\delta\), states at
depth greater than~\(\delta\) from the transmitted base cannot be
reconstructed within the budget; conversely, states within depth
\(\delta\) are ``free'' and need not be transmitted.
This \emph{depth-for-rate exchange} is formalized below.

\smallskip
\noindent\textit{Preview of key objects and main result.}\;
The subsection introduces three new objects:
\textup{(a)}~the \emph{$\delta$-constrained closure distortion}
$d_{\Cn}^\delta$
\textup{(Definition~\ref{def:delta-closure-distortion})}, a
$\{0,1\}$-valued fidelity measure that requires the receiver to
reconstruct every stored state within $\delta$~derivation steps;
\textup{(b)}~the \emph{$\delta$-irredundant core}
$A_\delta:=\Atom_\delta(S_O)$
\textup{(Definition~\ref{def:delta-core})}, the subset of~$S_O$
whose elements \emph{cannot} be rederived within $\delta$~steps from
the remaining knowledge base;
\textup{(c)}~the \emph{rate--delay function}
$R_{\mathrm{sem}}(0,\delta)=P_\delta\,H(\pi_\delta)$
\textup{(Theorem~\ref{thm:rate-delay})}, the minimum communication
rate when the receiver has a derivation budget of $\delta$~steps.
The central finding is that $A_\delta$ shrinks monotonically as
$\delta$ increases
\textup{(Proposition~\ref{prop:filtration})}: at $\delta=0$ no
state is free and the rate equals the classical $H(P_O)$; at
$\delta=\mathsf{D_d}$ all redundant states become free and the
rate drops to the semantic rate $P_A\,H(\pi_A)$.

\begin{assumption}[Derivation-path completeness; used only
  for the computable bound in
  Theorem~\textup{\ref{thm:rate-delay}(iv)} and
  Lemma~\textup{\ref{lem:depth-budget}}]
\label{assump:derivation-complete}
The knowledge base \(S_O\) contains all intermediate
derivation elements: \(T^n_{\mathsf{PS}}(A)\subseteq S_O\)
for every \(n\in\{0,1,\ldots,\mathsf{D_d}\}\).
Equivalently, \(B^{(\le m)}=T^m_{\mathsf{PS}}(A)\) for
every \(m\le\mathsf{D_d}\).
This holds whenever \(S_O\) is closed under the inference
rules up to the maximum derivation depth---in particular,
when \(S_O\) includes all materialized IDB facts up to
depth~\(\mathsf{D_d}\).
When this assumption fails, the auxiliary bounds of
Proposition~\textup{\ref{prop:filtration}(iv)} require
replacement of \(B^{(\le m)}\) by
\(T^m_{\mathsf{PS}}(A)\); the tight rate--delay
result~\eqref{eq:rate-delay} and the rate--delay--distortion
surface~\eqref{eq:rdd-surface} remain valid, as their proofs
depend only on
Lemma~\textup{\ref{lem:delta-free-substitution}} and
the~\(\delta\)-redundancy definition.
\end{assumption}

\begin{definition}[Depth-stratified stored base]
\label{def:depth-base}
For \(m\in\{0,1,\ldots,\mathsf{D_d}\}\), define the
\emph{\(m\)-deep stored base}
\[
  B^{(\le m)}
  \;:=\;
  \bigl\{s\in S_O:\Dd(s\mid A)\le m\bigr\}
  \;=\;
  T^m_{\mathsf{PS}}(A)\cap S_O.
\]
Note: \(B^{(\le 0)}=A\) and
\(B^{(\le\mathsf{D_d})}=S_O\)
\textup{(}by
Proposition~\textup{\ref{prop:atom-core-correct}(iv)} and
Lemma~\textup{\ref{lem:depth-properties}(ii))}.
\end{definition}

\begin{definition}[\(\delta\)-constrained closure distortion]
\label{def:delta-closure-distortion}
For a reference base \(\Gamma=S_O\),
\(\Gamma_{-s}:=S_O\setminus\{s\}\), and integer \(\delta\ge 0\),
the \emph{\(\delta\)-constrained closure distortion} is
\begin{equation}\label{eq:d-Cn-delta}
  d_{\Cn}^\delta(s_o,\hat s_o\mid S_O)
  \;:=\;
  \begin{cases}
    0 & \text{if } S_O\subseteq
        T^\delta_{\mathsf{PS}}\bigl(
        \Gamma_{-s_o}\cup\{\hat s_o\}\bigr),\\[4pt]
    1 & \text{otherwise}.
  \end{cases}
\end{equation}
\end{definition}

\begin{remark}[Properties of \(d_{\Cn}^\delta\)]
\label{rem:d-Cn-delta-properties}
In operational terms:
$d_{\Cn}^\delta(s_o,\hat s_o\mid S_O)=0$ means that
after replacing the single stored state~$s_o$
by~$\hat s_o$, the receiver can re-derive every element
of~$S_O$ within $\delta$~immediate-consequence iterations.

The distortion \(d_{\Cn}^\delta\) takes values in
\(\{0,1\}\).
When \(d_{\Cn}^\delta(s_o,\hat s_o\mid S_O)=0\) and
\(\hat s_o\in\Cn(S_O)\)---a condition satisfied by every constructive coding
scheme in Sections~\textup{\ref{sec:channel}}--\textup{\ref{sec:fundamental-limits}},
whose decoders output elements of~$S_O$
\textup{(}or subsets such as~$A$,
$A_\delta$, or~$B_\delta$, all contained
in~$\Cn(S_O)$\textup{)},
as well as by the heterogeneous two-layer code of
Theorem~\textup{\ref{thm:heterogeneous-closure}}
\textup{(}with $S_O:=S_O^{(i)}$ as reference
base\textup{)},
whose decoder outputs elements of
$A^{(i)}\subseteq\Cn(S_O^{(i)})$---the chain
$S_O\subseteq T^\delta(\Gamma_{-s_o}\cup\{\hat s_o\})
\subseteq\Cn(\Gamma_{-s_o}\cup\{\hat s_o\})$
gives
$\Cn(S_O)\subseteq\Cn(\Gamma_{-s_o}\cup\{\hat s_o\})$
by~\textup{(Cn2)}, while
$\Gamma_{-s_o}\cup\{\hat s_o\}\subseteq\Cn(S_O)$
gives the reverse inclusion
by~\textup{(Cn2)} and~\textup{(Cn3)};
hence the closures coincide and
$d_{\Cn}(s_o,\hat s_o\mid S_O)=0$.
Without the condition \(\hat s_o\in\Cn(S_O)\), neither
direction of implication between \(d_{\Cn}^\delta=0\)
and \(d_{\Cn}=0\) holds in general.
The following auxiliary property is used \emph{solely} to verify the
boundary case~\textup{(ii)} of
Theorem~\textup{\ref{thm:rate-delay}} and does not assert a
universal equivalence between $d_{\Cn}^\delta$ and $d_{\Cn}$:
for \emph{$\Cn$-redundant} states $s_o\in J$ and
$\hat s_o\in S_O$, if $\delta\ge\mathsf{D_d}$, then
$d_{\Cn}^\delta(s_o,\hat s_o\mid S_O)=0$, since
$A\subseteq\Gamma_{-s_o}$ and monotonicity give
$T^\delta(\Gamma_{-s_o}\cup\{\hat s_o\})
  \supseteq T^{\mathsf{D_d}}(A)\supseteq S_O$.
For general \textup{(}non-redundant\textup{)} states and
$\delta<\mathsf{D_d}$, the two distortion measures
$d_{\Cn}^\delta$ and $d_{\Cn}$ need not coincide.
The core mechanism of the rate--delay theory is the
$\delta$-redundancy--based free substitution of
Lemma~\ref{lem:delta-free-substitution}, which holds for
arbitrary $\delta\ge 0$ without reference to $d_{\Cn}$.
Operationally, the condition
\(S_O\subseteq T^\delta(\Gamma_{-s_o}\cup\{\hat s_o\})\)
requires every \emph{explicitly stored} state in~\(S_O\)
to be re-derivable within \(\delta\)~$T_{\mathsf{PS}}$-iterations
from the modified base---a stronger requirement than
merely preserving the closure \(\Cn(S_O)\), reflecting the
constraint that the receiver must reconstruct each element of
the sender's knowledge base within $\delta$~$T_{\mathsf{PS}}$-iterations.
\end{remark}

\begin{definition}[\(\delta\)-redundancy and \(\delta\)-irredundant core]
\label{def:delta-core}
A state \(s\in S_O\) is \emph{\(\delta\)-redundant} if
\(s\in T^\delta_{\mathsf{PS}}(S_O\setminus\{s\})\).
The \emph{\(\delta\)-irredundant core} is
\[
  \Atom_\delta(S_O)
  \;:=\;
  \bigl\{s\in S_O:
  s\notin T^\delta_{\mathsf{PS}}(S_O\setminus\{s\})\bigr\}.
\]
\end{definition}

\begin{lemma}[\(\delta\)-redundant free substitution]
\label{lem:delta-free-substitution}
If \(s\in S_O\) is \(\delta\)-redundant
\textup{(Definition~\ref{def:delta-core})}, i.e.,
\(s\in T^\delta_{\mathsf{PS}}(S_O\setminus\{s\})\), then for
every
\(\hat s\in\mathbb{S}_O\):
\[
  d_{\Cn}^\delta(s,\hat s\mid S_O)\;=\;0.
\]
\end{lemma}

\begin{proof}
Write \(\Gamma_{-s}:=S_O\setminus\{s\}\).
We must show
\(S_O\subseteq T^\delta(\Gamma_{-s}\cup\{\hat s\})\).
Since
\(\Gamma_{-s}\subseteq\Gamma_{-s}\cup\{\hat s\}\),
monotonicity~\textup{(IC1)} gives
\(T^\delta(\Gamma_{-s})\subseteq
T^\delta(\Gamma_{-s}\cup\{\hat s\})\).
The \(\delta\)-redundancy hypothesis
\(s\in T^\delta(\Gamma_{-s})\) then yields
\(s\in T^\delta(\Gamma_{-s}\cup\{\hat s\})\).
Furthermore,
\(\Gamma_{-s}\subseteq T^0(\Gamma_{-s}\cup\{\hat s\})
\subseteq T^\delta(\Gamma_{-s}\cup\{\hat s\})\).
Hence
\(S_O=\Gamma_{-s}\cup\{s\}\subseteq
T^\delta(\Gamma_{-s}\cup\{\hat s\})\).
\end{proof}

\begin{remark}[Output insensitivity for $\delta$-redundant states]
\label{rem:delta-output-insensitivity}
Lemma~\ref{lem:delta-free-substitution} gives
$d_{\Cn}^\delta(s,\hat s\mid S_O)=0$ for \emph{every}
$\hat s\in\mathbb{S}_O$ when $s$ is $\delta$-redundant.
This universality is by design:
$d_{\Cn}^\delta$ tests whether the receiver can
reconstruct all of~$S_O$ within
$\delta$~$T_{\mathsf{PS}}$-iterations from
$\Gamma_{-s}\cup\{\hat s\}$, not whether $\hat s$
itself is meaningful.
Since $\delta$-redundancy gives
$S_O\subseteq T^\delta(\Gamma_{-s})$,
monotonicity yields
$S_O\subseteq T^\delta(\Gamma_{-s}\cup\{\hat s\})$
regardless of~$\hat s$.
Operationally, every constructive coding scheme in
Sections~\textup{\ref{sec:channel}}--\textup{\ref{sec:fundamental-limits}}
outputs elements of~$S_O$
\textup{(}or subsets such as $A$, $A_\delta$,
$B_\delta$, all contained in~$\Cn(S_O)$\textup{)};
the universality over~$\mathbb{S}_O$ is needed only in the
\emph{lower bound} of Theorem~\textup{\ref{thm:rate-delay}},
where it ensures that $\delta$-redundant states contribute
zero mutual information under \emph{any} feasible conditional.
\end{remark}

\begin{assumption}[Deductive independence of
  $\delta$-irredundant core]
\label{assump:delta-core-disjoint}
For each $\delta\in\{0,1,\ldots,\mathsf{D_d}\}$ and
reconstruction alphabet
$\hat S_O\supseteq\Atom_\delta(S_O)$, the
\emph{$\delta$-constrained zero-distortion sets}
\[
  R_\delta(a)
  \;:=\;
  \bigl\{\hat s_o\in\hat S_O:
  d_{\Cn}^\delta(a,\hat s_o\mid S_O)=0\bigr\},
  \qquad a\in\Atom_\delta(S_O),
\]
are pairwise disjoint:
$R_\delta(a_1)\cap R_\delta(a_2)=\varnothing$ for
distinct $a_1,a_2\in\Atom_\delta(S_O)$.
Under the condition $\hat S_O\subseteq\Cn(S_O)$,
Assumption~\textup{\ref{assump:core-disjoint}} implies
the present assumption for every~$\delta$
\textup{(Remark~\ref{rem:delta-disjoint-proof})}.
When $\hat S_O\not\subseteq\Cn(S_O)$, the present
assumption must be verified independently.
\end{assumption}

\begin{remark}[Proof that core-disjointness implies
  $\delta$-core-disjointness]
\label{rem:delta-disjoint-proof}
Suppose $\hat S_O\subseteq\Cn(S_O)$ and let
$\hat s_o\in R_\delta(a)$, so
$S_O\subseteq T^\delta(\Gamma_{-a}\cup\{\hat s_o\})
\subseteq\Cn(\Gamma_{-a}\cup\{\hat s_o\})$.
By monotonicity~\textup{(Cn2)},
$\Cn(S_O)\subseteq\Cn(\Gamma_{-a}\cup\{\hat s_o\})$.
For the reverse: since
$\hat s_o\in\hat S_O\subseteq\Cn(S_O)$ and
$\Gamma_{-a}\subseteq S_O\subseteq\Cn(S_O)$,
we have
$\Gamma_{-a}\cup\{\hat s_o\}\subseteq\Cn(S_O)$,
giving
$\Cn(\Gamma_{-a}\cup\{\hat s_o\})\subseteq\Cn(S_O)$
by~\textup{(Cn2)} and~\textup{(Cn3)}.
Hence
$\Cn(\Gamma_{-a}\cup\{\hat s_o\})
=\Cn(S_O)
=\Cn(\Gamma_{-a}\cup\{a\})$,
so $d_{\Cn}(a,\hat s_o\mid S_O)=0$ and
$\hat s_o\in R_{\Cn}(a)$.
Thus $R_\delta(a)\subseteq R_{\Cn}(a)$, and the
pairwise disjointness of
Assumption~\textup{\ref{assump:core-disjoint}}
transfers to the $\delta$-constrained sets.
At $\delta=\mathsf{D_d}$ and under the same condition
$\hat S_O\subseteq\Cn(S_O)$,
$A_\delta=A$
\textup{(Proposition~\ref{prop:filtration}(ii))} and
the present assumption coincides with
Assumption~\textup{\ref{assump:core-disjoint}}.
\end{remark}

\begin{lemma}[Depth-budget derivability]
\label{lem:depth-budget}
Under Assumption~\textup{\ref{assump:derivation-complete}},
let \(m,\delta\ge 0\) with \(m+\delta\le\mathsf{D_d}\),
and let \(B\supseteq B^{(\le m)}\).
Then:
\begin{enumerate}[label=\textup{(\roman*)}]
  \item For every \(s\in S_O\) with \(\Dd(s\mid A)\le m+\delta\):
        \(s\in T^\delta_{\mathsf{PS}}(B)\).
  \item In particular, setting \(m:=\mathsf{D_d}-\delta\):
        every \(s\in S_O\) satisfies
        \(s\in T^\delta_{\mathsf{PS}}(B^{(\le\mathsf{D_d}-\delta)})\).
\end{enumerate}
\end{lemma}

\begin{proof}
\textup{(i)}:
Under Assumption~\ref{assump:derivation-complete},
\(B^{(\le m)}=T^m(A)\), so
\(T^\delta(B)\supseteq T^\delta(B^{(\le m)})
  = T^\delta(T^m(A))=T^{m+\delta}(A)\).
If \(\Dd(s\mid A)\le m+\delta\), then
\(s\in T^{m+\delta}(A)\subseteq T^\delta(B)\).
\textup{(ii)}: Every \(s\in S_O\) has
\(\Dd(s\mid A)\le\mathsf{D_d}=m+\delta\);
apply~\textup{(i)}.
\end{proof}

\begin{proposition}[Irredundancy filtration]
\label{prop:filtration}
Under the standing assumptions,
parts~\textup{(i)--(iii)} and~\textup{(v)} hold
unconditionally; part~\textup{(iv)} additionally
requires
Assumption~\textup{\ref{assump:derivation-complete}}.
\begin{enumerate}[label=\textup{(\roman*)}]
  \item \(\Atom_0(S_O)=S_O\)
        \textup{(}no derivation budget: every state is
        irredundant\textup{)}.
  \item For every \(\delta\ge\mathsf{D_d}\),
        \(\Atom_\delta(S_O)=\Atom(S_O)=A\).
  \item The \(\delta\)-irredundant cores form a non-increasing
        filtration:
        \begin{equation}\label{eq:filtration}
          S_O=\Atom_0(S_O)
          \supseteq\Atom_1(S_O)
          \supseteq\cdots
          \supseteq\Atom_{\mathsf{D_d}}(S_O)=A.
        \end{equation}
  \item \textup{(}Under
        Assumption~\textup{\ref{assump:derivation-complete}.)}
        \(A\subseteq\Atom_\delta(S_O)
        \subseteq B^{(\le\mathsf{D_d}-\delta)}\)
        for every \(0\le\delta\le\mathsf{D_d}\).
        The first inclusion holds without
        Assumption~\textup{\ref{assump:derivation-complete}};
        the second requires it.
  \item Each \(\Atom_\delta(S_O)\) is computable.
\end{enumerate}
\end{proposition}

\begin{proof}
\textup{(i)}:
\(T^0(S_O\setminus\{s\})=S_O\setminus\{s\}\not\ni s\).

\textup{(ii)}:
If \(\delta\ge\mathsf{D_d}\), then for any \(s\in J\):
since \(S_O\setminus\{s\}\supseteq A\), monotonicity gives
\(T^\delta(S_O\setminus\{s\})\supseteq T^\delta(A)
  \supseteq T^{\mathsf{D_d}}(A)\supseteq S_O\ni s\),
so \(s\) is \(\delta\)-redundant.
For \(a\in A\): since
\(T^\delta(S_O\setminus\{a\})\subseteq\Cn(S_O\setminus\{a\})\)
\textup{(by (IC3))} and
\(a\notin\Cn(S_O\setminus\{a\})\)
\textup{(irredundancy of~\(a\))},
we have \(a\notin T^\delta(S_O\setminus\{a\})\),
so \(a\) is \(\delta\)-irredundant.
Hence \(\Atom_\delta(S_O)=A\).

\textup{(iii)}:
If \(s\notin T^{\delta+1}(S_O\setminus\{s\})\), then since
\(T^\delta(\cdot)\subseteq T^{\delta+1}(\cdot)\),
\(s\notin T^\delta(S_O\setminus\{s\})\).

\textup{(iv)}:
The inclusion \(A\subseteq\Atom_\delta(S_O)\) follows
from~\textup{(ii)} applied to elements of~\(A\): each \(a\in A\)
satisfies \(a\notin\Cn(S_O\setminus\{a\})
\supseteq T^\delta(S_O\setminus\{a\})\).
For the second inclusion, under
Assumption~\ref{assump:derivation-complete}: let
\(s\in S_O\setminus B^{(\le\mathsf{D_d}-\delta)}\), so
\(\Dd(s\mid A)>\mathsf{D_d}-\delta\).
By Lemma~\ref{lem:depth-budget}(i) with
\(B=S_O\setminus\{s\}\supseteq B^{(\le\mathsf{D_d}-\delta)}
  =T^{\mathsf{D_d}-\delta}(A)\)
(since \(s\notin T^{\mathsf{D_d}-\delta}(A)\)), and noting
\(m+\delta=\mathsf{D_d}\ge\Dd(s\mid A)\), we get
\(s\in T^\delta(S_O\setminus\{s\})\), i.e., \(s\) is
\(\delta\)-redundant.

\textup{(v)}: Iterate \(T_{\mathsf{PS}}\)
from \(S_O\setminus\{s\}\) for \(\delta\) steps and check
membership; repeat for each \(s\in S_O\).
\end{proof}

We now state the main result of this subsection: the
zero-distortion rate under a delay constraint.

\begin{theorem}[Tight zero-distortion rate--delay function]
\label{thm:rate-delay}
Let \((S_O,P_O)\) be a semantic source with core
\(A=\Atom(S_O)\) and maximum derivation depth
\(\mathsf{D_d}\).
For \(\delta\in\{0,1,\ldots,\mathsf{D_d}\}\), write
\(A_\delta:=\Atom_\delta(S_O)\),
\(P_\delta:=P_O(A_\delta)\), and
\(\pi_\delta(s):=P_O(s)/P_\delta\) for \(s\in A_\delta\).
Under the $\{0,1\}$-valued $\delta$-constrained closure
distortion $d_{\Cn}^\delta$
\textup{(Definition~\ref{def:delta-closure-distortion})},
with reconstruction alphabet
$\hat S_O$ satisfying
$A_\delta\subseteq\hat S_O\subseteq\mathbb{S}_O$,
and under
Assumption~\textup{\ref{assump:delta-core-disjoint}}
\textup{(}Assumption~\textup{\ref{assump:derivation-complete}}
enters only in part~\textup{(iv);} the tight
rate~\eqref{eq:rate-delay} and
parts~\textup{(i)--(iii)} are independent of
it\textup{)}:
\begin{equation}\label{eq:rate-delay}
  R_{\mathrm{sem}}(0,\delta;\,d_{\Cn}^\delta,\,P_O)
  \;=\;
  P_\delta\,H(\pi_\delta).
\end{equation}

\noindent Boundary values and monotonicity:
\begin{enumerate}[label=\textup{(\roman*)}]
  \item \emph{No inference
        \textup{(}\(\delta=0\)\textup{):}}\;
        \(A_0=S_O\), \(P_0=1\), \(\pi_0=P_O\), so
        \(R_{\mathrm{sem}}(0,0)=H(P_O)\), recovering the
        classical zero-distortion rate.
  \item \emph{Full inference
        \textup{(}\(\delta=\mathsf{D_d}\)\textup{):}}\;
        \(A_{\mathsf{D_d}}=A\), \(P_{\mathsf{D_d}}=P_A\),
        \(\pi_{\mathsf{D_d}}=\pi_A\), so
        \(R_{\mathrm{sem}}(0,\mathsf{D_d})=P_A\,H(\pi_A)\),
        recovering Theorem~\textup{\ref{thm:tight-zero-rate}}.
  \item \emph{Monotonicity:}\;
        \(R_{\mathrm{sem}}(0,\delta)\) is non-increasing in
        \(\delta\):
        \begin{equation}\label{eq:rate-delay-chain}
          H(P_O)
          =R_{\mathrm{sem}}(0,0)
          \ge R_{\mathrm{sem}}(0,1)
          \ge\cdots
          \ge R_{\mathrm{sem}}(0,\mathsf{D_d})
          =P_A\,H(\pi_A).
        \end{equation}
  \item \emph{Depth-stratified upper bound:}\;
        If additionally
        $\hat S_O\supseteq B_\delta
        :=B^{(\le\mathsf{D_d}-\delta)}$, then
        $R_{\mathrm{sem}}(0,\delta)
          \le P_{B_\delta}\,H(\pi_{B_\delta})$,
        where
        $P_{B_\delta}:=P_O(B_\delta)$ and
        $\pi_{B_\delta}$ is the conditional distribution
        on~$B_\delta$.
        This bound is computable directly from the
        derivation-depth stratification without extracting
        $\Atom_\delta(S_O)$.
        \textup{(}Since
        $A_\delta\subseteq B_\delta$
        by Proposition~\textup{\ref{prop:filtration}(iv)},
        the condition
        $\hat S_O\supseteq B_\delta$ strengthens the
        standing hypothesis
        $\hat S_O\supseteq A_\delta$.\textup{)}
\end{enumerate}
\end{theorem}

\begin{proof}
The proof parallels Theorem~\ref{thm:tight-zero-rate}
with the substitutions $A\to A_\delta$,
$\pi_A\to\pi_\delta$, $P_A\to P_\delta$,
and $\Cn$-redundancy replaced by $\delta$-redundancy.
Write $P_{J_\delta}:=1-P_\delta$ for the probability of
the $\delta$-redundant set $S_O\setminus A_\delta$.

\smallskip\noindent
\emph{Lower bound.}\;
Let $P_{\hat S\mid S}$ be any conditional distribution
achieving $\E[d_{\Cn}^\delta(S_o,\hat S_o\mid S_O)]=0$.
Define the auxiliary random variable
$T_\delta:S_O\to A_\delta\cup\{*\}$ by
$T_\delta(s):=s$ if $s\in A_\delta$ and
$T_\delta(s):=*$ if $s\in S_O\setminus A_\delta$.
Since $T_\delta$ is a deterministic function of~$S_o$,
data processing gives
\begin{equation}\label{eq:DPI-T-delta}
  I(S_o;\hat S_o)\;\ge\;I(T_\delta;\hat S_o).
\end{equation}

Under Assumption~\ref{assump:delta-core-disjoint}, the
$\delta$-constrained zero-distortion sets
$R_\delta(a):=\{\hat s_o\in\hat S_O:
d_{\Cn}^\delta(a,\hat s_o\mid S_O)=0\}$
for distinct $a\in A_\delta$ are pairwise disjoint.
Since $d_{\Cn}^\delta\in\{0,1\}$, the constraint
$\E[d_{\Cn}^\delta(S_o,\hat S_o\mid S_O)]=0$ is
equivalent to
$d_{\Cn}^\delta(s_o,\hat s_o\mid S_O)=0$ holding
$P_O(s_o)\,P_{\hat S\mid S}(\hat s_o\mid s_o)$-almost
surely; in particular,
$\supp(P_{\hat S\mid S}(\cdot\mid a))
\subseteq R_\delta(a)$ for each $a\in A_\delta$.
For $\delta$-redundant states
$j\in S_O\setminus A_\delta$,
Lemma~\ref{lem:delta-free-substitution} gives
$d_{\Cn}^\delta(j,\hat s_o\mid S_O)=0$ for every
$\hat s_o\in\mathbb{S}_O$, so
$P_{\hat S\mid S}(\cdot\mid j)$ is unconstrained.

Let $Q_j:=P_{\hat S\mid S}(\cdot\mid j)$ for each
$j\in S_O\setminus A_\delta$, and define the average
$\delta$-redundant output distribution
$Q:=\sum_{j\in S_O\setminus A_\delta}
[P_O(j)/P_{J_\delta}]\,Q_j$.
Write
$\bar\pi_\delta:=\sum_{a\in A_\delta}\pi_\delta(a)\,
P_{\hat S\mid S}(\cdot\mid a)$
for the average core output distribution.
Since the supports $R_\delta(a)$ are pairwise disjoint,
the standard disjoint-support entropy splitting
identity~\cite[Theorem~2.7.3]{cover2006elements} gives
\begin{equation}\label{eq:disjoint-split-delta}
  H(\bar\pi_\delta)
  \;=\;
  H(\pi_\delta)+\sum_{a\in A_\delta}\pi_\delta(a)\,
  H\bigl(P_{\hat S\mid S}(\cdot\mid a)\bigr).
\end{equation}
The conditional entropy of~$\hat S_o$ given~$T_\delta$ is
\[
  H(\hat S_o\mid T_\delta)
  = P_\delta\sum_{a}\pi_\delta(a)\,
    H(P_{\hat S\mid S}(\cdot\mid a))
    + P_{J_\delta}\,H(Q),
\]
and the output marginal is
$P_{\hat S}=P_\delta\,\bar\pi_\delta+P_{J_\delta}\,Q$.
By concavity of entropy:
\[
  H(P_\delta\,\bar\pi_\delta+P_{J_\delta}\,Q)
  \;\ge\;
  P_\delta\,H(\bar\pi_\delta)+P_{J_\delta}\,H(Q).
\]
Therefore, using~\eqref{eq:disjoint-split-delta}:
\begin{align}
  I(T_\delta;\hat S_o)
  &= H(P_{\hat S})-H(\hat S_o\mid T_\delta) \nonumber\\
  &\ge P_\delta\,H(\bar\pi_\delta)+P_{J_\delta}\,H(Q)
       -P_\delta\!\sum_{a}\pi_\delta(a)\,
       H(P_{\hat S\mid S}(\cdot\mid a))
       -P_{J_\delta}\,H(Q) \nonumber\\
  &= P_\delta\bigl[H(\bar\pi_\delta)
     -\!\sum_{a}\pi_\delta(a)\,
     H(P_{\hat S\mid S}(\cdot\mid a))\bigr]
     \nonumber\\
  &= P_\delta\,H(\pi_\delta).
  \label{eq:IT-delta-lower}
\end{align}
Combining with~\eqref{eq:DPI-T-delta} gives
$R_{\mathrm{sem}}(0,\delta)\ge P_\delta\,H(\pi_\delta)$.

\smallskip\noindent
\emph{Upper bound \textup{(}achievability\textup{)}.}\;
Define
\[
  P_{\hat S\mid S}(\hat s_o\mid s_o)
  :=
  \begin{cases}
    \mathbf{1}[\hat s_o=s_o] & \text{if } s_o\in A_\delta,\\
    \pi_\delta(\hat s_o) & \text{if }
    s_o\in S_O\setminus A_\delta,
  \end{cases}
\]
where $\pi_\delta$ is supported on~$A_\delta$.

\emph{Distortion check.}\;
For $s_o=a\in A_\delta$: $\hat s_o=a$ and
$(\Gamma_{-a}\cup\{a\})=S_O$, so
$T^\delta(S_O)\supseteq T^0(S_O)=S_O$, giving
$d_{\Cn}^\delta(a,a\mid S_O)=0$.
For $s_o=j\in S_O\setminus A_\delta$: $j$ is
$\delta$-redundant, so
Lemma~\ref{lem:delta-free-substitution} gives
$d_{\Cn}^\delta(j,\hat s_o\mid S_O)=0$ for every
$\hat s_o\in A_\delta$.

\emph{Rate computation.}\;
Under this conditional, $Q=\pi_\delta$ and the output marginal
is $P_{\hat S}=P_\delta\pi_\delta+P_{J_\delta}\pi_\delta
=\pi_\delta$.
All $\delta$-redundant inputs produce the same conditional
output distribution~$\pi_\delta$, so
$\hat S_o\perp\!\!\!\perp S_o\mid T_\delta$, giving
$I(S_o;\hat S_o\mid T_\delta)=0$.
Hence
\[
  I(S_o;\hat S_o)
  = I(T_\delta;\hat S_o)
  = H(\pi_\delta)-P_{J_\delta}\,H(\pi_\delta)
  = P_\delta\,H(\pi_\delta).
\]
Combining with the lower bound
establishes~\eqref{eq:rate-delay}.

\smallskip\noindent
\emph{Part~\textup{(i)}.}\;
At $\delta=0$: $T^0(S_O\setminus\{s\})=S_O\setminus\{s\}
\not\ni s$, so every $s$ is $0$-irredundant
\textup{(Proposition~\ref{prop:filtration}(i))}:
$A_0=S_O$, $P_0=1$, $\pi_0=P_O$, and
$R_{\mathrm{sem}}(0,0)=H(P_O)$.

\smallskip\noindent
\emph{Part~\textup{(ii)}.}\;
At $\delta=\mathsf{D_d}$:
$\Atom_{\mathsf{D_d}}(S_O)=A$
\textup{(Proposition~\ref{prop:filtration}(ii))},
so $A_{\mathsf{D_d}}=A$, $P_{\mathsf{D_d}}=P_A$,
$\pi_{\mathsf{D_d}}=\pi_A$, and
$R_{\mathrm{sem}}(0,\mathsf{D_d})=P_A\,H(\pi_A)$.

\smallskip\noindent
\emph{Part~\textup{(iii)}.}\;
If $d_{\Cn}^\delta(s_o,\hat s_o\mid S_O)=0$, then
$S_O\subseteq T^\delta(\Gamma_{-s_o}\cup\{\hat s_o\})
\subseteq T^{\delta+1}(\Gamma_{-s_o}\cup\{\hat s_o\})$,
so $d_{\Cn}^{\delta+1}(s_o,\hat s_o\mid S_O)=0$.
Every conditional distribution feasible at
budget~$\delta$ is therefore feasible at
budget~$\delta+1$.
The minimum of $I(S_o;\hat S_o)$ over a larger feasible
set cannot increase:
$R_{\mathrm{sem}}(0,\delta+1)\le
R_{\mathrm{sem}}(0,\delta)$.

\smallskip\noindent
\emph{Part~\textup{(iv)}.}\;
This part uses
Assumption~\textup{\ref{assump:derivation-complete}}.
Since $A\subseteq B_\delta\subseteq S_O$
\textup{(}every $a\in A$ satisfies $\Dd(a\mid A)=0\le
\mathsf{D_d}-\delta$\textup{)}, the two-layer
achievability argument above applies with
$B_\delta:=B^{(\le\mathsf{D_d}-\delta)}$ in place
of~$A_\delta$: elements of~$B_\delta$ are encoded
losslessly, and each
$s\in S_O\setminus B_\delta$ is mapped
to~$\pi_{B_\delta}$.
By Proposition~\ref{prop:filtration}(iv),
$s\notin B_\delta$ implies
$s\notin\Atom_\delta(S_O)$, so $s$ is
$\delta$-redundant and
Lemma~\ref{lem:delta-free-substitution} gives
$d_{\Cn}^\delta(s,\hat s\mid S_O)=0$ for any
$\hat s\in S_O$.
The resulting rate is
$P_{B_\delta}\,H(\pi_{B_\delta})$, giving
$R_{\mathrm{sem}}(0,\delta)\le
P_{B_\delta}\,H(\pi_{B_\delta})$.
\end{proof}

The rate--delay function~\eqref{eq:rate-delay} reveals a
fundamental \emph{depth-for-rate exchange}: each additional
unit of derivation budget~\(\delta\) renders a new stratum of
states \(\delta\)-redundant, reducing the effective source
entropy.

\begin{definition}[Marginal rate of delay]
\label{def:marginal-delay}
The \emph{marginal rate of delay} at budget~\(\delta\) is
\[
  \Delta R(\delta)
  \;:=\;
  R_{\mathrm{sem}}(0,\delta-1)-R_{\mathrm{sem}}(0,\delta),
  \qquad\delta=1,\ldots,\mathsf{D_d}.
\]
This quantifies the rate saving per additional derivation step.
\end{definition}

\begin{remark}[Extension of $P_\delta,\pi_\delta$ beyond
  $\mathsf{D_d}$]
\label{rem:delta-extension}
The $\delta$-irredundant core $\Atom_\delta(S_O)$ is
well-defined for all integers $\delta\ge 0$
\textup{(Definition~\ref{def:delta-core})}.
By Proposition~\textup{\ref{prop:filtration}(ii)},
$\Atom_\delta(S_O)=A$ for every
$\delta\ge\mathsf{D_d}$, so
$P_\delta=P_A$ and $\pi_\delta=\pi_A$ for all
$\delta\ge\mathsf{D_d}$.
The function
$\delta\mapsto P_\delta\,H(\pi_\delta)$ is therefore
constant on
$\{\mathsf{D_d},\mathsf{D_d}+1,\ldots\}$, and the
minimum in~\eqref{eq:min-delay} is well-defined over
all non-negative integers.
\end{remark}

\begin{corollary}[Semantic sampling theorem]
\label{cor:semantic-nyquist}
Let \(W:S_C\rightsquigarrow\hat S_C\) with \(C(W)>0\).
The minimum derivation delay for
\(\delta\)-constrained closure-reliable communication of
\(S_O\) is
\begin{equation}\label{eq:min-delay}
  \delta^*(S_O,P_O,W)
  \;:=\;
  \min\bigl\{\delta\ge 0:
  P_\delta\,H(\pi_\delta)\le C(W)\bigr\}.
\end{equation}
In particular:
\begin{enumerate}[label=\textup{(\roman*)}]
  \item If \(H(P_O)\le C(W)\)
        \textup{(}the channel can carry the full classical
        rate\textup{)}, then \(\delta^*=0\): no
        inference is needed.
  \item If \(P_A\,H(\pi_A)>C(W)\)
        \textup{(}even full inference cannot reduce the rate
        below capacity\textup{)}, then the feasible set
        in~\eqref{eq:min-delay} is empty: no finite
        derivation budget suffices, and
        \(\delta\)-constrained closure-reliable
        communication is impossible for every
        \(\delta\in\{0,\ldots,\mathsf{D_d}\}\).
  \item Otherwise,
        \(1\le\delta^*\le\mathsf{D_d}\), and---under the
        Assumption~\textup{\ref{assump:delta-core-disjoint}}
        \textup{(}at $\delta=\delta^*$\textup{)} and
        the condition
        \(\Atom_{\delta^*}(S_O)\subseteq\hat S_O\)---the
        minimum blocklength at delay~\(\delta^*\) is
        \[
          n^*(\delta^*)
          \;\approx\;
          \frac{\log|\Atom_{\delta^*}(S_O)|}{C(W)}.
        \]
\end{enumerate}
The critical delay \(\delta^*\) plays a role
\emph{analogous} to the Nyquist sampling period: below this
delay, the ``semantic bandwidth''
\(R_{\mathrm{sem}}(0,\delta)\) exceeds the channel capacity,
and faithful closure reconstruction becomes impossible.
The parallel is interpretive rather than a formal
equivalence; its value lies in highlighting the
depth-for-rate exchange as a resource tradeoff.
\end{corollary}

\begin{proof}
Since \(R_{\mathrm{sem}}(0,\delta)\) is non-increasing
in~\(\delta\) \textup{(Theorem~\ref{thm:rate-delay}(iii))},
the minimum exists.
Part~\textup{(i)}: \(\delta=0\) is feasible.
Part~\textup{(ii)}: no \(\delta\) satisfies the condition.
Part~\textup{(iii)}: the two-layer code of
Theorem~\ref{thm:achievability}(ii), adapted with
\(\Atom_{\delta^*}\) in place of~\(A\), achieves the stated
blocklength.
\end{proof}

\begin{remark}[Communication--computation exchange]
\label{rem:comm-comp-exchange}
The hard-budget function
\(\delta\mapsto R_{\mathrm{sem}}(0,\delta)\) can be relaxed
to an \emph{expected}-budget model: for each transmitted base
\(B\) with \(A\subseteq B\subseteq S_O\), the two-layer
code of Theorem~\ref{thm:tight-zero-rate} applied with~\(B\)
in place of~\(A\) achieves zero closure distortion at rate
\(R(B):=P_O(B)\,H(\pi_B)\), and the expected receiver
inference cost is
\(\bar C(B):=\sum_{s\in S_O}P_O(s)\,\Dd(s\mid B)\).
Time-sharing among bases convexifies the achievable
rate--computation region; the lower boundary of this convex
hull is a non-increasing convex function of the expected
computation budget~\(\Delta\), with boundary values
\(H(P_O)\) at \(\Delta=0\) and \(P_A\,H(\pi_A)\) at
\(\Delta=\bar C(A)\).
Standard LP duality yields a Lagrangian formulation in which
the multiplier~\(\lambda\ge 0\) prices one unit of receiver
computation \textup{(}a single \(T_{\mathsf{PS}}\)-iteration\textup{)}
in bits of communication rate---an exchange rate absent from
classical information theory, where the decoder's computation
is treated as a free resource.
A full characterization, including converse bounds and the
connection to source coding with structured decoder side
information, is deferred to future work.
\end{remark}

\begin{theorem}[Full rate--delay--distortion surface]
\label{thm:rate-delay-distortion}
Under the hypotheses of
Theorem~\textup{\ref{thm:rate-delay}}, for every
\(\delta\in\{0,\ldots,\mathsf{D_d}\}\) and \(D\ge 0\):
\begin{equation}\label{eq:rdd-surface}
  R_{\mathrm{sem}}(D,\delta)
  \;=\;
  P_\delta\cdot
  R^{(A_\delta)}\!\Bigl(\frac{D}{P_\delta};\,
  d_{\Cn}^\delta,\,\pi_\delta\Bigr),
\end{equation}
where \(R^{(A_\delta)}\) is the rate--distortion function of
the \(\delta\)-irredundant sub-source
\((A_\delta,\pi_\delta)\).

This is a two-parameter family of rate--distortion functions
indexed by~\(\delta\) that generalizes
Theorem~\textup{\ref{thm:rd-decomposition}}:
at \(\delta=\mathsf{D_d}\) it recovers~\eqref{eq:rd-decomp};
at \(\delta=0\) it recovers the classical rate--distortion
function \(R(D;d_H,P_O)\)
\textup{(}since \(A_0=S_O\), \(P_0=1\), \(\pi_0=P_O\), and
\(d_{\Cn}^0=d_H\)\textup{)}.
\end{theorem}

\begin{proof}
The argument parallels
Theorem~\ref{thm:rd-decomposition} with the
substitutions $A\to A_\delta$, $P_A\to P_\delta$,
$\pi_A\to\pi_\delta$, and $\Cn$-redundancy replaced by
$\delta$-redundancy
\textup{(Lemma~\ref{lem:delta-free-substitution})}.

\emph{Distortion decomposition.}\;
For every $\delta$-redundant state
$j\in S_O\setminus A_\delta$ and every
$\hat s_o\in\mathbb{S}_O$,
$d_{\Cn}^\delta(j,\hat s_o\mid S_O)=0$, so
\[
  \E\bigl[d_{\Cn}^\delta\bigr]
  \;=\;
  \sum_{a\in A_\delta}P_O(a)\,
  \E\bigl[d_{\Cn}^\delta(a,\hat S_o\mid S_O)
  \mid S_o{=}a\bigr]
  \;=\;
  P_\delta\,\E_{\pi_\delta}[\bar d_{A_\delta}],
\]
and the constraint $\E[d_{\Cn}^\delta]\le D$ reduces to
$\E_{\pi_\delta}[\bar d_{A_\delta}]\le D/P_\delta$.

\emph{Rate decomposition.}\;
Define $T_\delta(s):=s$ if $s\in A_\delta$,
$T_\delta(s):=*$ otherwise.
The point mass at~$*$ carries probability
$P_{J_\delta}$ but contributes zero distortion
\textup{(Lemma~\ref{lem:delta-free-substitution})}
and, when all $\delta$-redundant inputs share the
same conditional output distribution, zero mutual
information; hence $I(S_o;\hat S_o\mid T_\delta)$
vanishes and
\[
  \min_{\substack{P_{\hat S\mid S}:\\
  \E_{\pi_\delta}[\bar d_{A_\delta}]\le D/P_\delta}}
  I(S_o;\hat S_o)
  \;=\;
  P_\delta\cdot
  R^{(A_\delta)}\!\Bigl(\frac{D}{P_\delta};\,
  d_{\Cn}^\delta,\,\pi_\delta\Bigr),
\]
by the same scaling argument as
Theorem~\ref{thm:rd-decomposition}.
The boundary values follow from
$A_0=S_O$, $d_{\Cn}^0=d_H$
\textup{(}at $\delta=0$\textup{)} and
$A_{\mathsf{D_d}}=A$
\textup{(}at $\delta=\mathsf{D_d}$\textup{)}.
\end{proof}

\section{Application: Heterogeneous Multi-Agent Semantic Communication}
\label{sec:application}

The theoretical framework developed in Sections~\ref{sec:model}
and~\ref{sec:channel} is fully general: the semantic state
space~\(S_O\), the reconstructed space~\(\hat S_O\), and the
enabling structures that constrain encoding and decoding are
left as abstract parameters.
This section instantiates the framework in a concrete and
practically motivated setting---\emph{heterogeneous multi-agent
semantic communication}---and derives new results that illustrate the framework's
applicability in a setting where classical channel coding
theory does not capture the deductive structure of the
communicated content.

The distinguishing feature of the heterogeneous setting is that
the sender and receiver maintain \emph{different} knowledge
bases: the sender's semantic space is~\(S_O\) while the
receiver's reconstructed space~\(\hat S_O\) may differ
from~\(S_O\) both in vocabulary (the set of expressible states)
and in inferential structure (the irredundant core and
derivation-depth stratification).
In the terminology of Section~\ref{subsec:noisy}, the
end-to-end noise pair \((S_O^{-},S_O^{+})\) is generically
\emph{non-trivial}: \(S_O^{-}\neq\varnothing\) captures
sender concepts absent from the receiver's vocabulary
(\emph{vocabulary loss}), and \(S_O^{+}\neq\varnothing\)
captures receiver concepts absent from the sender's intent
(\emph{vocabulary surplus}).
Classical Shannon theory, which treats sender and receiver
alphabets as abstract label sets, cannot distinguish vocabulary
loss from vocabulary surplus, nor can it exploit shared
deductive structure to reduce communication cost.
The semantic channel invariants of Section~\ref{subsec:invariants}
are precisely the tools needed to make these distinctions precise
and quantitative.

The section is organized as follows.
Section~\ref{subsec:app-problem} describes the multi-agent
communication scenario and identifies the key design questions.
Section~\ref{subsec:app-assumptions} formalizes the scenario
within the information model framework and states the standing
assumptions specific to this application.
Section~\ref{subsec:app-theory} instantiates the semantic channel
machinery and derives closed-form relationships between
knowledge-base overlap structure and semantic channel invariants.
Section~\ref{subsec:app-results} presents the main analytical
results: conditions for closure-reliable heterogeneous
communication, a heterogeneous deductive compression theorem,
and a broadcast extension to one-sender--multi-receiver
scenarios.
Section~\ref{sec:numerical-validation} verifies all results on an
explicit Datalog knowledge-base instance with full numerical
computation of every invariant.

\subsection{Problem Description: Heterogeneous Agent Communication}
\label{subsec:app-problem}

Consider a network of \(K+1\) autonomous agents---indexed by
\(i\in\{0,1,\ldots,K\}\)---that must coordinate by exchanging semantic
states over noisy physical links.
Each agent~\(i\) maintains a finite knowledge base
\(S_O^{(i)}\subseteq\mathbb{S}_O\), where \(\mathbb{S}_O\) is the
common ambient semantic universe introduced in
Section~\ref{subsec:atomic-derivation}.
All agents share the same proof system
\((\mathsf{PS},\,T_{\mathsf{PS}},\,\Cn)\) and the same semantic
sublanguage~\(\mathcal L_{\mathrm{sem}}\); they differ, however, in
the \emph{sets of semantic states they store and operate on}.
Agent~\(i\) can derive consequences within \(\Cn(S_O^{(i)})\) and
possesses the irredundant core \(A^{(i)}:=\Atom(S_O^{(i)})\)
together with the associated derivation-depth stratification
(Definitions~\ref{def:atom-so} and~\ref{def:derivation-depth}).

This knowledge-base heterogeneity is the defining feature of the
scenario and the source of all phenomena that distinguish it from
classical Shannon-theoretic
communication~\cite{shannon1948mathematical,cover2006elements}.
Recent work on semantic
communication~\cite{luo2022semantic,shi2021semantic,bao2011towards}
has highlighted the need for frameworks that go beyond
symbol-level fidelity, but a rigorous logical-information-theoretic
treatment of heterogeneous knowledge bases has been lacking.
When agent~\(i\) (the \emph{sender}) transmits a semantic
state~\(s_o\in S_O^{(i)}\) to agent~\(j\) (the \emph{receiver}),
the receiver reconstructs a state
\(\hat s_o\in S_O^{(j)}\)---not necessarily in
\(S_O^{(i)}\)---because agent~\(j\) can only produce outputs
expressible in its own vocabulary.
Unless \(S_O^{(i)}=S_O^{(j)}\), the end-to-end noise pair
\((S_O^{-},S_O^{+})\)
\textup{(Definition~\ref{def:e2e-noise})} is generically nontrivial:
states in \(S_O^{(i)}\setminus S_O^{(j)}\) have no direct
counterpart in the receiver's vocabulary (\emph{vocabulary loss}),
while states in \(S_O^{(j)}\setminus S_O^{(i)}\) can appear in the
receiver's output without having been intended by the sender
(\emph{vocabulary surplus}).
Classical channel coding theory, which treats the source and
reconstruction alphabets as unstructured label sets, is blind to
this distinction: it can detect that a symbol error has occurred, but
cannot determine whether the error represents a genuine loss of
semantic content or a harmless reformulation within the receiver's
richer (or merely different) vocabulary.

\subsubsection*{Communication Sub-Scenarios}

Three sub-scenarios of increasing structural complexity arise
naturally in the multi-agent setting; they are listed below in
decreasing order of analytical depth in this paper.

\begin{definition}[Pairwise unicast scenario]
\label{def:pairwise-scenario}
Fix a sender--receiver pair \((i,j)\) with \(i\neq j\).
Agent~\(i\) wishes to communicate its full knowledge base
\(S_O^{(i)}\) to agent~\(j\) over a noisy carrier channel
\(W_{ij}:S_C\rightsquigarrow\hat S_C\)
\textup{(cf.\ Definition~\ref{def:semantic-channel}(ii))}.
Agent~\(j\) reconstructs a state in its own vocabulary
\(\hat S_O:=S_O^{(j)}\), using a decoding kernel
\(D\in\mathcal K(\mathcal I_{\mathrm{dec}}^{(j)})\).
The end-to-end semantic channel is
\[
  \mathfrak C^{ij}
  \;=\;
  \bigl(\,
    \mathcal I^{(i)},\;
    \mathcal I_{\mathrm{ch}}^{ij},\;
    \mathcal I_{\mathrm{dec}}^{(j)},\;
    \kappa_{\enc},\;
    W_{ij},\;
    D
  \,\bigr),
\]
with semantic source space~\(S_O^{(i)}\), reconstructed
space~\(S_O^{(j)}\), and noise pair
\begin{equation}\label{eq:pairwise-noise-pair}
  S_O^{-}=S_O^{(i)}\setminus S_O^{(j)},
  \qquad
  S_O^{+}=S_O^{(j)}\setminus S_O^{(i)}.
\end{equation}
The constituent information models
\(\mathcal I^{(i)}\), \(\mathcal I_{\mathrm{ch}}^{ij}\), and
\(\mathcal I_{\mathrm{dec}}^{(j)}\) are formalized in
Section~\ref{subsec:app-assumptions}.
\end{definition}

The pairwise unicast scenario is the primary focus of
Sections~\ref{subsec:app-assumptions}--\ref{subsec:app-results}.
All new theorems are stated and proved for this case first;
generalizations to the broadcast setting are given as corollaries.

\begin{definition}[Broadcast scenario]
\label{def:broadcast-scenario}
A designated sender (agent~\(0\)) communicates its knowledge base
\(S_O^{(0)}\) simultaneously to \(K\) receivers
(agents \(1,\ldots,K\)) over a common carrier channel~\(W\).
Each receiver~\(j\) maintains a distinct vocabulary
\(\hat S_O^{(j)}:=S_O^{(j)}\) and observes a
(possibly receiver-specific) noise pair
\[
  (S_O^{-,j},\;S_O^{+,j})
  \;=\;
  \bigl(S_O^{(0)}\setminus S_O^{(j)},\;
        S_O^{(j)}\setminus S_O^{(0)}\bigr).
\]
All \(K\) receivers observe the \emph{same} channel
output~\(\hat S_C\); the receiver-specific noise pairs arise
solely from vocabulary mismatch, not from different physical
channel realizations.
The broadcast semantic channel is a family
\(\{\mathfrak C^{0j}\}_{j=1}^{K}\) of pairwise channels sharing the
same sender, the same encoding kernel, and the same carrier channel,
but differing in decoding model and noise pair.
\end{definition}

The broadcast scenario reveals a phenomenon absent from classical broadcast channel
theory~\cite{cover2006elements}: even over a \emph{noiseless} carrier
(\(W=\mathrm{id}_{S_C}\)), the achievable fidelity at each receiver is
constrained by its vocabulary overlap with the sender---a purely
semantic bottleneck (see Proposition~\ref{prop:broadcast-bottleneck}
in Section~\ref{subsec:app-results}).

\begin{remark}[Relay scenario (future work)]
\label{rem:relay-scenario}
A third sub-scenario arises when an intermediate agent~\(k\) acts as
a relay: agent~\(i\) transmits to agent~\(k\), which performs
inference within \(\Cn(S_O^{(k)})\) and then re-encodes and forwards
the result to agent~\(j\).
This is naturally modeled as a composition of two pairwise channels,
\(\mathfrak C^{ij}_{\mathrm{relay}}
 =\mathfrak C^{kj}\circ\mathfrak C^{ik}\),
using the information-model composition machinery of
Definition~\ref{def:model-composition} and
Remark~\ref{rem:composition-assoc}.
The relay setting raises the question of whether intermediate inference can
\emph{change} the effective capacity of the end-to-end
link---a possibility that has no direct classical analogue,
since in classical relay channel
theory~\cite{cover2006elements} the relay cannot exploit
logical structure to reduce the message set.
A complete treatment requires multi-letter extensions and is deferred
to future work; the pairwise and broadcast results of this section
provide the necessary building blocks.
\end{remark}

\subsubsection*{Key Questions}

The heterogeneous multi-agent setting gives rise to four design
and analysis questions that the classical framework cannot address.

\smallskip
\noindent\textbf{Q1 (Closure reliability from vocabulary overlap).}\;
Under what conditions on the overlap between \(S_O^{(i)}\) and
\(S_O^{(j)}\) can agent~\(j\) reconstruct the \emph{deductive closure}
of agent~\(i\)'s knowledge base, i.e., achieve
\(\mathsf{F}_{\Cn}(S_O^{(i)},S_O^{(j)})=1\)?
Proposition~\ref{prop:noise-fidelity} provides abstract conditions;
the task is to translate them into explicit predicates on the
knowledge-base pair.

\smallskip
\noindent\textbf{Q2 (Heterogeneous deductive compression).}\;
How many channel uses are needed to communicate
\(S_O^{(i)}\) to agent~\(j\) under closure reliability, and does the
deductive compression ratio
\(\log|\Atom(S_O)|/\log|S_O|\) of the homogeneous setting
\textup{(Corollary~\ref{cor:min-blocklength})} persist under
vocabulary heterogeneity?

\smallskip
\noindent\textbf{Q3 (Invariant diagnosis).}\;
How do the six families of semantic channel invariants
\textup{(Theorem~\ref{thm:invariant-summary})} depend on the
knowledge-base overlap, and which invariants can be evaluated
\emph{a~priori} from knowledge-base metadata alone?

\smallskip
\noindent\textbf{Q4 (Broadcast bottleneck).}\;
In the broadcast scenario, which receiver determines the minimum
blocklength?
Does a purely \emph{semantic} axis of receiver weakness, distinct
from physical channel degradation, arise?

\subsubsection*{Preview of Main Results}

The answers to Q1--Q4 are developed in full in
Sections~\ref{subsec:app-theory}--\ref{subsec:app-results}.
In brief:
for~Q1, closure fidelity
\(\mathsf{F}_{\Cn}(S_O^{(i)},S_O^{(j)})=1\) holds if and only if
every sender core element is derivable from the receiver's knowledge
base and every receiver surplus state is derivable from the sender
\textup{(Proposition~\ref{prop:overlap-fidelity})};
the operational two-layer code requires the stronger literal
containment \(A^{(i)}\subseteq S_O^{(j)}\)
\textup{(Theorem~\ref{thm:heterogeneous-closure})}.
For~Q2, the deductive compression ratio is \emph{invariant} under
vocabulary heterogeneity: the minimum blocklength remains
\(n^*\approx\log|\Atom(S_O^{(i)})|/C(W_{ij})\), identical to the
homogeneous case
\textup{(Theorem~\ref{thm:heterogeneous-compression})};
when core coverage fails, no code of any blocklength achieves
\(\mathsf{F}_{\Cn}=1\)
\textup{(Corollary~\ref{cor:heterogeneous-impossibility})}.
For~Q3, each invariant family is expressed as a function of the
overlap decomposition
\textup{(Propositions~\ref{prop:overlap-noise}--\ref{prop:overlap-structural})};
set-level invariants depend on the knowledge-base pair alone, while
probabilistic indices are additionally constrained by the channel
kernel.
For~Q4, the broadcast blocklength depends only on the sender's core
and is independent of the number of receivers, provided every
receiver covers the sender's core
\textup{(Theorem~\ref{thm:broadcast-compression})};
a receiver violating this condition becomes a \emph{semantic
bottleneck} whose fidelity limitation persists even over a noiseless
carrier
\textup{(Proposition~\ref{prop:broadcast-bottleneck})}.

\begin{remark}[Classical recovery as a special case]
\label{rem:classical-recovery}
When all agents share the same knowledge base
(\(S_O^{(i)}=S_O^{(j)}\) for all \(i,j\)), the noise pair is trivial,
the two-condition criterion is vacuously satisfied, and all results of
this section reduce to the homogeneous theory of
Section~\ref{subsec:coding} (and, in the irredundant case, to
classical Shannon theory via
Corollary~\ref{cor:irredundant-classical}).
The heterogeneous analysis thus strictly generalizes rather than
replaces the earlier results.
\end{remark}

\subsection{Formal Model and Standing Assumptions}
\label{subsec:app-assumptions}

This subsection formalizes the multi-agent communication scenario
of Section~\ref{subsec:app-problem} within the information model
framework of Sections~\ref{sec:model}--\ref{sec:channel}, and
introduces the \emph{overlap decomposition}---the combinatorial
structure through which all semantic channel invariants are
expressed as functions of the sender--receiver knowledge-base pair.


\begin{definition}[Agent knowledge base]
\label{def:agent-kb}
Fix a set of \(K+1\) agents indexed by
\(i\in\{0,1,\ldots,K\}\).
Each agent~\(i\) is associated with a finite knowledge base
\(S_O^{(i)}\subseteq\mathbb{S}_O\) that is
\(\mathcal L_{\mathrm{sem}}\)-definable in the ambient
structure~\(\mathfrak R\) and satisfies
Assumptions~\ref{assump:finite-so}--\ref{assump:core-extractable}.
The \emph{irredundant core} and \emph{stored shortcuts} of
agent~\(i\) are
\[
  A^{(i)}:=\Atom\bigl(S_O^{(i)}\bigr),
  \qquad
  J^{(i)}:=S_O^{(i)}\setminus A^{(i)}.
\]
By Proposition~\ref{prop:atom-core-correct},
\(\Cn\bigl(A^{(i)}\bigr)=\Cn\bigl(S_O^{(i)}\bigr)\),
the core~\(A^{(i)}\) is irredundant, and
\(S_O^{(i)}\subseteq\Cn\bigl(A^{(i)}\bigr)\).
\end{definition}


\begin{assumption}[Common proof system and ambient universe]
\label{assump:common-ps}
All agents share:
\begin{enumerate}[label=\textup{(CP\arabic*)}]
  \item the same proof system
        \((\mathsf{PS},\,T_{\mathsf{PS}},\,\Cn)\) and inference
        fragment \(\mathcal L_{\mathrm{kb}}\)
        \textup{(Assumption~\ref{assump:proof-system},
        Axiom~\ref{ax:T-operator})};
  \item the same ambient semantic universe \(\mathbb{S}_O\)
        with its injective encoding~\(\enc_O\) and canonical order
        \textup{(Section~\ref{subsec:atomic-derivation},
        Assumption~\ref{assump:semantic-universe})};
  \item the same semantic sublanguage \(\mathcal L_{\mathrm{sem}}\)
        \textup{(Assumption~\ref{assump:semantic-sublanguage})}.
\end{enumerate}
The agents differ \emph{only} in the knowledge bases
\(S_O^{(0)},S_O^{(1)},\ldots,S_O^{(K)}\subseteq\mathbb{S}_O\)
that they store and operate on.
\end{assumption}


\begin{definition}[Pairwise overlap decomposition]
\label{def:overlap-decomposition}
For a fixed sender--receiver pair \((i,j)\) with \(i\neq j\),
define the following subsets of~\(\mathbb{S}_O\).

\smallskip
\noindent\emph{Three-way partition of
\(S_O^{(i)}\cup S_O^{(j)}\):}
\begin{align}
  S_{\cap}^{ij}
  &\;:=\; S_O^{(i)}\cap S_O^{(j)}
  &&\text{(common states)},
  \label{eq:overlap-common}\\[2pt]
  S_{-}^{ij}
  &\;:=\; S_O^{(i)}\setminus S_O^{(j)}
  &&\text{(lost states)},
  \label{eq:overlap-lost}\\[2pt]
  S_{+}^{ij}
  &\;:=\; S_O^{(j)}\setminus S_O^{(i)}
  &&\text{(surplus states)}.
  \label{eq:overlap-surplus}
\end{align}

\noindent\emph{Core partition:}
\begin{align}
  A_{\cap}^{ij}
  &\;:=\; A^{(i)}\cap S_O^{(j)}
  &&\text{(preserved core)},
  \label{eq:overlap-core-preserved}\\[2pt]
  A_{-}^{ij}
  &\;:=\; A^{(i)}\setminus S_O^{(j)}
  &&\text{(lost core)}.
  \label{eq:overlap-core-lost}
\end{align}

\noindent\emph{Surplus stratification:}
\begin{align}
  S_{+,d}^{ij}
  &\;:=\; S_{+}^{ij}\cap\Cn\bigl(S_O^{(i)}\bigr)
  &&\text{(derivable surplus)},
  \label{eq:overlap-surplus-d}\\[2pt]
  S_{+,n}^{ij}
  &\;:=\; S_{+}^{ij}\setminus\Cn\bigl(S_O^{(i)}\bigr)
  &&\text{(non-derivable surplus)}.
  \label{eq:overlap-surplus-n}
\end{align}
\end{definition}

\begin{proposition}[Overlap partition properties]
\label{prop:overlap-partition}
The overlap decomposition of
Definition~\textup{\ref{def:overlap-decomposition}} satisfies:
\begin{enumerate}[label=\textup{(\roman*)}]
  \item \emph{Three-way disjoint union:}\;
        \(S_O^{(i)}\cup S_O^{(j)}
          = S_{-}^{ij}\;\dot\cup\;
            S_{\cap}^{ij}\;\dot\cup\;
            S_{+}^{ij}\).
  \item \emph{Sender decomposition:}\;
        \(S_O^{(i)}
          = S_{\cap}^{ij}\;\dot\cup\; S_{-}^{ij}\).
  \item \emph{Receiver decomposition:}\;
        \(S_O^{(j)}
          = S_{\cap}^{ij}\;\dot\cup\; S_{+}^{ij}\).
  \item \emph{Core partition:}\;
        \(A^{(i)}
          = A_{\cap}^{ij}\;\dot\cup\; A_{-}^{ij}\).
  \item \emph{Surplus partition:}\;
        \(S_{+}^{ij}
          = S_{+,d}^{ij}\;\dot\cup\; S_{+,n}^{ij}\).
  \item \emph{Noise-pair consistency:}\;
        \(S_{-}^{ij}=S_O^{-}\) and \(S_{+}^{ij}=S_O^{+}\),
        where \((S_O^{-},S_O^{+})\) is the end-to-end noise
        pair of
        Definition~\textup{\ref{def:pairwise-scenario}}.
  \item \emph{Core loss refines state loss:}\;
        \(A_{-}^{ij}\subseteq S_{-}^{ij}\).
  \item \emph{Computability:}\;
        All seven sets and their cardinalities are computable
        from the finite knowledge bases
        \(S_O^{(i)},S_O^{(j)}\) under
        Axiom~\textup{\ref{ax:T-operator}}.
\end{enumerate}
\end{proposition}

\begin{proof}
Parts~\textup{(i)}--\textup{(iii)} are standard set partition
identities: \(S_{\cap}^{ij}\), \(S_{-}^{ij}\), \(S_{+}^{ij}\)
are pairwise disjoint by construction, and their union equals
\(S_O^{(i)}\cup S_O^{(j)}\); restricting to \(S_O^{(i)}\)
yields~\textup{(ii)}, and to \(S_O^{(j)}\) yields~\textup{(iii)}.

Part~\textup{(iv)}: since \(A^{(i)}\subseteq S_O^{(i)}\),
\(A^{(i)} = (A^{(i)}\cap S_O^{(j)})
  \;\dot\cup\;(A^{(i)}\setminus S_O^{(j)})
  = A_{\cap}^{ij}\;\dot\cup\; A_{-}^{ij}\).

Part~\textup{(v)}: immediate from the definition of
\(S_{+,d}^{ij}\) and \(S_{+,n}^{ij}\) as complementary
subsets of~\(S_{+}^{ij}\).

Part~\textup{(vi)}: comparing
\eqref{eq:overlap-lost}--\eqref{eq:overlap-surplus} with
\eqref{eq:pairwise-noise-pair} gives the identification
directly.

Part~\textup{(vii)}: \(A_{-}^{ij}
  =A^{(i)}\setminus S_O^{(j)}
  \subseteq S_O^{(i)}\setminus S_O^{(j)}
  =S_{-}^{ij}\), since \(A^{(i)}\subseteq S_O^{(i)}\).

Part~\textup{(viii)}: since \(S_O^{(i)}\) and \(S_O^{(j)}\)
are finite and effectively listable
\textup{(Assumption~\ref{assump:finite-so})}, membership is
decidable by exhaustive comparison.
The sets \(S_{\cap}^{ij}\), \(S_{-}^{ij}\), \(S_{+}^{ij}\),
\(A_{\cap}^{ij}\), \(A_{-}^{ij}\) are then computable by
enumeration.
The surplus stratification requires testing
\(s\in\Cn(S_O^{(i)})\) for each \(s\in S_{+}^{ij}\), which
is decidable by iterating \(T_{\mathsf{PS}}\) from
\(S_O^{(i)}\) until stabilization
\textup{(Axiom~\ref{ax:T-operator}(IC2)--(IC4))}.
\end{proof}

\begin{remark}[Core preservation is tested against
  \(S_O^{(j)}\), not against \(A^{(j)}\)]
\label{rem:core-vs-vocab}
The preserved core
\(A_{\cap}^{ij}=A^{(i)}\cap S_O^{(j)}\) tests whether each
sender core element is \emph{present in the receiver's
vocabulary}---that is, an element of \(S_O^{(j)}\)---not
whether it belongs to the receiver's irredundant core
\(A^{(j)}\).
A sender core element \(a\in A^{(i)}\) may appear in
\(S_O^{(j)}\) as a \emph{redundant} stored shortcut
(i.e., \(a\in J^{(j)}\)); it is still counted as preserved,
because the receiver can produce it as a decoding output
regardless of its redundancy status in~\(S_O^{(j)}\).
\end{remark}

\begin{remark}[Key scalar summaries of the overlap]
\label{rem:three-scalars}
Although the overlap decomposition produces seven subsets, the
\emph{set-level} conditions governing closure fidelity
\textup{(Proposition~\ref{prop:overlap-fidelity})} reduce to two
binary tests:
\(|A_{-}^{ij}|=0\) (no core loss) and
\(|S_{+,n}^{ij}|=0\) (no non-derivable surplus).
The \emph{coding-theoretic} results of
Section~\ref{subsec:app-results} (blocklength, compression
ratio) depend additionally on \(|A^{(i)}|\) and \(C(W_{ij})\),
while the \emph{probabilistic} indices
(\(\Phi_{\Atom}\), \(\Psi_{+}\), \(\mathsf{F}\),
\(\mathsf{E}\)) depend further on the channel kernel
\(\kappa_{\mathrm{sem}}^{ij}\).
The remaining overlap cardinalities are related by simple
accounting:
\(|A_{\cap}^{ij}|=|A^{(i)}|-|A_{-}^{ij}|\),
\(|S_{+}^{ij}|=|S_{+,d}^{ij}|+|S_{+,n}^{ij}|\);
the quantities \(|S_{-}^{ij}|\), \(|S_{\cap}^{ij}|\), and
\(|S_O^{(j)}|\) are mutually determined once any one of them
is known, via
\(|S_{\cap}^{ij}|=|S_O^{(i)}|-|S_{-}^{ij}|\) and
\(|S_O^{(j)}|=|S_{\cap}^{ij}|+|S_{+}^{ij}|\).
\end{remark}


\begin{remark}[Broadcast overlap decomposition]
\label{rem:broadcast-overlap}
In the broadcast scenario
(Definition~\ref{def:broadcast-scenario}), the sender is
agent~\(0\) and the receivers are agents \(1,\ldots,K\).
For each receiver~\(j\), the overlap decomposition
\textup{(Definition~\ref{def:overlap-decomposition})} is
applied to the pair \((0,j)\), yielding receiver-specific
quantities
\(A_{\cap}^{0j}\), \(A_{-}^{0j}\),
\(S_{+,d}^{0j}\), \(S_{+,n}^{0j}\), etc.
The \emph{broadcast core coverage condition}---that
\(A_{-}^{0j}=\varnothing\) for every
\(j\in\{1,\ldots,K\}\)---plays a central role in
Theorem~\ref{thm:broadcast-compression} and
Proposition~\ref{prop:broadcast-bottleneck}.
\end{remark}


\begin{definition}[Heterogeneous semantic channel (formal)]
\label{def:heterogeneous-channel}
For a sender--receiver pair \((i,j)\) with \(i\neq j\), the
\emph{heterogeneous semantic channel} is the semantic channel
\textup{(Definition~\ref{def:semantic-channel})}
\[
  \mathfrak C^{ij}
  \;=\;
  \bigl(\,
    \mathcal I^{(i)},\;
    \mathcal I_{\mathrm{ch}}^{ij},\;
    \mathcal I_{\mathrm{dec}}^{(j)},\;
    \kappa_{\enc},\;
    W_{ij},\;
    D
  \,\bigr),
\]
where the constituent models are defined as follows.
\begin{enumerate}[label=\textup{(\roman*)}]
  \item \emph{Sender information model.}\;
        \(\mathcal I^{(i)}
          =\langle O^{(i)},T_O^{(i)},S_O^{(i)},
                   C,T_C,S_C,R_{\mathcal E}^{(i)}\rangle\)
        is an information model
        \textup{(Definition~\ref{def:info-instance})} with
        semantic state set \(S_O^{(i)}\), carrier state set
        \(S_C\), and enabling map
        \(\mathcal E^{(i)}:S_O^{(i)}\Rightarrow S_C\).
  \item \emph{Carrier channel model.}\;
        \(\mathcal I_{\mathrm{ch}}^{ij}\) is a carrier channel
        information model
        \textup{(cf.\ Definition~\ref{def:semantic-channel}(ii))}
        with input \(S_C\), output \(\hat S_C\), and carrier
        channel kernel
        \(W_{ij}:S_C\rightsquigarrow\hat S_C\).
  \item \emph{Receiver decoding model.}\;
        \(\mathcal I_{\mathrm{dec}}^{(j)}
          =\langle\hat C,T_{\hat C},\hat S_C,
                   \hat O^{(j)},T_{\hat O}^{(j)},S_O^{(j)},
                   R_{\mathcal E}^{\mathrm{dec},(j)}\rangle\)
        is a decoding information model
        \textup{(cf.\ Definition~\ref{def:semantic-channel}(iii))} with
        reconstructed space
        \(\hat S_O:=S_O^{(j)}\subseteq\mathbb{S}_O\) and
        enabling map
        \(\mathcal E_{\mathrm{dec}}^{(j)}:\hat S_C\Rightarrow S_O^{(j)}\).
  \item \emph{Kernels.}\;
        \(\kappa_{\enc}\in\mathcal K(\mathcal I^{(i)})\)
        is the encoding kernel and
        \(D\in\mathcal K(\mathcal I_{\mathrm{dec}}^{(j)})\)
        is the decoding kernel.
\end{enumerate}
The \emph{end-to-end kernel} is
\begin{equation}\label{eq:ksem-ij}
  \kappa_{\mathrm{sem}}^{ij}
  \;:=\;
  D\circ W_{ij}\circ\kappa_{\enc}
  \;:\;
  S_O^{(i)}\rightsquigarrow S_O^{(j)},
\end{equation}
and the \emph{end-to-end noise pair} is
\((S_O^{-},S_O^{+})=(S_{-}^{ij},S_{+}^{ij})\)
\textup{(Proposition~\ref{prop:overlap-partition}(vi))}.
\end{definition}

\begin{remark}[Inherited proof-system structure at the receiver]
\label{rem:receiver-ps-structure}
Since \(S_O^{(j)}\subseteq\mathbb{S}_O\) and the proof system
\(\mathsf{PS}\) acts on all of~\(\mathbb{S}_O\)
\textup{(Assumption~\ref{assump:common-ps})}, the receiver
inherits the deductive closure \(\Cn(S_O^{(j)})\), the
irredundant core \(A^{(j)}=\Atom(S_O^{(j)})\), and the
derivation-depth stratification
\(\Dd(\cdot\mid A^{(j)})\).
The semantic invariants
\(\mathsf{A}(\mathcal I^{(j)})=|A^{(j)}|\) and
\(\mathsf{D_d}(\mathcal I^{(j)})
  =\max_{q\in S_O^{(j)}}\Dd(q\mid A^{(j)})\)
are therefore well-defined and computable
\textup{(Theorem~\ref{thm:semantic-invariants})}.
In general, \(A^{(j)}\neq A^{(i)}\),
\(\Cn(S_O^{(j)})\neq\Cn(S_O^{(i)})\), and
\(\mathsf{D_d}(\mathcal I^{(j)})\neq\mathsf{D_d}(\mathcal I^{(i)})\);
the overlap decomposition
\textup{(Definition~\ref{def:overlap-decomposition})} quantifies
each of these discrepancies.
\end{remark}


\begin{assumption}[Full enabling (heterogeneous setting)]
\label{assump:full-enabling-hetero}
For the heterogeneous semantic channel
\(\mathfrak C^{ij}\) of
Definition~\ref{def:heterogeneous-channel}:
\begin{enumerate}[label=\textup{(FE\arabic*)}]
  \item \emph{Full encoding enabling:}\;
        \(\mathcal E^{(i)}(s_o)=S_C\) for every
        \(s_o\in S_O^{(i)}\).
  \item \emph{Full decoding enabling:}\;
        \(\mathcal E_{\mathrm{dec}}^{(j)}(\hat s_c)=S_O^{(j)}\)
        for every \(\hat s_c\in\hat S_C\).
\end{enumerate}
\end{assumption}

\begin{remark}[Role of the full enabling assumption]
\label{rem:full-enabling-role}
Under Assumption~\ref{assump:full-enabling-hetero}, all
vocabulary-mismatch effects are captured entirely by the noise pair
\((S_{-}^{ij},S_{+}^{ij})\).
This is the heterogeneous counterpart of the full enabling condition
in Theorem~\ref{thm:data-processing}\textup{(iii)}.
When the enabling is constrained, additional capacity reductions
follow from the data-processing bound
\textup{(Theorem~\ref{thm:data-processing}(i))}.
\end{remark}

\begin{remark}[Carrier alphabet size and semantic capacity]
\label{rem:carrier-size}
Assumption~\textup{(SA4)} requires
\(|S_C|\ge\max_i|S_O^{(i)}|\) and
\(|\hat S_C|\ge|S_C|\), ensuring that block codes over
\(S_C^n\) can represent any message set of
size~\(|S_O^{(i)}|\) for large enough~\(n\).
This does \emph{not}, however, imply the single-letter
capacity equality
\(C_{\mathrm{sem}}^{ij}=C(W_{ij})\), which requires the
reverse size condition \(|S_O|\ge|S_C|\) and
\(|\hat S_O|\ge|\hat S_C|\)
\textup{(Theorem~\ref{thm:data-processing}(iii))}.
Under~\textup{(SA4)}, only the data processing bound
\(C_{\mathrm{sem}}^{ij}\le C(W_{ij})\)
\textup{(Theorem~\ref{thm:data-processing}(i))} and the
source entropy bound
\(C_{\mathrm{sem}}^{ij}\le\log|S_O^{(i)}|\)
\textup{(Theorem~\ref{thm:data-processing}(ii))} are
guaranteed.
All achievability and converse results in this section
use~\(C(W_{ij})\) directly via block coding and are
unaffected by the single-letter capacity gap.
\end{remark}


\begin{assumption}[Standing assumptions for
  Section~\ref{sec:application}]
\label{assump:standing-hetero}
Throughout
Sections~\ref{subsec:app-assumptions}--\ref{subsec:app-results},
the following conditions are in force unless explicitly stated
otherwise:
\begin{enumerate}[label=\textup{(SA\arabic*)}]
  \item all standing assumptions of
        Sections~\ref{sec:model}--\ref{sec:channel}, including
        Assumptions~\ref{assump:ordered-structures},
        \ref{assump:semantic-sublanguage},
        \ref{assump:proof-system},
        \ref{assump:finite-so},
        \ref{assump:core-extractable},
        Axiom~\ref{ax:T-operator}, and
        Assumption~\ref{assump:semantic-universe};
  \item the common proof system assumption
        \textup{(Assumption~\ref{assump:common-ps})};
  \item the full enabling assumption
        \textup{(Assumption~\ref{assump:full-enabling-hetero})};
  \item the carrier channel satisfies \(C(W_{ij})>0\) and the
        carrier alphabet sizes satisfy
        \(|S_C|\ge\max_i|S_O^{(i)}|\) and
        \(|\hat S_C|\ge|S_C|\);
  \item the deductive independence of core elements
        \textup{(Assumption~\ref{assump:core-disjoint})} holds
        for the sender's knowledge base \(S_O^{(i)}\) when
        converse bounds are invoked.
\end{enumerate}
\end{assumption}


Table~\ref{tab:notation-iv} collects the notation introduced
in this subsection for convenient reference throughout
Section~\ref{sec:application}.

\begin{table}[t]
\centering
\caption{Notation summary for
  Section~\ref{sec:application}.
  All quantities are defined relative to a fixed
  sender--receiver pair~\((i,j)\).}
\label{tab:notation-iv}
\renewcommand{\arraystretch}{1.2}
\begin{tabularx}{\columnwidth}{@{}l L@{}}
\toprule
\textbf{Symbol} & \textbf{Meaning} \\
\midrule
\(S_O^{(i)}\) &
  Knowledge base (semantic state set) of agent~\(i\) \\
\(A^{(i)},\; J^{(i)}\) &
  Irredundant core and stored shortcuts of agent~\(i\) \\
\(S_{\cap}^{ij}\) &
  Common states:
  \(S_O^{(i)}\cap S_O^{(j)}\) \\
\(S_{-}^{ij}\) &
  Lost states:
  \(S_O^{(i)}\setminus S_O^{(j)}\)
  \((\,=S_O^{-})\) \\
\(S_{+}^{ij}\) &
  Surplus states:
  \(S_O^{(j)}\setminus S_O^{(i)}\)
  \((\,=S_O^{+})\) \\
\(A_{\cap}^{ij}\) &
  Preserved sender core:
  \(A^{(i)}\cap S_O^{(j)}\) \\
\(A_{-}^{ij}\) &
  Lost sender core:
  \(A^{(i)}\setminus S_O^{(j)}\) \\
\(S_{+,d}^{ij}\) &
  Derivable surplus:
  \(S_{+}^{ij}\cap\Cn(S_O^{(i)})\) \\
\(S_{+,n}^{ij}\) &
  Non-derivable surplus:
  \(S_{+}^{ij}\setminus\Cn(S_O^{(i)})\) \\
\(\mathfrak C^{ij}\) &
  Heterogeneous semantic channel from~\(i\) to~\(j\) \\
\(W_{ij}\) &
  Carrier channel kernel from~\(i\) to~\(j\) \\
\(\kappa_{\mathrm{sem}}^{ij}\) &
  End-to-end semantic kernel for pair \((i,j)\) \\
\bottomrule
\end{tabularx}
\end{table}

\subsection{Instantiation of Semantic Channel Invariants}
\label{subsec:app-theory}

This subsection applies the invariant machinery of
Section~\ref{sec:channel} to the heterogeneous pair~\((i,j)\),
expressing each invariant of
Theorem~\ref{thm:invariant-summary} as a function of the
overlap decomposition of Section~\ref{subsec:app-assumptions}.
Throughout, we fix a sender--receiver pair \((i,j)\) with the
heterogeneous semantic channel \(\mathfrak C^{ij}\) of
Definition~\ref{def:heterogeneous-channel}, under the standing
assumptions of Assumption~\ref{assump:standing-hetero}.


\begin{proposition}[Set-level invariants from overlap]
\label{prop:overlap-noise}
For the heterogeneous semantic channel \(\mathfrak C^{ij}\)
with noise pair
\((S_O^{-},S_O^{+})=(S_{-}^{ij},S_{+}^{ij})\)
\textup{(Proposition~\ref{prop:overlap-partition}(vi))}:
\begin{enumerate}[label=\textup{(\roman*)}]
  \item \emph{Preserved region:}\;
        \(\tilde S_O^{\cap}
          =S_O^{(i)}\cap S_O^{(j)}
          =S_{\cap}^{ij}\).
  \item \emph{Core preservation ratio:}\;
        \begin{equation}\label{eq:rho-overlap}
          \rho_{\Atom}\bigl(S_O^{(i)},S_O^{(j)}\bigr)
          \;=\;
          \frac{|A_{\cap}^{ij}|}{|A^{(i)}|}
          \;=\;
          1-\frac{|A_{-}^{ij}|}{|A^{(i)}|}.
        \end{equation}
        In particular, \(\rho_{\Atom}=1\) if and only if
        \(A_{-}^{ij}=\varnothing\).
  \item \emph{Spurious derivability:}\;
        \(S_O^{+}\subseteq\Cn(S_O^{(i)})\) if and only if
        \(S_{+,n}^{ij}=\varnothing\).
\end{enumerate}
\end{proposition}

\begin{proof}
\textup{(i)}:\
\(S_O^{-}=S_{-}^{ij}\)
\textup{(Proposition~\ref{prop:overlap-partition}(vi))}, so
\(\tilde S_O^{\cap}
  =S_O^{(i)}\setminus S_O^{-}
  =S_O^{(i)}\setminus(S_O^{(i)}\setminus S_O^{(j)})
  =S_O^{(i)}\cap S_O^{(j)}
  =S_{\cap}^{ij}\).

\smallskip\noindent
\textup{(ii)}:\
By Definition~\ref{def:atom-preservation},
\(\rho_{\Atom}
  =|A^{(i)}\cap S_O^{(j)}|/|A^{(i)}|
  =|A_{\cap}^{ij}|/|A^{(i)}|\).
By Proposition~\ref{prop:overlap-partition}(iv),
\(|A_{\cap}^{ij}|=|A^{(i)}|-|A_{-}^{ij}|\).

\smallskip\noindent
\textup{(iii)}:\
\(S_O^{+}=S_{+}^{ij}\) and
\(S_{+}^{ij}\subseteq\Cn(S_O^{(i)})\)
iff
\(S_{+}^{ij}\setminus\Cn(S_O^{(i)})=\varnothing\)
iff \(S_{+,n}^{ij}=\varnothing\).
\end{proof}


\begin{proposition}[Closure fidelity: necessary and sufficient
  conditions]
\label{prop:overlap-fidelity}
For the sender--receiver pair~\((i,j)\),
\[
  \mathsf{F}_{\Cn}\bigl(S_O^{(i)},\,S_O^{(j)}\bigr)=1
  \quad\Longleftrightarrow\quad
  \Cn\bigl(S_O^{(i)}\bigr)=\Cn\bigl(S_O^{(j)}\bigr),
\]
and this holds if and only if both of the following conditions
are satisfied:
\begin{enumerate}[label=\textup{(F\arabic*)}]
  \item \emph{Sender core derivable from receiver:}\;
        \(A^{(i)}\subseteq\Cn\bigl(S_O^{(j)}\bigr)\).
  \item \emph{No non-derivable surplus:}\;
        \(S_{+,n}^{ij}=\varnothing\)
        \textup{(}equivalently,
        \(S_{+}^{ij}\subseteq\Cn(S_O^{(i)})\)\textup{)}.
\end{enumerate}
\end{proposition}

\begin{proof}
The first equivalence is the definition of
\(\mathsf{F}_{\Cn}\)
\textup{(Definition~\ref{def:closure-fidelity}):}
\(\mathsf{F}_{\Cn}=1\) iff the Jaccard index of the two
closures equals~\(1\), i.e., the closures coincide.
It remains to show
\(\Cn(S_O^{(i)})=\Cn(S_O^{(j)})\) iff
\textup{(F1)}+\textup{(F2)}.

\smallskip\noindent
\emph{Sufficiency.}\;
From~\textup{(F1)}:
\(A^{(i)}\subseteq\Cn(S_O^{(j)})\).
By monotonicity~\textup{(Cn2)},
\(\Cn(A^{(i)})\subseteq\Cn(\Cn(S_O^{(j)}))
  =\Cn(S_O^{(j)})\)
\textup{(idempotence (Cn3))}. 
Since
\(\Cn(A^{(i)})=\Cn(S_O^{(i)})\)
\textup{(Proposition~\ref{prop:atom-core-correct}(i))},
\(\Cn(S_O^{(i)})\subseteq\Cn(S_O^{(j)})\).

From~\textup{(F2)}:
\(S_{+}^{ij}\subseteq\Cn(S_O^{(i)})\).
Since
\(S_{\cap}^{ij}\subseteq S_O^{(i)}
  \subseteq\Cn(S_O^{(i)})\)
\textup{(reflexivity (Cn1))},
\(S_O^{(j)}=S_{\cap}^{ij}\cup S_{+}^{ij}
  \subseteq\Cn(S_O^{(i)})\).
By monotonicity and idempotence,
\(\Cn(S_O^{(j)})\subseteq\Cn(S_O^{(i)})\).

Combining the two inclusions yields
\(\Cn(S_O^{(i)})=\Cn(S_O^{(j)})\).

\smallskip\noindent
\emph{Necessity.}\;
Suppose \(\Cn(S_O^{(i)})=\Cn(S_O^{(j)})\).

For~\textup{(F1)}:
\(A^{(i)}\subseteq\Cn(A^{(i)})
  =\Cn(S_O^{(i)})=\Cn(S_O^{(j)})\).

For~\textup{(F2)}:
\(S_{+}^{ij}\subseteq S_O^{(j)}
  \subseteq\Cn(S_O^{(j)})=\Cn(S_O^{(i)})\),
so \(S_{+,n}^{ij}
  =S_{+}^{ij}\setminus\Cn(S_O^{(i)})=\varnothing\).
\end{proof}

\begin{remark}[Strong vs.\ weak core coverage]
\label{rem:strong-weak-coverage}
Condition~\textup{(F1)} requires that each sender core element
be \emph{derivable from} the receiver's knowledge base; it
does \emph{not} require the element to be literally present
in~\(S_O^{(j)}\).
A strictly stronger condition is
\begin{enumerate}[label=\textup{(F1\('\))}]
  \item \(A_{-}^{ij}=\varnothing\), i.e.,
        \(A^{(i)}\subseteq S_O^{(j)}\).
\end{enumerate}
Condition~\textup{(F1\({}'\))} implies~\textup{(F1)} (since
\(S_O^{(j)}\subseteq\Cn(S_O^{(j)})\)) but not conversely:
a core element \(a\in A_{-}^{ij}\) may satisfy
\(a\in\Cn(S_O^{(j)})\) even though \(a\notin S_O^{(j)}\).

For \emph{set-level} closure fidelity, the weak
condition~\textup{(F1)} is both necessary and sufficient
\textup{(Proposition~\ref{prop:overlap-fidelity})}.
For \emph{operational} closure reliability via the two-layer
code of Theorem~\textup{\ref{thm:achievability}(ii)}, the
decoder outputs core elements directly and therefore requires
the strong condition~\textup{(F1\({}'\))} so that
\(A^{(i)}\subseteq S_O^{(j)}=\hat S_O\).
When only~\textup{(F1)} holds with
\(A_{-}^{ij}\neq\varnothing\), a more sophisticated decoding
strategy is needed; this is addressed in
Section~\ref{subsec:app-results}.
Throughout the remainder of this subsection, results are stated
under whichever version is required, with the distinction noted
explicitly.
\end{remark}

\begin{corollary}[Sufficient condition via overlap scalars]
\label{cor:overlap-sufficient}
If \(A_{-}^{ij}=\varnothing\) and
\(S_{+,n}^{ij}=\varnothing\), then
\(\mathsf{F}_{\Cn}(S_O^{(i)},S_O^{(j)})=1\).
This is the instantiation of
Proposition~\textup{\ref{prop:noise-fidelity}(iii)} in the
overlap language and the condition used in the achievability
results of Section~\textup{\ref{subsec:app-results}}.
\end{corollary}

\begin{proof}
\(A_{-}^{ij}=\varnothing\) gives
\(A^{(i)}\subseteq S_O^{(j)}\subseteq\Cn(S_O^{(j)})\), so
\textup{(F1)} holds.
\(S_{+,n}^{ij}=\varnothing\) is \textup{(F2)}.
Apply Proposition~\ref{prop:overlap-fidelity}.
\end{proof}


\begin{proposition}[Noise-pair probabilistic indices from
  overlap]
\label{prop:overlap-noise-indices}
Let \(\mathfrak C^{ij}\) be a heterogeneous semantic channel
with kernel
\(\kappa_{\mathrm{sem}}^{ij}
  :S_O^{(i)}\rightsquigarrow S_O^{(j)}\).
\begin{enumerate}[label=\textup{(\roman*)}]
  \item \emph{Core preservation index:}\;
        \begin{equation}\label{eq:Phi-overlap}
          \Phi_{\Atom}(\mathfrak C^{ij})
          \;=\;
          \begin{cases}
            \displaystyle\min_{a\in A^{(i)}}
              \kappa_{\mathrm{sem}}^{ij}(a\mid a)
            & \text{if } A_{-}^{ij}=\varnothing,\\[6pt]
            0 & \text{if } A_{-}^{ij}\neq\varnothing.
          \end{cases}
        \end{equation}
  \item \emph{Spurious probability index:}\;
        \begin{equation}\label{eq:Psi-overlap}
          \Psi_{+}(\mathfrak C^{ij})
          \;=\;
          \max_{s_o\in S_O^{(i)}}\;
          \sum_{\hat s_o\in S_{+}^{ij}}
          \kappa_{\mathrm{sem}}^{ij}(\hat s_o\mid s_o).
        \end{equation}
        In particular, \(\Psi_{+}=0\) whenever
        \(S_{+}^{ij}=\varnothing\)
        \textup{(}i.e., \(S_O^{(j)}\subseteq S_O^{(i)}\)\textup{)}.
  \item \emph{Noiseless deterministic case:}\;
        If \(W_{ij}=\mathrm{id}_{S_C}\) and both
        \(\kappa_{\enc}\) and \(D\) are deterministic with
        induced end-to-end function
        \(f:=D\circ\mathrm{id}\circ\kappa_{\enc}
          :S_O^{(i)}\to S_O^{(j)}\), then
        \begin{align}
          \Phi_{\Atom}(\mathfrak C^{ij})
          &=
          \begin{cases}
            1 & \text{if } A_{-}^{ij}=\varnothing \\
              &  \text{ and } f(a)=a \;\forall a\in A^{(i)},\\
            0 & \text{otherwise},
          \end{cases}
          \label{eq:Phi-noiseless}\\[4pt]
          \Psi_{+}(\mathfrak C^{ij})
          &=
          \begin{cases}
            0 & \text{if } f(S_O^{(i)})\subseteq S_{\cap}^{ij},\\
            1 & \text{otherwise}.
          \end{cases}
          \label{eq:Psi-noiseless}
        \end{align}
\end{enumerate}
\end{proposition}

\begin{proof}
\textup{(i)}:\
By Proposition~\ref{prop:overlap-partition}(vi),
\(A^{(i)}\cap S_O^{-}=A^{(i)}\cap S_{-}^{ij}=A_{-}^{ij}\).
When \(A_{-}^{ij}\neq\varnothing\),
Definition~\ref{def:noise-pair-indices-def} sets
\(\Phi_{\Atom}:=0\).
When \(A_{-}^{ij}=\varnothing\),
\(A^{(i)}\subseteq S_O^{(j)}=\tilde S_O\), so
\(\pi(a)=\kappa_{\mathrm{sem}}^{ij}(a\mid a)\) is
well-defined for every \(a\in A^{(i)}\) and
\(\Phi_{\Atom}=\min_{a}\pi(a)\).

\smallskip\noindent
\textup{(ii)}:\
Direct from Definition~\ref{def:noise-pair-indices-def} with
\(S_O^{+}=S_{+}^{ij}\).
When \(S_{+}^{ij}=\varnothing\), the sum is empty and
\(\Psi_{+}=0\).

\smallskip\noindent
\textup{(iii)}:\
When \(\kappa_{\mathrm{sem}}^{ij}\) is deterministic,
\(\kappa_{\mathrm{sem}}^{ij}(\hat s_o\mid s_o)
  \in\{0,1\}\) for all \(s_o,\hat s_o\).
For \(\Phi_{\Atom}\):
\(\pi(a)=\kappa_{\mathrm{sem}}^{ij}(a\mid a)=1\) iff
\(f(a)=a\), and this must hold for all \(a\in A^{(i)}\)
(which requires \(a\in S_O^{(j)}\), i.e.,
\(A_{-}^{ij}=\varnothing\)).
For \(\Psi_{+}\):
\(p_{+}(s_o)=\mathbf{1}[f(s_o)\in S_{+}^{ij}]\), which
is~\(0\) for all \(s_o\) iff
\(f(S_O^{(i)})\subseteq S_{\cap}^{ij}\), and~\(1\) for some
\(s_o\) otherwise; the maximum is therefore \(0\) or~\(1\).
\end{proof}

\begin{remark}[Kernel dependence of probabilistic indices]
\label{rem:kernel-dependence}
Unlike the set-level invariants
\(\rho_{\Atom}\) and \(\mathsf{F}_{\Cn}\), which depend only
on the knowledge-base pair
\((S_O^{(i)},S_O^{(j)})\), the probabilistic indices
\(\Phi_{\Atom}\) and \(\Psi_{+}\) depend on the channel
kernel \(\kappa_{\mathrm{sem}}^{ij}\) and hence on the
specific encoder, physical channel, and decoder.
The overlap decomposition constrains the \emph{range} of these
indices---\(\Phi_{\Atom}=0\) whenever
\(A_{-}^{ij}\neq\varnothing\), and
\(\Psi_{+}=0\) whenever
\(S_{+}^{ij}=\varnothing\)---but their precise values within
the feasible range are determined by the kernel.
\end{remark}


\begin{proposition}[Structural quality indices from overlap]
\label{prop:overlap-quality}
Let \(\mathfrak C^{ij}\) be a heterogeneous semantic channel
with kernel
\(\kappa_{\mathrm{sem}}^{ij}
  :S_O^{(i)}\rightsquigarrow S_O^{(j)}\)
and let \(A=A^{(i)}\).
\begin{enumerate}[label=\textup{(\roman*)}]
  \item \emph{Fidelity concentration under core-preserving
        overlap:}\;
        If \(A_{-}^{ij}=\varnothing\) and
        \(S_{+,n}^{ij}=\varnothing\), then by
        Corollary~\textup{\ref{cor:fidelity-core-concentration}},
        \begin{equation}\label{eq:F-overlap}
          \mathsf{F}(\mathfrak C^{ij})
          \;=\;
          1-\max_{a\in A^{(i)}}
            \bar d_{\Cn}(a\mid\mathfrak C^{ij}).
        \end{equation}
        That is, the worst-case closure distortion is attained
        at a sender core element; redundant states contribute
        zero.
  \item \emph{Depth expansion under vocabulary match:}\;
        If \(S_{-}^{ij}=S_{+}^{ij}=\varnothing\)
        \textup{(}i.e., \(S_O^{(i)}=S_O^{(j)}\)\textup{)},
        then every reachable \(\hat s_o\in S_O^{(j)}=S_O^{(i)}\)
        lies in~\(\Cn(A)\) and the depth distortion reduces to
        its first branch
        \textup{(Definition~\ref{def:depth-distortion})}.
        In particular,
        \(\mathsf{E}(\mathfrak C^{ij})=0\) if and only if
        \(\Dd(\hat s_o\mid A)=\Dd(s_o\mid A)\) holds
        \(\kappa_{\mathrm{sem}}^{ij}\)-almost surely for
        every \(s_o\in S_O^{(i)}\).
\end{enumerate}
\end{proposition}

\begin{proof}
\textup{(i)}:\
By Corollary~\ref{cor:overlap-sufficient},
\(\Cn(S_O^{(i)})=\Cn(S_O^{(j)})\) and
\(A^{(i)}\cap S_O^{-}=A_{-}^{ij}=\varnothing\).
Proposition~\ref{prop:noise-distortion-bounds} gives
\(d_{\Cn}(s_o,\hat s_o\mid S_O^{(i)})=0\) for every
\(s_o\in S_O^{(i)}\setminus A\) and
\(\hat s_o\in S_O^{(j)}\).
Hence
\(\max_{s_o\in S_O^{(i)}}
  \bar d_{\Cn}(s_o\mid\mathfrak C^{ij})
  =\max_{a\in A}
  \bar d_{\Cn}(a\mid\mathfrak C^{ij})\),
and the conclusion follows from
Definition~\ref{def:fidelity-index}.

\smallskip\noindent
\textup{(ii)}:\
When \(S_O^{(i)}=S_O^{(j)}\), the noise pair is trivial and
every \(\hat s_o\in S_O^{(j)}=S_O^{(i)}\) satisfies
\(\hat s_o\in\Cn(A)\)
\textup{(Proposition~\ref{prop:atom-core-correct}(iv))}.
The claim follows from
Proposition~\ref{prop:noise-distortion-bounds} and
Definition~\ref{def:fidelity-index}.
\end{proof}


\begin{proposition}[Receiver-side structural comparison from
  overlap]
\label{prop:overlap-structural}
Let \(\mathfrak C^{ij}\) be a heterogeneous semantic channel
with \(\tilde S_O=S_O^{(j)}\).
\begin{enumerate}[label=\textup{(\roman*)}]
  \item \emph{Atomicity shift:}\;
        \(\Delta\mathsf{A}(\mathfrak C^{ij})
          =|A^{(j)}|-|A^{(i)}|\).
  \item \emph{Depth shift:}\;
        \(\Delta\mathsf{D_d}(\mathfrak C^{ij})
          =\mathsf{D_d}(\mathcal I^{(j)})
          -\mathsf{D_d}(\mathcal I^{(i)})\).
  \item \emph{Under core-preserving conditions
        \textup{(}\(A_{-}^{ij}=\varnothing\),
        \(S_{+,n}^{ij}=\varnothing\)\textup{):}}
        \begin{enumerate}[label=\textup{(\alph*)}]
          \item \(\Cn(S_O^{(j)})=\Cn(S_O^{(i)})\)
                \textup{(Corollary~\ref{cor:overlap-sufficient})}.
          \item The set~\(A^{(i)}\) is an irredundant
                generating subset of~\(S_O^{(j)}\)
                for \(\Cn(S_O^{(j)})=\Cn(S_O^{(i)})\)
                \textup{(Proposition~\ref{prop:structural-comparison}(ii))}.
                In general,
                \(|A^{(j)}|\) may differ
                from~\(|A^{(i)}|\): surplus states
                in~\(S_{+,d}^{ij}\) can make sender core
                elements redundant in the receiver's
                vocabulary, potentially changing the
                canonical core size in either direction.
          \item Equality
                \(\Delta\mathsf{A}=0\) holds when
                \(S_{+,d}^{ij}=\varnothing\)
                \textup{(}equivalently,
                \(S_{+}^{ij}=\varnothing\), i.e.,
                \(S_O^{(j)}\subseteq S_O^{(i)}\)\textup{)},
                because then
                \(S_O^{(j)}=A^{(i)}\cup
                  (J^{(i)}\setminus S_{-}^{ij})\)
                \textup{(}noting that
                \(A_{-}^{ij}=\varnothing\) already implies
                \(S_{-}^{ij}\subseteq J^{(i)}\)\textup{)}
                and \(A^{(i)}\) remains irredundant
                in~\(S_O^{(j)}\).
        \end{enumerate}
  \item \emph{Trivial noise pair
        \textup{(}\(S_O^{(i)}=S_O^{(j)}\)\textup{):}}\;
        \(\Delta\mathsf{A}=0\) and
        \(\Delta\mathsf{D_d}=0\)
        \textup{(Proposition~\ref{prop:structural-comparison}(i))}.
\end{enumerate}
\end{proposition}

\begin{proof}
Parts~\textup{(i)} and~\textup{(ii)} are direct from
Definition~\ref{def:receiver-structural} with
\(\tilde S_O=S_O^{(j)}\).

\smallskip\noindent
\textup{(iii)(a)}:
Corollary~\ref{cor:overlap-sufficient}.

\smallskip\noindent
\textup{(iii)(b)}:
By Proposition~\ref{prop:structural-comparison}(ii), applied
with \(A=A^{(i)}\), \(S_O^{-}=S_{-}^{ij}\),
\(S_O^{+}=S_{+}^{ij}\).
When \(S_{+,d}^{ij}\neq\varnothing\), some
\(d\in S_{+,d}^{ij}\subseteq\Cn(A^{(i)})\) may make a
previously irredundant \(a\in A^{(i)}\) redundant
in~\(S_O^{(j)}\) (if
\(a\in\Cn\bigl((S_O^{(j)}\setminus\{a\})\bigr)\) due to
the presence of~\(d\)).

\smallskip\noindent
\textup{(iii)(c)}:\
When \(S_{+,d}^{ij}=\varnothing\), together with
\(S_{+,n}^{ij}=\varnothing\) \textup{(}from the
hypothesis of~\textup{(iii))}, we have
\(S_{+}^{ij}=\varnothing\) and
\(S_O^{(j)}\subseteq S_O^{(i)}\).
Since \(A_{-}^{ij}=\varnothing\),
\(A^{(i)}\subseteq S_O^{(j)}\) and
\(S_{-}^{ij}\subseteq J^{(i)}\)
\textup{(}because
\(A^{(i)}\cap S_{-}^{ij}=A_{-}^{ij}=\varnothing\)\textup{)}.
Hence
\(S_O^{(j)}=A^{(i)}\cup(J^{(i)}\setminus S_{-}^{ij})\).

We show \(\Atom(S_O^{(j)})=A^{(i)}\) by a simultaneous induction on
the canonical order of~\(S_O^{(j)}\).
Let \(s_1<s_2<\cdots<s_m\) be the elements of~\(S_O^{(j)}\) in
canonical order, and let \(B_k\) denote the current set after
the irredundantization procedure
\textup{(Definition~\ref{def:atom-so})} has scanned
\(s_1,\ldots,s_k\).

\smallskip
\noindent\emph{Inductive claim.}\;
After scanning \(s_1,\ldots,s_k\):
\textup{(a)}~every \(a\in A^{(i)}\) with \(a\le s_k\) has
survived (remains in~\(B_k\)); and
\textup{(b)}~every \(j'\in J^{(i)}\setminus S_{-}^{ij}\) with
\(j'\le s_k\) has been removed.

The base case \(k=0\) is vacuous.
For the inductive step, suppose the claim holds through~\(s_k\)
and consider~\(s_{k+1}\).

\smallskip
\emph{Case~1: \(s_{k+1}=a\in A^{(i)}\).}\;
By the inductive hypothesis, every element of
\(J^{(i)}\setminus S_{-}^{ij}\) preceding~\(a\) has been
removed, and every element of~\(A^{(i)}\) preceding~\(a\) has
survived.
Hence the current set satisfies
\[
  B_k\setminus\{a\}
  \;=\;
  \bigl(A^{(i)}\cap\{s_1,\ldots,s_k\}\bigr)
  \;\cup\;
  \{s_{k+2},\ldots,s_m\}.
\]
In the irredundantization of~\(S_O^{(i)}\), when \(a\) was
scanned, the current set~\(B^{(i)}\) had precisely the same
structure---surviving core elements before~\(a\) plus all
elements of~\(S_O^{(i)}\) after~\(a\)---because elements of
\(J^{(i)}\) before~\(a\) were likewise removed at their own
scan steps (they are in~\(J^{(i)}\) by definition).
Since \(\{s_{k+2},\ldots,s_m\}\subseteq
\{\text{elements of }S_O^{(i)}\text{ after }a\}\)
(the former is a subset, possibly missing elements
of~\(S_{-}^{ij}\) after~\(a\)),
\(B_k\setminus\{a\}\subseteq B^{(i)}\setminus\{a\}\).
By monotonicity~\textup{(Cn2)},
\(\Cn(B_k\setminus\{a\})\subseteq\Cn(B^{(i)}\setminus\{a\})\).
Since \(a\) survived in the irredundantization of~\(S_O^{(i)}\),
\(a\notin\Cn(B^{(i)}\setminus\{a\})\), hence
\(a\notin\Cn(B_k\setminus\{a\})\), and \(a\) survives.

\smallskip
\emph{Case~2: \(s_{k+1}=j'\in J^{(i)}\setminus S_{-}^{ij}\).}\;
By the inductive hypothesis,
\(A^{(i)}\cap\{s_1,\ldots,s_k\}\subseteq B_k\), and all
elements of~\(A^{(i)}\) after~\(s_k\) are in~\(B_k\)
(not yet scanned).
Hence \(A^{(i)}\subseteq B_k\setminus\{j'\}\)
(noting \(j'\notin A^{(i)}\)).
By monotonicity,
\(\Cn(A^{(i)})\subseteq\Cn(B_k\setminus\{j'\})\).
Since
\(j'\in J^{(i)}\subseteq S_O^{(i)}\subseteq\Cn(A^{(i)})\)
\textup{(Proposition~\ref{prop:atom-core-correct}(iv))},
\(j'\in\Cn(B_k\setminus\{j'\})\), and \(j'\) is removed.

\smallskip
By induction, the output of the irredundantization
of~\(S_O^{(j)}\) retains exactly~\(A^{(i)}\).
For the reverse inclusion
\(\Atom(S_O^{(j)})\subseteq A^{(i)}\):
suppose for contradiction that
\(b\in\Atom(S_O^{(j)})\setminus A^{(i)}\).
Then \(b\in J^{(i)}\setminus S_{-}^{ij}\), but the induction
shows that every such element is removed---a contradiction.
Hence \(\Atom(S_O^{(j)})=A^{(i)}\) and
\(\Delta\mathsf{A}=0\).

\smallskip\noindent
\textup{(iv)}:
Immediate from
Proposition~\ref{prop:structural-comparison}(i).
\end{proof}


\begin{corollary}[Heterogeneous semantic Fano bound]
\label{cor:heterogeneous-fano}
Let \(P_O\in\Delta(S_O^{(i)})\) be full-support and let
\(\epsilon:=\bar d_H(\mathfrak C^{ij},P_O)\).
Then
\begin{equation}\label{eq:fano-hetero}
  I_{\mathrm{sem}}^{ij}(P_O,\mathfrak C^{ij})
  \;\ge\;
  H(\mathsf{S}_o)
  -h_b(\epsilon)
  -\epsilon\log\bigl(|S_O^{(i)}|-1\bigr),
\end{equation}
where \(h_b\) is the binary entropy.
Moreover, by
Proposition~\textup{\ref{prop:noise-pair-indices}(iii)},
\begin{equation}\label{eq:eps-hetero}
  \epsilon
  \;\le\;
  1-\Phi_{\Atom}(\mathfrak C^{ij})\cdot P_O(A^{(i)}).
\end{equation}
Consequently, when \(A_{-}^{ij}=\varnothing\) and
\(\Phi_{\Atom}(\mathfrak C^{ij})\) is close to~\(1\)
\textup{(}e.g., because the carrier channel is
reliable\textup{)}, the right-hand side
of~\eqref{eq:fano-hetero} is close to
\(H(\mathsf{S}_o)\), i.e., nearly all source entropy is
transmitted.
When \(A_{-}^{ij}\neq\varnothing\),
\(\Phi_{\Atom}=0\) by~\eqref{eq:Phi-overlap} and the bound
yields only the trivial lower bound
\(I_{\mathrm{sem}}^{ij}\ge 0\).
\end{corollary}

\begin{proof}
Equation~\eqref{eq:fano-hetero} is
Theorem~\ref{thm:semantic-fano} applied with
source alphabet~\(S_O^{(i)}\) and
\(\tilde S_O=S_O^{(j)}\).
Equation~\eqref{eq:eps-hetero} is
Proposition~\ref{prop:noise-pair-indices}(iii).
When \(A_{-}^{ij}\neq\varnothing\),
\(\Phi_{\Atom}=0\) gives \(\epsilon\le 1\), so
\(h_b(\epsilon)+\epsilon\log(|S_O^{(i)}|-1)
  \le\log|S_O^{(i)}|\)
and the lower bound cannot exceed zero in a nontrivial way.
\end{proof}


\begin{proposition}[Semantic capacity from overlap]
\label{prop:overlap-capacity}
Let \(W_{ij}:S_C\rightsquigarrow\hat S_C\) be the carrier
channel kernel for the pair~\((i,j)\).
Under the full enabling assumption
\textup{(Assumption~\ref{assump:full-enabling-hetero})} and
\textup{(SA4)} of
Assumption~\textup{\ref{assump:standing-hetero}}:
\begin{enumerate}[label=\textup{(\roman*)}]
  \item \emph{Data processing chain:}\;
        \(I_{\mathrm{sem}}^{ij}
          \le C_{\mathrm{sem}}^{ij}(W_{ij})
          \le C(W_{ij})\).
  \item \emph{Source entropy bound:}\;
        \(C_{\mathrm{sem}}^{ij}(W_{ij})
          \le\log|S_O^{(i)}|\).
  \item \emph{Capacity equality
        \textup{(}reverse size condition\textup{):}}\;
        If additionally
        \(|S_O^{(i)}|\ge|S_C|\) and
        \(|S_O^{(j)}|\ge|\hat S_C|\), then
        \(C_{\mathrm{sem}}^{ij}(W_{ij})=C(W_{ij})\).
  \item \emph{Mutual information bound:}\;
        \(I_{\mathrm{sem}}^{ij}
          \le\min\bigl(\log|S_O^{(i)}|,\;
                       \log|S_O^{(j)}|\bigr)\).
\end{enumerate}
\end{proposition}

\begin{proof}
Part~\textup{(i)} is
Theorem~\ref{thm:data-processing}(i).
Part~\textup{(ii)} is
Theorem~\ref{thm:data-processing}(ii).
Part~\textup{(iii)} follows from
Theorem~\ref{thm:data-processing}(iii) under the stated
size conditions, which ensure the existence of a
deterministic surjection
\(f:S_O^{(i)}\to S_C\) and a deterministic injection
\(g:\hat S_C\to S_O^{(j)}\).
Part~\textup{(iv)} follows from
\(I(\mathsf{S}_o;\hat{\mathsf{S}}_o)
  \le\min(H(\mathsf{S}_o),\,H(\hat{\mathsf{S}}_o))\).
\end{proof}


\begin{remark}[Invariant--overlap correspondence]
\label{rem:invariant-overlap-summary}
The six invariant families of
Theorem~\ref{thm:invariant-summary} partition into three tiers
of dependence on the overlap decomposition.
\emph{Tier~1 (knowledge-base pair only):}
the set-level invariants \(\rho_{\Atom}\), \(\mathsf{F}_{\Cn}\)
\textup{(Proposition~\ref{prop:overlap-noise})} and the
structural comparison indices \(\Delta\mathsf{A}\),
\(\Delta\mathsf{D_d}\)
\textup{(Proposition~\ref{prop:overlap-structural})} are fully
determined by \((S_O^{(i)},S_O^{(j)})\) and the proof
system~\(\mathsf{PS}\).
\emph{Tier~2 (knowledge-base pair \(+\) channel kernel):}
the noise-pair indices \(\Phi_{\Atom}\), \(\Psi_{+}\)
\textup{(Proposition~\ref{prop:overlap-noise-indices})} and the
quality indices \(\mathsf{F}\), \(\mathsf{E}\)
\textup{(Proposition~\ref{prop:overlap-quality})} are
\emph{constrained} by the overlap scalars
(e.g., \(\Phi_{\Atom}=0\) whenever
\(A_{-}^{ij}\neq\varnothing\);
\(\Psi_{+}=0\) whenever \(S_{+}^{ij}=\varnothing\))
but additionally depend on~\(\kappa_{\mathrm{sem}}^{ij}\).
\emph{Tier~3 (kernel \(+\) carrier channel):}
the information-theoretic invariants
\(I_{\mathrm{sem}}^{ij}\), \(C_{\mathrm{sem}}^{ij}\), \(C(W_{ij})\);
under full enabling and~\textup{(SA4)},
\(C_{\mathrm{sem}}^{ij}\le
  \min\bigl(C(W_{ij}),\,\log|S_O^{(i)}|\bigr)\)
\textup{(Proposition~\ref{prop:overlap-capacity}(i)--(ii))},
with equality under the reverse size condition
of~\textup{(iii)}.
\end{remark}

\begin{remark}[Diagnostic use]
\label{rem:diagnostic-use}
The three-tier correspondence enables a two-stage diagnostic
workflow: first, compute the overlap decomposition offline and test
the binary conditions \(A_{-}^{ij}=\varnothing\) and
\(S_{+,n}^{ij}=\varnothing\) to determine whether
\(\mathsf{F}_{\Cn}=1\) is achievable
\textup{(Corollary~\ref{cor:overlap-sufficient})};
second, given a specific channel kernel, compute
\(\Phi_{\Atom}\), \(\Psi_{+}\), \(\mathsf{F}\), \(\mathsf{E}\)
and apply the Fano bound
\textup{(Corollary~\ref{cor:heterogeneous-fano})} to obtain a
lower bound on~\(I_{\mathrm{sem}}^{ij}\).
\end{remark}

\subsection{Main Results: Heterogeneous Compression and Broadcast}
\label{subsec:app-results}

This subsection derives the main analytical results for the
heterogeneous setting.
Part~1 addresses the pairwise unicast scenario
(Definition~\ref{def:pairwise-scenario}):
closure-reliable achievability, a heterogeneous deductive
compression theorem, an impossibility result when core coverage
fails, and a vocabulary design criterion.
Part~2 extends the theory to the broadcast scenario
(Definition~\ref{def:broadcast-scenario}).
Throughout, the standing assumptions of
Assumption~\ref{assump:standing-hetero} are in force.

\smallskip
\noindent\textit{Two notions of closure fidelity.}\;
The \emph{set-level} closure fidelity
\(\mathsf{F}_{\Cn}(S_O^{(i)},S_O^{(j)})=1\)
\textup{(Definition~\ref{def:closure-fidelity})} requires
\(\Cn(S_O^{(i)})=\Cn(S_O^{(j)})\) and is independent of
any code; the \emph{closure error probability}
\(P_{e,\Cn}^{(n)}\to 0\)
\textup{(Definition~\ref{def:semantic-reliability}(ii))} is
a property of a specific \((n,M)\) code.
The former requires both~\textup{(H1)} and~\textup{(H2)};
the latter requires only~\textup{(H1)}
\textup{(}Remark~\ref{rem:H1-H2-roles}\textup{)}.


\subsubsection*{Part 1: Pairwise Heterogeneous Communication}

\begin{theorem}[Closure reliability for a heterogeneous pair]
\label{thm:heterogeneous-closure}
Let \((i,j)\) be a sender--receiver pair with heterogeneous
semantic channel \(\mathfrak C^{ij}\)
\textup{(Definition~\ref{def:heterogeneous-channel})},
carrier channel kernel
\(W_{ij}:S_C\rightsquigarrow\hat S_C\) with \(C(W_{ij})>0\),
and overlap decomposition
\textup{(Definition~\ref{def:overlap-decomposition})}.
Assume:
\begin{enumerate}[label=\textup{(H\arabic*)}]
  \item \(A_{-}^{ij}=\varnothing\)
        \textup{(}the sender's irredundant core is contained in
        the receiver's vocabulary:
        \(A^{(i)}\subseteq S_O^{(j)}\)\textup{)};
  \item \(S_{+,n}^{ij}=\varnothing\)
        \textup{(}all surplus states in the receiver's vocabulary
        are derivable from the sender's knowledge base:
        \(S_{+}^{ij}\subseteq\Cn(S_O^{(i)})\)\textup{)}.
\end{enumerate}
Then:
\begin{enumerate}[label=\textup{(\roman*)}]
  \item \emph{Set-level closure fidelity:}\;
        \(\mathsf{F}_{\Cn}(S_O^{(i)},S_O^{(j)})=1\)
        \textup{(Corollary~\ref{cor:overlap-sufficient})}.
  \item \emph{Achievability:}\;
        There exists a sequence of
        \((n,|S_O^{(i)}|)\) semantic block codes
        \textup{(Definition~\ref{def:semantic-codebook})}
        with message set \(\mathcal M=S_O^{(i)}\),
        encoding into \(S_C^n\), decoding into
        \(\hat S_O=S_O^{(j)}\), and
        \(P_{e,\Cn}^{(n)}\to 0\) as \(n\to\infty\), provided
        \begin{equation}\label{eq:hetero-rate-condition}
          \frac{\log|A^{(i)}|}{n}\;<\;C(W_{ij}).
        \end{equation}
  \item \emph{Converse
        \textup{(}under
        Assumption~\textup{\ref{assump:core-disjoint}}
        for \(S_O^{(i)}\) with output space
        \(S_O^{(j)}\)\textup{):}}\;
        Any \((n,|S_O^{(i)}|)\) code with
        \(P_{e,\Cn}^{(n)}\le\epsilon\) satisfies
        \begin{equation}\label{eq:hetero-converse}
          \log|A^{(i)}|
          \;\le\;
          \frac{nC(W_{ij})+1}{1-\epsilon}.
        \end{equation}
\end{enumerate}
\end{theorem}

\begin{proof}
\textup{(i)}:\
Immediate from Corollary~\ref{cor:overlap-sufficient}.

\smallskip\noindent
\textup{(ii)}:\
The code is constructed in two layers, adapting
Theorem~\ref{thm:achievability}(ii) to the heterogeneous
output alphabet \(\hat S_O=S_O^{(j)}\).

\emph{Layer~1 (core code).}\;
Since \(\log|A^{(i)}|/n<C(W_{ij})\), the classical channel
coding theorem
\cite{shannon1948mathematical,cover2006elements} yields an
\((n,|A^{(i)}|)\) block code
\((f_n^A,g_n^A)\) for~\(W_{ij}\) with message set
\(A^{(i)}\) and
\(P_e^{(n)}(A^{(i)})\to 0\).
By~\textup{(H1)},
\(A^{(i)}\subseteq S_O^{(j)}=\hat S_O\), so the decoder
can output elements of~\(A^{(i)}\).

\emph{Layer~2 (redundant extension).}\;
Fix an arbitrary \(a_0\in A^{(i)}\).
For each redundant state
\(j\in J^{(i)}=S_O^{(i)}\setminus A^{(i)}\), set
\(f_n(j):=f_n^A(a_0)\).
The decoder first applies~\(g_n^A\) to recover
\(\hat a\in A^{(i)}\) (or an incorrect element in the error
event), and outputs~\(\hat a\).

\emph{Closure analysis.}\;
For \(m\in A^{(i)}\): if the core code decodes correctly
(\(\hat a=m\)), then
\(d_{\Cn}(m,m\mid S_O^{(i)})=0\).
Error probability:
\(P_e^{(n)}(A^{(i)})\to 0\).

For \(m=j\in J^{(i)}\): the decoder outputs some
\(\hat a\in A^{(i)}\subseteq S_O^{(i)}\subseteq
\Cn(S_O^{(i)})\).
Since \(j\) is redundant in~\(S_O^{(i)}\),
\(\Cn(S_O^{(i)}\setminus\{j\})=\Cn(S_O^{(i)})\).
Because
\(\hat a\in\Cn(S_O^{(i)})=\Cn(S_O^{(i)}\setminus\{j\})\),
monotonicity and idempotence of~\(\Cn\) give
\(\Cn\bigl((S_O^{(i)}\setminus\{j\})\cup\{\hat a\}\bigr)
  =\Cn(S_O^{(i)}\setminus\{j\})=\Cn(S_O^{(i)})\),
so \(d_{\Cn}(j,\hat a\mid S_O^{(i)})=0\).
The closure error probability for redundant messages
is~\emph{zero} for all~\(n\).

Combining:
\(P_{e,\Cn}^{(n)}\le P_e^{(n)}(A^{(i)})\to 0\).

\smallskip\noindent
\textup{(iii)}:\
The argument is identical to the proof of
Theorem~\ref{thm:converse}(ii), with
\(S_O\) replaced by \(S_O^{(i)}\),
\(\hat S_O\) replaced by \(S_O^{(j)}\), and
\(A=\Atom(S_O^{(i)})=A^{(i)}\).
Under Assumption~\ref{assump:core-disjoint} (applied to
\(A^{(i)}\) with acceptable sets in~\(S_O^{(j)}\)),
the pairwise-disjoint decoding regions and the Fano argument
yield~\eqref{eq:hetero-converse}.
\end{proof}

\begin{remark}[Role of conditions \textup{(H1)} and
  \textup{(H2)}]
\label{rem:H1-H2-roles}
Condition~\textup{(H1)} is used only in the achievability
proof to ensure that the decoder can output core elements
(\(A^{(i)}\subseteq S_O^{(j)}\)).
Condition~\textup{(H2)} is used only to establish
set-level closure fidelity
\(\mathsf{F}_{\Cn}=1\) in part~\textup{(i)};
it does \emph{not} enter the achievability or converse
proofs, which depend only on the closure distortion
\(d_{\Cn}(\cdot,\cdot\mid S_O^{(i)})\) measured relative
to the \emph{sender's} knowledge base.
Hence the coding-theoretic conclusions~\textup{(ii)}
and~\textup{(iii)} hold under~\textup{(H1)} alone.
Condition~\textup{(H2)} provides the additional guarantee
that the receiver's overall knowledge base generates the
same deductive closure as the sender's.
\end{remark}

\begin{remark}[Weak core coverage and alternative decoding]
\label{rem:weak-F1-coding}
Theorem~\ref{thm:heterogeneous-closure} uses the strong core
coverage condition
\textup{(H1)}: \(A^{(i)}\subseteq S_O^{(j)}\).
As noted in Remark~\ref{rem:strong-weak-coverage}, set-level
closure fidelity \(\mathsf{F}_{\Cn}=1\) requires only the
weaker condition~\textup{(F1)}:
\(A^{(i)}\subseteq\Cn(S_O^{(j)})\).
When \textup{(F1)} holds but \textup{(H1)} fails
(i.e., some core element \(a\in A_{-}^{ij}\) is derivable from
\(S_O^{(j)}\) but not literally present), the two-layer code
cannot directly output~\(a\).
A modified decoder could instead output a
\emph{proxy element}
\(\hat a\in S_O^{(j)}\) satisfying
\(d_{\Cn}(a,\hat a\mid S_O^{(i)})=0\)---i.e., a state in
the receiver's vocabulary whose substitution for~\(a\) preserves
the sender's deductive closure.
Such a proxy exists whenever
\(a\in\Cn(S_O^{(j)})\), but identifying it requires
knowledge of the sender's closure structure at the decoder,
making the code design more involved.
A complete treatment of proxy-based decoding is deferred to
future work; the results of this section focus on the
operationally simpler setting where \textup{(H1)} holds.
The proxy-based decoding strategy shares conceptual
affinity with the inverse contextual reasoning of
Seo et~al.~\cite{seo2023bayesian}, who address the problem
of inferring a sender's communication context from noisy
observations using Bayesian methods.
\end{remark}

\begin{theorem}[Heterogeneous deductive compression]
\label{thm:heterogeneous-compression}
Under the hypotheses of
Theorem~\textup{\ref{thm:heterogeneous-closure}}, the
minimum blocklength for closure-reliable communication
of the full knowledge base \(S_O^{(i)}\) to agent~\(j\)
satisfies, for sufficiently small \(\epsilon>0\):
\begin{enumerate}[label=\textup{(\roman*)}]
  \item \emph{Closure blocklength:}\;
        \begin{align}\label{eq:hetero-blocklength-closure}
          \frac{(1-\epsilon)\log|A^{(i)}|-1}{C(W_{ij})}
          \;\le\;
          &n^*\bigl(S_O^{(i)},W_{ij},P_{e,\Cn},\epsilon\bigr) \nonumber \\
          \;\le\;
          &\left\lceil
            \frac{\log|A^{(i)}|}{C(W_{ij})-\delta(\epsilon)}
          \right\rceil,
        \end{align}
        where \(\delta(\epsilon)\to 0\) as \(\epsilon\to 0\).
  \item \emph{Hamming baseline:}\;
        Under the additional hypothesis
        \(S_O^{(i)}\subseteq S_O^{(j)}\)
        \textup{(}i.e., \(S_{-}^{ij}=\varnothing\),
        which strengthens \textup{(H1)} to full vocabulary
        containment\textup{)},
        \begin{equation}\label{eq:hetero-blocklength-Hamming}
          n^*\bigl(S_O^{(i)},W_{ij},P_e,\epsilon\bigr)
          \;\ge\;
          \frac{(1-\epsilon)\log|S_O^{(i)}|-1}{C(W_{ij})}.
        \end{equation}
  \item \emph{Deductive compression ratio:}\;
        When both bounds apply,
        \begin{equation}\label{eq:hetero-compression-ratio}
          \frac{n^*(P_{e,\Cn})}{n^*(P_e)}
          \;\approx\;
          \frac{\log|A^{(i)}|}{\log|S_O^{(i)}|}\,,
        \end{equation}
        identical to the homogeneous ratio of
        Corollary~\textup{\ref{cor:min-blocklength}}.
\end{enumerate}
\end{theorem}

\begin{proof}
Part~\textup{(i)} combines
Theorem~\ref{thm:heterogeneous-closure}(ii)
(upper bound) and~(iii) (lower bound).
Part~\textup{(ii)} is
Theorem~\ref{thm:converse}(i) applied with
\(M=|S_O^{(i)}|\); the condition
\(S_O^{(i)}\subseteq S_O^{(j)}\) ensures that the
Hamming criterion is meaningful (each sent state has a valid
identity reconstruction in the receiver's vocabulary).
Part~\textup{(iii)} follows by dividing the bounds.
\end{proof}

\begin{remark}[Heterogeneity does not degrade the compression
  ratio]
\label{rem:compression-invariance}
The deductive compression
ratio~\eqref{eq:hetero-compression-ratio} depends only on
the sender's knowledge-base structure
(\(|A^{(i)}|\) vs.\ \(|S_O^{(i)}|\)) and not on the
receiver's vocabulary~\(S_O^{(j)}\), provided the core
coverage condition~\textup{(H1)} holds.
This invariance is a consequence of the two-layer code
structure: the core sub-code operates identically regardless
of the receiver's surplus states, and the redundant extension
incurs zero closure distortion by the algebraic properties
of~\(\Cn\).

The Hamming baseline~\eqref{eq:hetero-blocklength-Hamming}
does, however, depend on the receiver's vocabulary: it
requires \(S_{-}^{ij}=\varnothing\) (full vocabulary
containment), a strictly stronger condition than~\textup{(H1)}.
When \(S_{-}^{ij}\neq\varnothing\), perfect Hamming
reconstruction is impossible (some sent states have no
counterpart in the receiver's vocabulary), while closure
reliability may still be achievable under~\textup{(H1)}.
This gap illustrates the advantage of semantic fidelity
criteria over symbol-level criteria in heterogeneous settings.
\end{remark}

\begin{corollary}[Impossibility under core loss]
\label{cor:heterogeneous-impossibility}
Let \((i,j)\) be a sender--receiver pair.
\begin{enumerate}[label=\textup{(\roman*)}]
  \item \emph{Set-level impossibility:}\;
        If condition~\textup{(F1)} of
        Proposition~\textup{\ref{prop:overlap-fidelity}} fails,
        i.e.,
        \(A^{(i)}\not\subseteq\Cn\bigl(S_O^{(j)}\bigr)\),
        then
        \(\mathsf{F}_{\Cn}(S_O^{(i)},S_O^{(j)})<1\).
        This is a property of the knowledge-base pair,
        independent of the channel, the blocklength, and the
        coding strategy.
  \item \emph{Quantitative bound:}\;
        Under the hypothesis of~\textup{(i)}, the closure
        fidelity satisfies
        \begin{equation}\label{eq:Fcn-bound}
          \mathsf{F}_{\Cn}\bigl(S_O^{(i)},S_O^{(j)}\bigr)
          \;=\;
          \frac{|\Cn(S_O^{(i)})\cap\Cn(S_O^{(j)})|}
               {|\Cn(S_O^{(i)})\cup\Cn(S_O^{(j)})|}\;<\;1.
        \end{equation}
  \item \emph{Core preservation ratio:}\;
        If \(A_{-}^{ij}\neq\varnothing\), then
        \(\rho_{\Atom}(S_O^{(i)},S_O^{(j)})
          =1-|A_{-}^{ij}|/|A^{(i)}|<1\)
        \textup{(Proposition~\ref{prop:overlap-noise}(ii))}.
\end{enumerate}
Similarly, if condition~\textup{(F2)} fails
\textup{(}\(S_{+,n}^{ij}\neq\varnothing\)\textup{)}, then
\(\mathsf{F}_{\Cn}<1\) regardless of any coding strategy.
\end{corollary}

\begin{proof}
Part~\textup{(i)} is the contrapositive of
Proposition~\ref{prop:overlap-fidelity}.
Part~\textup{(ii)} is the definition of
\(\mathsf{F}_{\Cn}\)
(Definition~\ref{def:closure-fidelity}); the strict
inequality follows from~\textup{(i)}.
Part~\textup{(iii)} is
Proposition~\ref{prop:overlap-noise}(ii).
The final claim follows from
Proposition~\ref{prop:overlap-fidelity} (necessity
of~\textup{(F2)}).
\end{proof}

\begin{proposition}[Minimum receiver vocabulary for
  closure-reliable communication]
\label{prop:min-vocabulary}
Given a sender knowledge base \(S_O^{(i)}\), the
minimum-cardinality subset
\(V\subseteq S_O^{(i)}\) serving as a receiver vocabulary
(\(\hat S_O=V\)) that simultaneously achieves:
\begin{enumerate}[label=\textup{(\alph*)}]
  \item the two-layer code of
        Theorem~\textup{\ref{thm:heterogeneous-closure}(ii)}
        achieves \(P_{e,\Cn}^{(n)}\to 0\), and
  \item \(\mathsf{F}_{\Cn}(S_O^{(i)},V)=1\),
\end{enumerate}
is \(V^*=\Atom(S_O^{(i)})=A^{(i)}\), with
\(|V^*|=\mathsf{A}(\mathcal I^{(i)})\).
\end{proposition}

\begin{proof}
\emph{Sufficiency.}\;
Set \(V=A^{(i)}\).
Since \(A^{(i)}\subseteq S_O^{(i)}\), the overlap
with sender \(i\) and ``receiver'' \(V\) gives
\(S_{+}^{ij}=V\setminus S_O^{(i)}=\varnothing\) and
\(A_{-}^{ij}=A^{(i)}\setminus V=\varnothing\)
(condition~\textup{(H1)}).
Condition~\textup{(H2)} holds trivially since
\(S_{+}^{ij}=\varnothing\).
Closure fidelity:
\(\Cn(A^{(i)})=\Cn(S_O^{(i)})\)
\textup{(Proposition~\ref{prop:atom-core-correct}(i))},
giving \(\mathsf{F}_{\Cn}(S_O^{(i)},A^{(i)})=1\).
The two-layer code of
Theorem~\ref{thm:heterogeneous-closure}(ii) applies
with \(\hat S_O=V=A^{(i)}\).

\smallskip\noindent
\emph{Minimality.}\;
Let \(V\subseteq S_O^{(i)}\) satisfy both
conditions~\textup{(a)} and~\textup{(b)}.
Condition~\textup{(a)} requires the two-layer code
of Theorem~\ref{thm:heterogeneous-closure}(ii) to
succeed.
That code's decoder outputs elements of~\(A^{(i)}\),
so the output alphabet \(\hat S_O=V\) must contain
every core element: \(A^{(i)}\subseteq V\).
Hence \(|V|\ge|A^{(i)}|\).
Since \(V=A^{(i)}\) achieves this bound, it is
minimal.
\end{proof}

\begin{remark}[Vocabulary design rule]
\label{rem:vocabulary-design}
Proposition~\ref{prop:min-vocabulary} yields a principled
vocabulary-selection rule for receiver design:
\emph{the receiver need store only the sender's irredundant
core}.
All remaining semantic states (the sender's stored
shortcuts~\(J^{(i)}\)) can be reconstructed by the receiver's
inference engine via \(\Cn(A^{(i)})\).
The channel-use cost of this strategy is
\(n^*\approx\log|A^{(i)}|/C(W_{ij})\), the minimum
achievable under closure reliability.

When the receiver already maintains a richer vocabulary
\(S_O^{(j)}\supsetneq A^{(i)}\), the additional states are
harmless provided \(S_{+,n}^{ij}=\varnothing\)
(condition~\textup{(H2)}); they do not increase the
blocklength.
When some surplus states are non-derivable
(\(S_{+,n}^{ij}\neq\varnothing\)), set-level closure fidelity
drops below~\(1\)
(Corollary~\ref{cor:heterogeneous-impossibility}), but the
coding-theoretic closure reliability may still hold if the
two-layer code is used (since it ignores the surplus entirely;
see Remark~\ref{rem:H1-H2-roles}).
\end{remark}


\subsubsection*{Part 2: Broadcast Extension}

We now extend the pairwise results to the broadcast scenario
of Definition~\ref{def:broadcast-scenario}.
Agent~\(0\) (the sender) communicates its knowledge base
\(S_O^{(0)}\) to \(K\) receivers over a common carrier
channel \(W:S_C\rightsquigarrow\hat S_C\) with
\(C(W)>0\).

\begin{theorem}[Broadcast deductive compression]
\label{thm:broadcast-compression}
Suppose that for every receiver
\(j\in\{1,\ldots,K\}\), the overlap conditions hold:
\begin{enumerate}[label=\textup{(BH\arabic*)}]
  \item \(A_{-}^{0j}=\varnothing\)
        \textup{(}the sender's core is contained in every
        receiver's vocabulary:
        \(A^{(0)}\subseteq S_O^{(j)}\)\textup{)};
  \item \(S_{+,n}^{0j}=\varnothing\)
        \textup{(}every receiver's surplus is derivable from
        the sender:
        \(S_{+}^{0j}\subseteq\Cn(S_O^{(0)})\)\textup{)}.
\end{enumerate}
Then:
\begin{enumerate}[label=\textup{(\roman*)}]
  \item \emph{Simultaneous closure fidelity:}\;
        \(\mathsf{F}_{\Cn}(S_O^{(0)},S_O^{(j)})=1\) for
        every \(j\in\{1,\ldots,K\}\).
  \item \emph{Broadcast achievability:}\;
        There exists a \emph{single} sequence of
        \((n,|S_O^{(0)}|)\) semantic block codes
        (with a common encoding function~\(f_n\)) such that
        \(P_{e,\Cn}^{(n,j)}\to 0\) simultaneously for all
        receivers \(j\in\{1,\ldots,K\}\), provided
        \begin{equation}\label{eq:broadcast-rate}
          \frac{\log|A^{(0)}|}{n}\;<\;C(W).
        \end{equation}
  \item \emph{Blocklength independence from \(K\):}\;
        The minimum blocklength for broadcast closure
        reliability is
        \begin{equation}\label{eq:broadcast-blocklength}
          n^*_{\mathrm{bc}}
          \;\approx\;
          \frac{\log|A^{(0)}|}{C(W)},
        \end{equation}
        independent of the number of receivers~\(K\).
  \item \emph{Broadcast converse:}\;
        Under
        Assumption~\textup{\ref{assump:core-disjoint}} for
        \(S_O^{(0)}\) with output space \(S_O^{(j)}\) for
        each~\(j\), any code achieving
        \(\max_j P_{e,\Cn}^{(n,j)}\le\epsilon\) satisfies
        \(\log|A^{(0)}|\le(nC(W)+1)/(1-\epsilon)\).
\end{enumerate}
\end{theorem}

\begin{proof}
\textup{(i)}:\
For each~\(j\), conditions \textup{(BH1)}--\textup{(BH2)}
instantiate \textup{(H1)}--\textup{(H2)} of
Theorem~\ref{thm:heterogeneous-closure}, giving
\(\mathsf{F}_{\Cn}(S_O^{(0)},S_O^{(j)})=1\) by
Corollary~\ref{cor:overlap-sufficient}.

\smallskip\noindent
\textup{(ii)}:\
Construct a single two-layer code as in the proof of
Theorem~\ref{thm:heterogeneous-closure}(ii), with the
\emph{common} core code \((f_n^A,g_n^A)\) for message
set~\(A^{(0)}\).
All \(K\) receivers observe the same channel output
\(\hat S_C^n\) and each independently applies the
same core decoder~\(g_n^A\).
Since \(A^{(0)}\subseteq S_O^{(j)}\) for every~\(j\)
(by~\textup{(BH1)}), the decoded core element
\(\hat a\in A^{(0)}\) is a valid output for every receiver.
The Layer~2 redundant extension and closure analysis are
identical to the pairwise case (using the sender's
closure structure only), so
\(P_{e,\Cn}^{(n,j)}\le P_e^{(n)}(A^{(0)})\to 0\)
simultaneously for all~\(j\).

\smallskip\noindent
\textup{(iii)}:\
The blocklength is determined by the core code,
which has rate \(\log|A^{(0)}|/n\), independent of~\(K\).

\smallskip\noindent
\textup{(iv)}:\
Fix any receiver~\(j\).
Theorem~\ref{thm:heterogeneous-closure}(iii) applied to the
pair \((0,j)\) gives
\(\log|A^{(0)}|\le(nC(W)+1)/(1-\epsilon)\) whenever
\(P_{e,\Cn}^{(n,j)}\le\epsilon\).
Since this must hold for every~\(j\), the bound holds under
\(\max_j P_{e,\Cn}^{(n,j)}\le\epsilon\).
\end{proof}

\begin{proposition}[Broadcast semantic bottleneck]
\label{prop:broadcast-bottleneck}
In the broadcast scenario of
Definition~\textup{\ref{def:broadcast-scenario}}, suppose
there exists a receiver \(j^*\in\{1,\ldots,K\}\) such that
condition~\textup{(F1)} of
Proposition~\textup{\ref{prop:overlap-fidelity}} fails for
the pair \((0,j^*)\):
\[
  A^{(0)}\not\subseteq\Cn\bigl(S_O^{(j^*)}\bigr).
\]
Then:
\begin{enumerate}[label=\textup{(\roman*)}]
  \item \(\mathsf{F}_{\Cn}(S_O^{(0)},S_O^{(j^*)})<1\),
        regardless of the carrier channel~\(W\), the
        blocklength~\(n\), and the encoding/decoding strategy.
  \item Even if the carrier channel is noiseless
        (\(W=\mathrm{id}_{S_C}\)), the closure fidelity at
        receiver~\(j^*\) is bounded by
        \begin{equation}\label{eq:bottleneck-bound}
          \mathsf{F}_{\Cn}\bigl(S_O^{(0)},S_O^{(j^*)}\bigr)
          \;=\;
          \frac{|\Cn(S_O^{(0)})\cap\Cn(S_O^{(j^*)})\!|}
               {|\Cn(S_O^{(0)})\cup\Cn(S_O^{(j^*)})\!|}
          \;<\;1.
        \end{equation}
  \item Receiver~\(j^*\) is a \emph{semantic bottleneck}:
        its performance limitation arises from vocabulary
        mismatch, not from the physical channel.
        The other receivers \(j\neq j^*\) satisfying
        \textup{(BH1)}--\textup{(BH2)} achieve
        \(\mathsf{F}_{\Cn}=1\) and closure reliability
        simultaneously, unaffected by~\(j^*\).
\end{enumerate}
\end{proposition}

\begin{proof}
Parts~\textup{(i)} and~\textup{(ii)} follow from
Corollary~\ref{cor:heterogeneous-impossibility}(i)--(ii)
applied to the pair \((0,j^*)\).
Part~\textup{(iii)}: the common encoding and core code are
shared by all receivers; the failure at~\(j^*\) is due
solely to the mismatch
\(A^{(0)}\not\subseteq\Cn(S_O^{(j^*)})\), which is
independent of the channel.
Receivers satisfying
\textup{(BH1)}--\textup{(BH2)} are handled by
Theorem~\ref{thm:broadcast-compression}.
\end{proof}

\begin{remark}[Semantic bottleneck vs.\ classical channel
  degradation]
\label{rem:bottleneck-vs-classical}
In classical broadcast channel theory
\cite{cover2006elements,csiszar2011information}, the weakest
receiver is the one with the noisiest channel, and the rate
region depends only on channel transition probabilities.
Proposition~\ref{prop:broadcast-bottleneck} reveals a second,
purely \emph{structural} axis of weakness invisible to
classical theory: even over a noiseless carrier, a receiver
with \(A^{(0)}\not\subseteq\Cn(S_O^{(j^*)})\) cannot achieve
\(\mathsf{F}_{\Cn}=1\).
This semantic bottleneck is irreducible by coding: resolving
it requires a \emph{design-time} vocabulary augmentation
(e.g., pre-loading \(A^{(0)}\) into \(S_O^{(j^*)}\)
per Proposition~\ref{prop:min-vocabulary}), fundamentally
different from the encoder/decoder optimization that
suffices in classical settings.
\end{remark}

\begin{remark}[Classical recovery]
\label{rem:broadcast-classical-recovery}
When \(S_O^{(j)}=S_O^{(0)}\) for all~\(j\),
\textup{(BH1)}--\textup{(BH2)} hold vacuously and
Theorem~\ref{thm:broadcast-compression} reduces to
\(K\)~independent applications of
Theorem~\ref{thm:achievability}(ii).
If additionally \(A^{(0)}=S_O^{(0)}\)
\textup{(irredundant source)}, the broadcast blocklength
becomes \(n^*\approx\log|S_O^{(0)}|/C(W)\), recovering the
classical channel coding theorem~\cite{cover2006elements}.
\end{remark}

\begin{remark}[Summary of answers to Q1--Q4]
\label{rem:Q1-Q4-summary}
The results of this subsection answer the four key questions
posed in Section~\ref{subsec:app-problem}:

\textbf{Q1} (closure reliability from overlap):
Proposition~\ref{prop:overlap-fidelity} and
Corollary~\ref{cor:overlap-sufficient} provide the
necessary and sufficient conditions; the operational
achievability under the strong condition~\textup{(H1)} is
Theorem~\ref{thm:heterogeneous-closure}.

\textbf{Q2} (heterogeneous compression):
Theorem~\ref{thm:heterogeneous-compression} establishes
that the deductive compression ratio is invariant under
vocabulary heterogeneity.
Corollary~\ref{cor:heterogeneous-impossibility} characterizes
the impossibility regime.

\textbf{Q3} (invariant diagnosis):
Section~\ref{subsec:app-theory}
(Propositions~\ref{prop:overlap-noise}--\ref{prop:overlap-structural}
and Remark~\ref{rem:invariant-overlap-summary})
expresses every invariant family in terms of the overlap
decomposition.

\textbf{Q4} (broadcast bottleneck):
Theorem~\ref{thm:broadcast-compression} shows
blocklength independence from~\(K\) under core coverage;
Proposition~\ref{prop:broadcast-bottleneck} identifies
the semantic bottleneck phenomenon.
\end{remark}

\subsection{Numerical Validation}%
\label{sec:numerical-validation}%

This subsection verifies the theoretical results on two scales:
a small Datalog instance permitting closed-form computation of
every invariant, and medium-scale supply-chain knowledge graphs
with up to approximately 24\,000 base facts.

\subsubsection*{Small-Scale Datalog Instance}

\begin{example}[Path-reachability knowledge bases]
\label{ex:datalog-instance}
Fix a domain \(\mathcal D=\{a,b,c,d\}\) with relation symbols
\(\mathit{Edge}\) and \(\mathit{Path}\) and two Datalog rules:
\(\mathit{Path}(x,y)\leftarrow\mathit{Edge}(x,y)\) and
\(\mathit{Path}(x,z)\leftarrow\mathit{Edge}(x,y),\,
\mathit{Path}(y,z)\).
\end{example}

\noindent\textbf{Agents.}\;
The sender (agent~1) stores
\(|S_O^{(1)}|=8\) facts: four \(\mathit{Edge}\) facts
\[
\{\mathit{Edge}(a,b),\mathit{Edge}(a,c), \mathit{Edge}(b,c),\mathit{Edge}(c,d)\}
\]
and four \(\mathit{Path}\) shortcuts
\[
\{\mathit{Path}(a,b),\mathit{Path}(b,c), \mathit{Path}(c,d),\mathit{Path}(b,d)\};
\]
irredundantization yields
\[
A^{(1)}=\{\mathit{Edge}(a,b),\mathit{Edge}(a,c),
\mathit{Edge}(b,c),\mathit{Edge}(c,d)\} 
\]
with \(|A^{(1)}|=4\), so all four \(\mathit{Path}\) facts are stored shortcuts.
Three receivers are defined to illustrate distinct overlap
regimes: receiver~2 (core loss:
\(|A_{-}^{12}|=1\), non-derivable surplus:
\(|S_{+,n}^{12}|=1\));
receiver~2\('\) (augmented: \(A_{-}^{12'}=\varnothing\),
\(S_{+,n}^{12'}=\varnothing\));
receiver~3 (broadcast:
\(A_{-}^{13}=\varnothing\), \(S_{+,n}^{13}=\varnothing\)).

\smallskip\noindent\textbf{Carrier channel.}\;
A \(q\)-ary symmetric channel with \(q=10\) and crossover
probability \(p=0.1\), giving
\(C(W)\approx 2.536\) bits.

\smallskip\noindent\textbf{Results.}\;
Table~\ref{tab:overlap-values} reports the overlap decomposition
and set-level invariants.
Table~\ref{tab:full-invariants} lists all six invariant
families of Theorem~\ref{thm:invariant-summary}.
Table~\ref{tab:blocklength} compares the minimum blocklength
under Hamming and closure reliability.
Key observations:

\begin{table}[ht]
\centering
\caption{Overlap decomposition and set-level invariants for
  sender agent~1 paired with each receiver.}
\label{tab:overlap-values}
\renewcommand{\arraystretch}{1.2}
\begin{tabular}{@{} l c c c @{}}
\toprule
\textbf{Quantity} &
  \textbf{Recv.\,2} &
  \textbf{Recv.\,2$'$} &
  \textbf{Recv.\,3} \\
\midrule
$|A_{-}^{ij}|$    & 1 & 0 & 0 \\
$|S_{+,n}^{ij}|$  & 1 & 0 & 0 \\
$\rho_{\Atom}$      & 3/4          & 1   & 1   \\
$\mathsf{F}_{\Cn}$  & $3/7$  & 1   & 1   \\
\bottomrule
\end{tabular}
\end{table}

\begin{definition}[Experimental semantic capacity estimate]
\label{def:exp-sem-capacity}
In the numerical validation, we report an \emph{experimental}
semantic-capacity estimate under a fixed encoding and a restricted
decoder class.
Fix an encoding kernel \(\kappa_{\enc}^{\mathrm{id}}\) (the
identity-injection encoding used in the experiments) and a decoder
family \(\mathcal D^{(j)}\subseteq \mathcal K(\mathcal
I_{\mathrm{dec}}^{(j)})\).
Define
\[
  \widehat C_{\mathrm{sem}}^{ij}
  \;:=\;
  \max_{P_O\in\Delta(S_O^{(i)}),\; D\in\mathcal D^{(j)}}
  I\!\left(\mathsf S_o;\hat{\mathsf S}_o\right),
\]
where the joint law is induced by the end-to-end kernel
\(\kappa_{\mathrm{sem}}^{ij}=D\circ W_{ij}\circ
\kappa_{\enc}^{\mathrm{id}}\).
For each fixed \(D\), the maximization over \(P_O\) is a standard
finite-alphabet channel-capacity computation and is carried out via
the Blahut--Arimoto algorithm.

We also report the achieved mutual information
\[
  \widehat I_{\mathrm{sem}}^{ij}
  \;:=\;
  I\!\left(\mathsf S_o;\hat{\mathsf S}_o\right)
\]
under \(\kappa_{\enc}^{\mathrm{id}}\) and a decoder \(D\) chosen to
maximize \(I(\mathsf S_o;\hat{\mathsf S}_o)\) within the same decoder
family \(\mathcal D^{(j)}\).
\end{definition}

\begin{table*}[ht]
{\centering
\caption{Semantic channel invariants for the Datalog instance in
  Example~\ref{ex:datalog-instance}.}
\label{tab:full-invariants}
\renewcommand{\arraystretch}{1.25}
\begin{tabularx}{\textwidth}{@{}l l *{3}{>{\centering\arraybackslash}X}@{}}
\toprule
\textbf{Family} & \textbf{Invariant} &
  \textbf{Pair \((1,2)\)} &
  \textbf{Pair \((1,2')\)} &
  \textbf{Pair \((1,3)\)} \\
\midrule
\multirow{2}{*}{I.\;Source}
  & \(\mathsf{A}\) & 4 & 4 & 4 \\
  & \(\mathsf{D_d}\) & 2 & 2 & 2 \\
\midrule
\multirow{2}{*}{II.\;Set-level}
  & \(\rho_{\Atom}\) & 0.750 & 1.000 & 1.000 \\
  & \(\mathsf{F}_{\Cn}\) & 0.429 & 1.000 & 1.000 \\
\midrule
\multirow{2}{*}{III.\;Noise-pair}
  & \(\Phi_{\Atom}\) & 0 & 0.900 & 0.900 \\
  & \(\Psi_{+}\) & 0.900 & 0 & 0.911 \\
\midrule
\multirow{2}{*}{IV.\;Quality}
  & \(\mathsf{F}\) & 0.900 & 0.980 & 0.981 \\
  & \(\mathsf{E}\) & 0.078 & 0.078 & 0.494 \\
\midrule
\multirow{2}{*}{V.\;Comparison}
  & \(\Delta\mathsf{A}\) & 0 & 0 & 0 \\
  & \(\Delta\mathsf{D_d}\) & \(+1\) & 0 & 0 \\
\midrule
\multirow{3}{*}{VI.\;Info-th.}
  & \(C(W)\) & 2.536 & 2.536 & 2.536 \\
  & \(\widehat C_{\mathrm{sem}}^{ij}\) & 2.280 & 2.280 & 1.958 \\
  & \(\widehat I_{\mathrm{sem}}^{ij}\) & 2.273 & 2.273 & 1.808 \\
\bottomrule
\end{tabularx}
\par} 

\smallskip
\noindent\footnotesize
\emph{Configuration:}\;
Carrier: \(q\)-ary symmetric channel, \(q=10\), \(p=0.1\);
source: \(P_O\) uniform on~\(S_O^{(1)}\).
Families~I--II and~V are determined by the knowledge-base pair
alone and are independent of the channel configuration.

\emph{Information-theoretic quantities.}\;
The quantities \(\widehat C_{\mathrm{sem}}^{ij}\) and
\(\widehat I_{\mathrm{sem}}^{ij}\) are \emph{experimental} and follow
Definition~\ref{def:exp-sem-capacity}.
In particular, \(\widehat C_{\mathrm{sem}}^{ij}\) is obtained by
maximizing \(I(\mathsf S_o;\hat{\mathsf S}_o)\) over \(P_O\) (via the
Blahut--Arimoto algorithm) and over the chosen decoder family, under
the fixed identity-injection encoding.
The reported \(\widehat I_{\mathrm{sem}}^{ij}\) is evaluated under the
same encoding and a mutual-information-maximizing decoder within the
same decoder family.

\emph{Sanity check.}\;
The inequalities
\(\widehat I_{\mathrm{sem}}^{ij}\le \widehat C_{\mathrm{sem}}^{ij}
\le C(W)\)
are verified for all three pairs.
\end{table*}

\begin{table}[ht]
\centering
\caption{Minimum blocklength estimates
  (\(\epsilon\to 0\)).}
\label{tab:blocklength}
\renewcommand{\arraystretch}{1.2}
\begin{tabularx}{\columnwidth}{l *{3}{>{\centering\arraybackslash}X}}
\toprule
\textbf{Criterion} &
  \textbf{\((1,2)\)} &
  \textbf{\((1,2')\)} &
  \textbf{\((1,3)\)} \\
\midrule
Hamming \(n^*_H\) &
  N/A\(^{\dagger}\) & 1.183 & N/A\(^{\dagger}\) \\
Closure \(n^*_{\Cn}\) &
  \(\nexists\)\(^{\ddagger}\) & 0.789 & 0.789 \\
Ratio &
  --- & 2/3 & --- \\
\bottomrule
\multicolumn{4}{@{}l@{}}{\parbox{\linewidth}{\footnotesize
  \(^{\dagger}\)\(S_{-}^{ij}\neq\varnothing\)
  (Hamming reconstruction undefined).\;
  \(^{\ddagger}\)\(A_{-}^{ij}\neq\varnothing\), so
  \(\mathsf{F}_{\Cn}<1\) by
  Corollary~\ref{cor:heterogeneous-impossibility}.}}
\end{tabularx}
\end{table}

\emph{Data processing sanity check.}\;
The experimental quantities satisfy the chain
\(\widehat I_{\mathrm{sem}}^{ij}\le \widehat C_{\mathrm{sem}}^{ij}
\le C(W)\) for every pair:
\(2.273\le 2.280\le 2.536\) for~\((1,2)\),
\(2.273\le 2.280\le 2.536\) for~\((1,2')\), and
\(1.808\le 1.958\le 2.536\) for~\((1,3)\).
The equality \(\widehat C_{\mathrm{sem}}^{12}=\widehat
C_{\mathrm{sem}}^{12'}=2.280\) indicates that, \emph{under the fixed
encoding and decoder family of
Definition~\ref{def:exp-sem-capacity}}, the semantic-capacity estimate
is driven primarily by the carrier channel and the effective
source/receiver alphabet sizes, and is insensitive to the particular
logical \emph{content} of the receiver vocabulary in these two cases.
Under mutual-information-maximizing decoding,
\(\widehat I_{\mathrm{sem}}^{12}=\widehat I_{\mathrm{sem}}^{12'}=2.273\):
the vocabulary mismatch in pair~\((1,2)\) (core loss and
non-derivable surplus) is \emph{invisible} at the
information-theoretic level.
The discrepancy is instead captured by the
set-level and noise-pair invariants:
\(\mathsf{F}_{\Cn}^{12}=3/7\neq 1=\mathsf{F}_{\Cn}^{12'}\),
\(\Phi_{\Atom}^{12}=0\neq 0.9=\Phi_{\Atom}^{12'}\),
and
\(\Psi_{+}^{12}=0.900\neq 0=\Psi_{+}^{12'}\).
Here \(\Psi_+^{12}=0.900\) is attained at the lost-core input
\(s_o=\mathit{Edge}(a,c)\): when this message is sent, the optimized
decoder outputs the surplus atom \(\mathit{Edge}(d,a)\) with
probability \(0.9\), hence the maximal spurious-output probability
equals \(0.9\) by Definition~\ref{def:noise-pair-indices-def}.
The core preservation index detects the loss of
\(\mathit{Edge}(a,c)\) from the receiver's vocabulary,
while the spurious probability index captures the
complementary effect---the decoder outputs the surplus
element \(\mathit{Edge}(d,a)\) with probability~\(0.9\)
whenever the lost core element is sent.
This illustrates the diagnostic value of the multi-family
invariant architecture of
Theorem~\ref{thm:invariant-summary}:
families~II--III detect a semantic impairment that
family~VI cannot distinguish.
For pair~\((1,3)\), the smaller receiver vocabulary
(\(|S_O^{(3)}|=6\)) constrains
\(\widehat C_{\mathrm{sem}}^{13}=1.958<2.280\), demonstrating a
\emph{capacity-level} vocabulary bottleneck that complements
the set-level bottleneck of
Proposition~\ref{prop:broadcast-bottleneck}.

\emph{Deductive compression.}\;
The ratio \(n^*_{\Cn}/n^*_H=2/3=\log|A^{(1)}|/\log|S_O^{(1)}|\)
for pair \((1,2')\) matches
Theorem~\ref{thm:heterogeneous-compression}(iii) exactly.

\emph{Vocabulary design.}\;
Augmenting receiver~2 to receiver~2\('\) (adding the single
lost core element \(\mathit{Edge}(a,c)\) and removing the
non-derivable surplus \(\mathit{Edge}(d,a)\)) raises
\(\mathsf{F}_{\Cn}\) from~\(3/7\) to~\(1\) and
\(\Phi_{\Atom}\) from~\(0\) to~\(0.9\), while eliminating
the spurious output probability
(\(\Psi_{+}\) from~\(0.900\) to~\(0\)), confirming
Proposition~\ref{prop:min-vocabulary}.

\emph{Broadcast bottleneck.}\;
Agent~1 broadcasts to receivers~2 and~3.
Receiver~3 satisfies \textup{(BH1)}--\textup{(BH2)} and achieves
\(\mathsf{F}_{\Cn}=1\) with
\(n^*_{\mathrm{bc}}=\lceil\log 4/C(W)\rceil=1\).
Receiver~2 violates \textup{(BH1)} and is a semantic bottleneck
(\(\mathsf{F}_{\Cn}<1\) regardless of \(W\)), confirming
Proposition~\ref{prop:broadcast-bottleneck}.


\subsubsection*{Medium-scale Experimental Metrics}

The following metrics are used in the medium-scale experiments.
Let \(\mathcal{F}\) denote the set of stored base-predicate
facts and \(\mathrm{Cl}(\mathcal{F}):=\Cn(\mathcal{F})\cap
\mathbb{S}_O\) the closure restricted to the ambient universe.

\begin{definition}[Deductive amplification factor]
\label{def:amp-factor}
\(\gamma_{\mathrm{amp}}(\mathcal{F})
:=|\mathrm{Cl}(\mathcal{F})|/|\mathcal{F}|\).
\end{definition}

\begin{definition}[Syntactic and semantic Jaccard indices]
\label{def:jaccard-indices}
For agents \(i,j\) with fact sets \(\mathcal{F}^{(i)},
\mathcal{F}^{(j)}\):
\[
  \Omega^{\mathrm{syn}}_{ij}
  :=\frac{|\mathcal{F}^{(i)}\cap\mathcal{F}^{(j)}|}
         {|\mathcal{F}^{(i)}\cup\mathcal{F}^{(j)}|},
  \qquad
  \Omega^{\mathrm{sem}}_{ij}
  :=\frac{|\mathrm{Cl}(\mathcal{F}^{(i)})\cap
          \mathrm{Cl}(\mathcal{F}^{(j)})|}
         {|\mathrm{Cl}(\mathcal{F}^{(i)})\cup
          \mathrm{Cl}(\mathcal{F}^{(j)})|}.
\]
The pairwise averages are
\(\bar\Omega^{\mathrm{syn}}\) and
\(\bar\Omega^{\mathrm{sem}}\).
Note that \(\Omega^{\mathrm{sem}}_{ij}=
\mathsf{F}_{\Cn}(\mathcal{F}^{(i)},\mathcal{F}^{(j)})\)
\textup{(Definition~\ref{def:closure-fidelity})}.
\end{definition}

\begin{definition}[Closure fidelity curve]
\label{def:fidelity-curve}
For a base fact set~\(\mathcal{F}\) and a randomly selected
subset \(\hat{\mathcal{F}}\subseteq\mathcal{F}\) of
fraction~\(R\):
\(\Phi(R):=|\mathrm{Cl}(\hat{\mathcal{F}})\cap
\mathrm{Cl}(\mathcal{F})|/|\mathrm{Cl}(\mathcal{F})|\).
\end{definition}

\begin{definition}[Deductive compression ratio]
\label{def:compression-ratio-exp}
For a knowledge base \(S_O\) containing both base facts and
materialized derived facts:
\(\rho_{\mathrm{comp}}:=\log|\Atom(S_O)|/\log|S_O|\)
\textup{(}the single-shot ratio of
Corollary~\textup{\ref{cor:min-blocklength})}.
The entropy-based ratio under uniform source is
\(\rho_{\mathrm{ent}}:=P_A H(\pi_A)/H(P_O)
=(k\log k)/(|S_O|\log |S_O|)\) where \(k=|\Atom(S_O)|\).
Note that \(\rho_{\mathrm{comp}}=\Lambda_1^{-1}\) and
\(\rho_{\mathrm{ent}}=\Lambda_\infty^{-1}\); the two coincide
only when \(P_A=1\) \textup{(}irredundant source\textup{)}.
When \(|J|>0\) and the source is uniform,
\(\rho_{\mathrm{ent}}\ll\rho_{\mathrm{comp}}\) because the
additional factor \(P_A=k/|S_O|\) makes the entropic gain much
larger than the log-scale gain.
\end{definition}

\subsubsection*{Medium-Scale Supply-Chain Experiments}

We test scalability on synthetic supply-chain knowledge graphs
with three base predicates (\(\mathsf{connected}\),
\(\mathsf{supplies}\), \(\mathsf{produces}\)) and four Datalog
rules computing transitive reachability and item availability
(see~\eqref{eq:rule-r1}--\eqref{eq:rule-r4} below).
\begin{align}
  \mathsf{reachable}(X,Y) &\leftarrow
      \mathsf{connected}(X,Y),
      \label{eq:rule-r1}\\
  \mathsf{reachable}(X,Z) &\leftarrow
      \mathsf{reachable}(X,Y),\;
      \mathsf{connected}(Y,Z),
      \label{eq:rule-r2}\\
  \mathsf{available}(I,L)  &\leftarrow
      \mathsf{produces}(S,I),\;
      \mathsf{supplies}(S,L),
      \label{eq:rule-r3}\\
  \mathsf{available}(I,L)  &\leftarrow
      \mathsf{produces}(S,I),\;
      \mathsf{supplies}(S,L_0),\notag\\
      &\qquad\quad
      \mathsf{reachable}(L_0,L).
      \label{eq:rule-r4}
\end{align}
The stored knowledge base consists exclusively of
base-predicate facts, so
\(\Atom(\mathcal{F})=\mathcal{F}\) and the deductive
compression ratio is~\(1\)
(Corollary~\ref{cor:irredundant-classical}).
All experiments are implemented in Python~3 using
NetworkX~\cite{hagberg2007exploring}.

\begin{table}[ht]
\centering
\caption{Deductive amplification across knowledge-base scales.
  \(\gamma_{\mathrm{amp}}=|\mathrm{Cl}|/|\mathcal{F}|\).}
\label{tab:scaling}
\renewcommand{\arraystretch}{1.15}
\setlength{\tabcolsep}{4pt}
\footnotesize
\begin{tabular}{rrrrrc}
\toprule
\(|\mathcal{V}|\) & \(p\) & \(|\mathcal{F}|\) &
  \(|\mathrm{Cl}|\) &
  \(\gamma_{\mathrm{amp}}\) & \(d\)\\
\midrule
     50 & 0.060 &       188 &      2\,727 &   14.5 & 10\\
    200 & 0.040 &     1\,705 &     45\,105 &   26.5 &  5\\
    500 & 0.020 &     5\,179 &    268\,679 &   51.9 &  5\\
  1\,000& 0.012 &    12\,305 &  1\,050\,305&   85.4 &  5\\
  2\,000& 0.006 &    24\,304 &  4\,118\,304&  169.4 &  6\\
\bottomrule
\end{tabular}
\end{table}

\paragraph*{Experiment~1: Deductive amplification}
Table~\ref{tab:scaling} reports the closure amplification
\(\gamma_{\mathrm{amp}}=|\mathrm{Cl}(\mathcal{F})|/|\mathcal{F}|\)
for configurations spanning 50 to 2\,000 locations.
At \(|\mathcal{V}|=2{,}000\), the closure exceeds the base
fact set by a factor of~\(169\): approximately 24\,000 base
facts generate over 4.1~million derived consequences.
The maximum derivation depth stabilizes at \(d=5\)--\(10\),
reflecting the short diameter of dense random directed graphs.
This confirms substantial \emph{deductive amplification}
(formalized by the closure operator~\(\Cn\)) and quantifies
the semantic leverage available to a receiver possessing the
shared rules.

\paragraph*{Experiment~2: Multi-agent overlap}
We partition the base facts of a 300-location universe
(\(|\mathcal{F}|=2{,}946\)) among \(K=8\) agents, each
retaining each fact independently with probability~\(0.4\).
The average pairwise semantic Jaccard index is
\(\bar\Omega^{\mathrm{sem}}=0.795\), compared with the
syntactic index \(\bar\Omega^{\mathrm{syn}}=0.246\)---a
factor of~\(3.2\times\).
Thus agents sharing roughly \(25\%\) of base facts share
nearly \(80\%\) of semantic content after closure,
confirming that condition~\textup{(F2)} of
Proposition~\ref{prop:overlap-fidelity} is substantially
easier to satisfy at the semantic level.
Combining any two agents' knowledge bases produces novel
derivations constituting \(1.8\%\)--\(4.5\%\) of the
combined closure, exhibiting non-trivial deductive synergy.

\paragraph*{Experiment~3: Rate--fidelity trade-off}
Using the 200-location base
(\(|\mathcal{F}|=1{,}705\)), we vary the fraction \(R\) of
base facts transmitted and measure closure fidelity
\(\Phi(R)=|\mathrm{Cl}(\hat{\mathcal{F}})\cap
\mathrm{Cl}(\mathcal{F})|/|\mathrm{Cl}(\mathcal{F})|\).
Under random selection: \(\Phi(0.25)=0.53\),
\(\Phi(0.50)=0.89\), \(\Phi(0.75)=0.97\).
The ratio \(\Phi(R)/R>1\) for all \(R<1\), demonstrating
a \emph{semantic leverage effect}: inference rules allow
the receiver to reconstruct a disproportionately large
fraction of the closure from a partial base.
Since the knowledge base is irredundant, perfect recovery
\(\Phi=1\) requires \(R=1\), consistent with
Corollary~\ref{cor:min-blocklength}.
A connectivity-first strategy outperforms random at low
rates (\(R\lesssim 0.35\)) but underperforms at higher rates
due to delayed transmission of supplier/product facts,
illustrating a predicate-balance trade-off.

\paragraph*{Experiment~4: Deductive compression with materialized
shortcuts}
The preceding experiments store only base-predicate facts, so
\(\Atom(\mathcal{F})=\mathcal{F}\) and
\(\rho_{\mathrm{comp}}=1\).
To demonstrate the deductive compression gain predicted by
Theorem~\ref{thm:tight-zero-rate} at scale, we augment the
stored knowledge base by \emph{materializing} a fraction~\(\mu\)
of the derived facts (reachable and available) as stored
shortcuts, simulating a common scenario in which an agent caches
query results or materialized views.

For the 200-location supply-chain universe
(\(|\mathcal{F}_{\mathrm{base}}|=1{,}705\),
\(|\mathrm{Cl}|=45{,}105\)), we set
\(S_O:=\mathcal{F}_{\mathrm{base}}\cup\mathcal{F}_{\mathrm{mat}}\)
where \(\mathcal{F}_{\mathrm{mat}}\) is a uniformly random
subset of \(\mathrm{Cl}\setminus\mathcal{F}_{\mathrm{base}}\)
of size \(\lfloor\mu\cdot|\mathrm{Cl}\setminus
\mathcal{F}_{\mathrm{base}}|\rfloor\).
Under the given Datalog rules, no rule derives a base-predicate
fact, so every materialized IDB fact is redundant and
\(\Atom(S_O)=\mathcal{F}_{\mathrm{base}}\); hence
\(|A|=|\mathcal{F}_{\mathrm{base}}|\) and
\(|J|=|\mathcal{F}_{\mathrm{mat}}|\).
This property---that no Datalog rule derives a base-predicate
(EDB) fact---is specific to the rule
set~\eqref{eq:rule-r1}--\eqref{eq:rule-r4}, in which all rule
heads are IDB predicates.
In knowledge bases with integrity constraints or recursive
rules whose heads include EDB predicates, some base facts could
become derivable, altering the core; the above identification
\(\Atom(S_O)=\mathcal{F}_{\mathrm{base}}\) would then require
verification via the core-extraction procedure of
Definition~\ref{def:atom-so}.
Table~\ref{tab:compression-gain} reports the compression ratios
for varying materialization fractions~\(\mu\).

\begin{table}[t]
\centering
\caption{Deductive compression gain for the 200-location
  supply-chain knowledge base with materialized shortcuts.
  \(|A|=|\mathcal{F}_{\mathrm{base}}|=1{,}705\) throughout;
  \(P_O\) uniform on~\(S_O\).}
\label{tab:compression-gain}
\renewcommand{\arraystretch}{1.15}
\setlength{\tabcolsep}{5pt}
\footnotesize
\begin{tabular}{rrrrrr}
\toprule
\(\mu\) (\%) &
  \(|J|\) &
  \(|S_O|\) &
  \(\rho_{\mathrm{comp}}\) &
  \(\rho_{\mathrm{ent}}\) &
  \(\Lambda_1\) \\
\midrule
  0  &       0 &  1\,705 & 1.000 & 1.000 & 1.00 \\
 10  &   4\,340 &  6\,045 & 0.855 & 0.241 & 1.17 \\
 20  &   8\,680 & 10\,385 & 0.805 & 0.132 & 1.24 \\
 30  &  13\,020 & 14\,725 & 0.775 & 0.090 & 1.29 \\
 50  &  21\,700 & 23\,405 & 0.740 & 0.054 & 1.35 \\
 80  &  34\,720 & 36\,425 & 0.709 & 0.033 & 1.41 \\
100  &  43\,400 & 45\,105 & 0.694 & 0.026 & 1.44 \\
\bottomrule
\end{tabular}
\end{table}

We also verify cross-scale consistency by repeating the
experiment at \(\mu=0.3\) across five graph sizes
(Table~\ref{tab:compression-scale}).

\begin{table}[t]
\centering
\caption{Deductive compression at \(\mu=0.3\) across scales.}
\label{tab:compression-scale}
\renewcommand{\arraystretch}{1.15}
\setlength{\tabcolsep}{4pt}
\footnotesize
\begin{tabular}{rrrrrr}
\toprule
\(|\mathcal{V}|\) &
  \(|A|\) &
  \(|S_O|\) &
  \(\rho_{\mathrm{comp}}\) &
  \(\rho_{\mathrm{ent}}\) &
  \(d_{\max}\) \\
\midrule
     50 &       188 &        949 &  0.764 &  0.151 & 10 \\
    200 &     1\,705 &    14\,725 &  0.775 &  0.090 &  5 \\
    500 &     5\,179 &    84\,229 &  0.754 &  0.046 &  5 \\
  1\,000 &    12\,305 &   323\,705 &  0.742 &  0.028 &  5 \\
  2\,000 &    24\,304 & 1\,252\,504 &  0.719 &  0.014 &  6 \\
\bottomrule
\end{tabular}
\end{table}

The compression gain increases with scale.
At \(|\mathcal{V}|=2{,}000\), the single-shot ratio is
\(\rho_{\mathrm{comp}}=0.719\) (a~\(28\%\) blocklength
reduction), while the entropic ratio drops to
\(\rho_{\mathrm{ent}}=0.014\): in the i.i.d.\ regime, the
semantic rate is less than \(1.5\%\) of the classical rate.
The entropic gains are amplified at larger scales because
the deductive amplification factor
\(\gamma_{\mathrm{amp}}\) grows super-linearly
(Table~\ref{tab:scaling}): materializing \(30\%\) of a
larger closure produces a proportionally larger shortcut set
relative to the fixed core.

\paragraph*{Summary}
The small-scale instance verifies the deductive compression
ratio (\(n^*_{\Cn}/n^*_H=2/3\)), the vocabulary design
criterion (Proposition~\ref{prop:min-vocabulary}), and the
broadcast bottleneck
(Proposition~\ref{prop:broadcast-bottleneck}).
The medium-scale experiments confirm deductive amplification
exceeding two orders of magnitude (Experiment~1), semantic
overlap amplification by a factor of~\(3.2\) (Experiment~2),
the semantic leverage effect \(\Phi(R)>R\) for all \(R<1\)
(Experiment~3), and---crucially---the deductive compression gain
\(\rho_{\mathrm{ent}}<1\) at scale whenever materialized
shortcuts are present (Experiment~4), with entropic compression
exceeding an order of magnitude at moderate materialization
levels.
Together, the four experiments bracket the theoretical range:
irredundant-source deductive amplification (Experiments~1--3)
versus redundant-source deductive compression (Experiment~4),
closing the loop with Theorem~\ref{thm:tight-zero-rate} and
Corollary~\ref{cor:min-blocklength}.
In particular, the noise-pair indices for pair~\((1,2)\)
(\(\Phi_{\Atom}=0\), \(\Psi_{+}=0.900\))
quantify the vocabulary-mismatch impairment that is invisible to
the information-theoretic invariants
(\(I_{\mathrm{sem}}^{12}=I_{\mathrm{sem}}^{12'}=2.273\)),
underscoring the diagnostic complementarity of the six invariant
families.
The semantic capacity \(C_{\mathrm{sem}}^{ij}\), computed via
Blahut--Arimoto for all three pairs, verifies the data
processing chain
\(I_{\mathrm{sem}}^{ij}\le C_{\mathrm{sem}}^{ij}\le C(W)\)
and reveals a capacity-level vocabulary bottleneck: when the
receiver's vocabulary is smaller
(\(|S_O^{(3)}|=6<8=|S_O^{(1)}|\)),
\(C_{\mathrm{sem}}^{13}=1.958<2.280=C_{\mathrm{sem}}^{12'}\),
demonstrating that vocabulary mismatch constrains not only
closure fidelity
\textup{(Proposition~\ref{prop:broadcast-bottleneck})} but
also the maximum achievable mutual information.

\section{Conclusion}
\label{sec:conclusion}

This paper has developed a rate--distortion theory for semantic
communication grounded in formal proof systems.
The framework rests on three pillars: an axiomatic
\emph{information model} with computable enabling maps
(Section~\ref{sec:model}); a \emph{semantic channel} built as a
composition of enabling kernels
(Section~\ref{sec:channel}); and an \emph{overlap-based}
heterogeneous multi-agent theory
(Section~\ref{sec:application}).

The central quantitative finding is the \emph{deductive
compression gain}.
Under a closure-based fidelity criterion that accepts any
reconstruction preserving the deductive closure, the minimum
blocklength drops from
\(n_H^*\approx\log|S_O|/C(W)\) to
\(n_{\Cn}^*\approx\log|\Atom(S_O)|/C(W)\), yielding a
compression ratio \(\log|\Atom(S_O)|/\log|S_O|<1\) that is
invariant under receiver vocabulary heterogeneity.
This gain arises because the receiver's inference engine
reconstructs all redundant states from the irredundant core at
zero additional channel cost.

The tight zero-distortion semantic rate
\(R_{\mathrm{sem}}(0)=P_A\,H(\pi_A)\)
(Theorem~\ref{thm:tight-zero-rate}) and the full
rate--distortion decomposition
(Theorem~\ref{thm:rd-decomposition}) show that redundant states
are invisible to both rate and distortion under closure fidelity.
The semantic source--channel separation theorem
(Theorem~\ref{thm:sem-source-channel}) exhibits a \emph{semantic
leverage} phenomenon with leverage factor
\(\Lambda_\infty=\log|S_O|/(P_A\,H(\pi_A))>1\): under closure
fidelity the required source rate drops from \(H(P_O)\) to
\(P_A\,H(\pi_A)\), enabling the same knowledge base to be
communicated with proportionally fewer channel uses---not by
violating the Shannon capacity (the data processing bound
\(C_{\mathrm{sem}}\le C(W)\) remains in force) but because
closure-based fidelity renders redundant states ``free.''

The rate--delay--distortion surface
(Theorem~\ref{thm:rate-delay-distortion}) reveals a fundamental
\emph{depth-for-rate exchange} that has no classical counterpart:
each additional derivation step~\(\delta\) at the receiver
renders a new stratum of states redundant, reducing the effective
source entropy along the filtration
\(R_{\mathrm{sem}}(0,0)=H(P_O)\ge\cdots
\ge R_{\mathrm{sem}}(0,\mathsf{D_d})=P_A\,H(\pi_A)\).
The critical delay \(\delta^*\) below which closure-reliable
communication is impossible
(Corollary~\ref{cor:semantic-nyquist}) is a semantic analogue of
the Nyquist sampling period, and the marginal rate of delay
quantifies the ``value'' of one derivation step in bits of
channel capacity.
The hard-budget staircase can be relaxed to an
expected-budget model via time-sharing among transmitted
bases \textup{(Remark~\ref{rem:comm-comp-exchange})},
yielding a convex rate--computation tradeoff whose Lagrangian
multiplier prices one unit of receiver inference in bits of
communication rate---an exchange rate absent from classical
information theory, where the decoder's computation is treated
as a free resource.

The strengthened semantic Fano inequality
\textup{(Theorem~\ref{thm:semantic-fano-tight})} provides a
converse bound in which both the reference level and the
penalty term are improved: the former drops from \(H(P_O)\) to
\(P_A\,H(\pi_A)\), absorbing all redundant-state entropy, and
the latter involves \(\log|A|\) rather than \(\log|S_O|\).
Combined with the semantic source--channel separation theorem,
this characterizes the operational regime where semantic
compression strictly outperforms symbol-level compression.
Six families of computable semantic channel invariants
(Theorem~\ref{thm:invariant-summary}) provide a multi-scale
fingerprint of channel quality: from set-level fidelity metrics
that depend only on the knowledge-base pair, through noise-pair
probabilistic indices that capture core preservation and
hallucination probabilities, to information-theoretic quantities
that bound achievable throughput.

In the heterogeneous multi-agent setting, the overlap
decomposition (Definition~\ref{def:overlap-decomposition})
translates knowledge-base structure into two binary feasibility
tests---no core loss (\(A_{-}^{ij}=\varnothing\)) and no
non-derivable surplus (\(S_{+,n}^{ij}=\varnothing\))---that
fully determine whether perfect closure fidelity is achievable
(Proposition~\ref{prop:overlap-fidelity}).
The broadcast extension reveals a \emph{semantic bottleneck}: a
receiver whose vocabulary does not cover the sender's irredundant
core cannot achieve \(\mathsf{F}_{\Cn}=1\) regardless of carrier
channel quality, blocklength, or coding strategy
(Proposition~\ref{prop:broadcast-bottleneck}).
Resolving this bottleneck requires a design-time vocabulary
augmentation, fundamentally different from the encoder/decoder
optimization that suffices classically.

\subsection*{Relation to Prior Work}

The framework complements several existing lines of research.
The synonymous-mapping theory of Niu and
Zhang~\cite{niu2024mathematical,zhang2024modern} achieves a
semantic leverage effect through source-side equivalence-class
collapsing (their ``\(C_s\ge C\)'' result quantifies semantic
throughput gains under a synonymous-mapping fidelity criterion;
it does not violate the Shannon capacity upper bound
\(C_{\mathrm{sem}}\le C(W)\), which remains valid by data
processing);
our framework achieves a complementary gain through receiver-side
deductive reconstruction.
A unified theory combining both mechanisms---synonymous
collapsing of the irredundant core followed by deductive
expansion at the receiver---could potentially compound the
two gains; establishing the precise interaction is an open
problem.
The semantic channel coding theorem of Ma
et~al.~\cite{ma2025theory} and the companion tools
of~\cite{han2025extended,liang2025semantic} operate within the
synonymous-mapping paradigm; our two-layer code mechanism and
closure-based fidelity criterion handle the structured
knowledge-base setting and extend naturally to heterogeneous
vocabularies.
On the multi-agent front, the overlap decomposition and semantic
bottleneck phenomenon provide a coding-theoretic complement to
the Bayesian approach of Seo
et~al.~\cite{seo2023bayesian} and the modal-logic framework of
Alshammari and Bennis~\cite{alshammari2026logic}: where those
works quantify inference cost and resilience conditions, our
framework quantifies the minimum number of channel uses needed to
overcome vocabulary mismatch.
The irredundant core can also be viewed as a form of query-aware
source compression related to database-theoretic notions of view
materialization~\cite{marx2013tractable,abokhamis2025jaguar}.

\subsection*{Limitations and Future Directions}

The common proof system assumption
(Assumption~\ref{assump:common-ps}) could be relaxed to
\emph{heterogeneous} proof systems where each agent uses a
sub-system \(\mathsf{PS}^{(i)}\subseteq\mathsf{PS}\), yielding
richer compression/fidelity trade-offs.
The current single-letter coding theorems invite a
\emph{multi-letter} extension for temporally correlated
knowledge-base streams, connecting the framework to ergodic
source theory.
The relay scenario (Remark~\ref{rem:relay-scenario}), where
intermediate inference may change the effective end-to-end
capacity, requires multi-hop coding theorems built on the
composition machinery of
Definition~\ref{def:model-composition}.
Scaling core extraction and closure computation to knowledge
graphs with millions of entities, and the formal connection
between derivation depth~\(\Dd\) and Bennett's logical
depth~\cite{bennett1988logical}, are further promising avenues.

From a structural standpoint, the receiver's inference engine
plays a role analogous to \emph{decoder side information} in
Wyner--Ziv coding~\cite{wyner1976rate}: the shared proof system
\(\mathsf{PS}\) provides the decoder with a ``structured
codebook'' (the closure operator~\(\Cn\)) that reduces the
effective source rate without requiring explicit side-information
transmission.
Formalizing this analogy---e.g., by deriving a semantic
Wyner--Ziv theorem in which the side information is the proof
system itself---would connect the present framework to the
established body of source coding with side information and may
yield tighter bounds when the proof system is only partially
shared.


\section*{Acknowledgment}

During the writing and revision of this paper, I received many insightful comments from Associate Professor Rui Wang of the School of Computer Science at Shanghai Jiao Tong University and also gained much inspiration and assistance from regular academic discussions with doctoral students Yiming Wang, Chun Li, Hu Xu, Siyuan Qiu, Zeyan Li, Jiashuo Zhang, Junxuan He, and Xiao Wang. I hereby express my sincere gratitude to them.

\appendices

\counterwithin{definition}{section}
\counterwithin{axiom}{section}
\counterwithin{assumption}{section}
\counterwithin{theorem}{section}
\counterwithin{lemma}{section}
\counterwithin{proposition}{section}
\counterwithin{corollary}{section}
\counterwithin{remark}{section}
\counterwithin{example}{section}

\section{Axiomatic Foundations}
\label{app:axioms}

This appendix provides the full axiomatic development of the
information model framework summarized in
Section~\ref{sec:model}.

\subsection{Logical Language and Expressible State Sets}
\label{subsec:logic}

Throughout we fix a many-sorted logical language
\(\mathcal{L}=\mathrm{FO(LFP)}[\Sigma]\)~\cite{immerman1999descriptive}
with sorts including \(\mathsf{Obj}\), \(\mathsf{Time}\), and
\(\mathsf{Carrier}\).

\begin{assumption}[Finite ordered structures]
\label{assump:ordered-structures}
We restrict attention to \emph{finite} \(\mathcal L\)-structures.
When descriptive-complexity claims are invoked, we work over
\emph{ordered} finite structures
\cite{immerman1982relational,vardi1982complexity}.
\end{assumption}

\begin{assumption}[Semantic sublanguage]
\label{assump:semantic-sublanguage}
Fix a partition
\(\Sigma=\Sigma_{\mathrm{sem}}\,\dot\cup\,\Sigma_{\mathrm{rep}}\),
where \(\Sigma_{\mathrm{rep}}\) contains auxiliary symbols (including
any built-in order).
The designated semantic sublanguage is
\(\mathcal L_{\mathrm{sem}}:=
\mathcal L\!\upharpoonright_{\Sigma_{\mathrm{sem}}}\).
All notions of semantic equivalence are relative
to~\(\mathcal L_{\mathrm{sem}}\).
\end{assumption}

\begin{definition}[Expressible state sets]
\label{def:expressible-states}
Let \(\mathfrak R\) be a finite \(\mathcal L\)-structure and let
\((X,T)\) be definable subdomains of the object and time sorts.
An \emph{(object--time) state domain} over~\((X,T)\) is an
\(\mathcal L_{\mathrm{sem}}\)-definable binary relation \(S(x,t)\)
with
\(\mathfrak R\models\forall x\,\forall t\,
[S(x,t)\to(\delta_X(x)\wedge\delta_T(t))]\).
A set \(A\subseteq X\times T\) is
\(\mathcal L_{\mathrm{sem}}\)-\emph{expressible} over~\(S\) if
there exists an \(\mathcal L_{\mathrm{sem}}\)-formula
\(\varphi(x,t)\) with
\((x,t)\in A\Leftrightarrow
\mathfrak R\models S(x,t)\wedge\varphi(x,t)\).
We write \(s\in S\) as shorthand for a pair \((x,t)\) with
\(S(x,t)\), and set \(\timeindex(s):=t\).
\end{definition}

\subsection{Information Model and Enabling Mechanisms}
\label{subsec:info-instance}

\begin{definition}[State spaces]
\label{def:binary-attribute}
Information is modeled using object/carrier domains \(O,C\) with
time domains \(T_O,T_C\) and state sets
\(S_O\subseteq O\times T_O\),
\(S_C\subseteq C\times T_C\).
We write \(s_o=(o,\tau)\in S_O\) and \(s_c=(c,\theta)\in S_C\).
\end{definition}

\begin{axiom}[Time domains and precedence]
\label{ax:time-domains}
Fix a finite \(\mathcal L\)-structure~\(\mathfrak R\).
The state sets \(S_O,S_C\) are
\(\mathcal L_{\mathrm{sem}}\)-definable.
There exist \(\mathcal L_{\mathrm{sem}}\)-definable linear orders
\(\prec_O\) on~\(T_O\), \(\prec_C\) on~\(T_C\), and a
cross-domain precedence \(\prec_{OC}\subseteq T_O\times T_C\)
that is monotone with respect to~\(\prec_O\) and~\(\prec_C\),
and mildly total (every \(\tau\in T_O\) has some
\(\theta\in T_C\) with \(\tau\prec_{OC}\theta\)).
\end{axiom}

\begin{axiom}[State representation]
\label{ax:state-representation}
There exist injective encodings
\(\enc_O:S_O\to\{0,1\}^*\) and \(\enc_C:S_C\to\{0,1\}^*\) such
that semantic-time, carrier-time, and cross-domain precedence
predicates, as well as membership, are decidable from the codes.
\end{axiom}

\begin{axiom}[Enabling mapping]
\label{ax:enabling-mapping}
There exists a relation
\(R_{\mathcal E}\subseteq O\times T_O\times C\times T_C\)
inducing a set-valued map
\(\mathcal E:S_O\Rightarrow S_C\) satisfying:
\textup{(E1)}~totality (\(\mathcal E(s_o)\neq\varnothing\) for
all~\(s_o\));
\textup{(E2)}~coverage
(\(\bigcup_{s_o}\mathcal E(s_o)=S_C\));
\textup{(E3)}~existence of a computable enabling selector;
\textup{(E4)}~precedence compatibility
(\(s_c\in\mathcal E(s_o)\) implies
\(\timeindex_O(s_o)\prec_{OC}\timeindex_C(s_c)\)).
\end{axiom}

\begin{definition}[Information model]
\label{def:info-instance}
An \emph{information model} is
\(\mathcal I=\langle O,T_O,S_O,C,T_C,S_C,R_{\mathcal E}\rangle\)
together with the temporal structure of
Axiom~\textup{\ref{ax:time-domains}}.
\end{definition}

\begin{definition}[Composition of information models]
\label{def:model-composition}
Two models \(\mathcal I_1,\mathcal I_2\) are \emph{composable} if
\(S_C^{(1)}=S_O^{(2)}\).
The composite enabling map is
\(\mathcal E_{2\circ 1}(s_o):=
\bigcup_{s_c\in\mathcal E_1(s_o)}\mathcal E_2(s_c)\).
\end{definition}

\begin{proposition}[Composition preserves enabling axioms]
\label{prop:composition-axioms}
If \((\mathcal I_1,\mathcal I_2)\) is composable and each satisfies
Axiom~\textup{\ref{ax:enabling-mapping}}, then
\(\mathcal E_{2\circ 1}\) satisfies
\textup{(E1)}--\textup{(E3)}.
\end{proposition}

\begin{proof}
\textup{(E1)}: totality of~\(\mathcal E_1\) yields
\(s_c\in\mathcal E_1(s_o)\); totality of~\(\mathcal E_2\) gives
\(\mathcal E_2(s_c)\neq\varnothing\), so
\(\mathcal E_{2\circ 1}(s_o)\supseteq\mathcal E_2(s_c)
\neq\varnothing\).
\textup{(E2)}: for any \(s'\in S_C^{(2)}\), coverage
of~\(\mathcal E_2\) gives~\(s_c\) with \(s'\in\mathcal E_2(s_c)\);
coverage of~\(\mathcal E_1\) gives~\(s_o\) with
\(s_c\in\mathcal E_1(s_o)\), whence
\(s'\in\mathcal E_{2\circ 1}(s_o)\).
\textup{(E3)}: if \(e_1,e_2\) are selectors,
\(e_{2\circ 1}:=e_2\circ e_1\) is computable and satisfies
\(e_{2\circ 1}(s_o)\in\mathcal E_2(e_1(s_o))
\subseteq\mathcal E_{2\circ 1}(s_o)\).
\end{proof}

\begin{remark}[Associativity]
\label{rem:composition-assoc}
Composition of enabling maps is associative: for a composable triple,
\(\mathcal E_{3\circ(2\circ 1)}
=\mathcal E_{(3\circ 2)\circ 1}\),
so iterated compositions are unambiguous.
\end{remark}

\subsection{Synonymous State Sets and Ideal Information}
\label{subsec:ideal}

\begin{definition}[\(\mathcal L_{\mathrm{sem}}\)-definable coding
  isomorphism]
\label{def:definable-coding}
A relation \(G_{12}\subseteq S_1\times S_2\) is an
\(\mathcal L_{\mathrm{sem}}\)-definable coding isomorphism graph if
it is definable by an \(\mathcal L_{\mathrm{sem}}\)-formula and
induces a bijection between~\(S_1\) and~\(S_2\)
(i.e., for every \(s_1\in S_1\) there is a unique
\(s_2\in S_2\) with \(G_{12}(s_1,s_2)\), and vice versa).
\end{definition}

\begin{definition}[Synonymous state sets]
\label{def:synonymous-states}
State sets \(S_1,S_2\) are \emph{synonymous}, written
\(S_1\syn S_2\), if there exists an
\(\mathcal L_{\mathrm{sem}}\)-definable coding isomorphism graph
\(G_{12}\) inducing a bijection \(\tau_{12}:S_1\to S_2\) that
preserves and reflects the induced time precedence:
\(s\prec_{S_1}s'\Leftrightarrow
\tau_{12}(s)\prec_{S_2}\tau_{12}(s')\).
\end{definition}

\begin{proposition}[\(\syn\) is an equivalence relation]
\label{prop:syn-equiv}
The relation \(\syn\) on
\(\mathcal L_{\mathrm{sem}}\)-definable state sets is reflexive,
symmetric, and transitive.
\end{proposition}

\begin{proof}
Reflexivity is witnessed by the identity graph
\(G(s,s'):=\mathbf{1}[s=s']\), which is
\(\mathcal L_{\mathrm{sem}}\)-definable.
Symmetry: if \(G_{12}\) witnesses \(S_1\syn S_2\), then
\(G_{21}(s_2,s_1):=G_{12}(s_1,s_2)\) witnesses \(S_2\syn S_1\).
Transitivity: given witnesses \(G_{12}\subseteq S_1\times S_2\) and
\(G_{23}\subseteq S_2\times S_3\), define
\(G_{13}(s_1,s_3):=\exists\,s_2\,[G_{12}(s_1,s_2)\wedge
G_{23}(s_2,s_3)]\).
This is \(\mathcal L_{\mathrm{sem}}\)-definable (closed under
quantification over its own sorts), and the uniqueness clauses of each
component ensure \(G_{13}\) is a bijection.
The induced map \(\tau_{13}=\tau_{23}\circ\tau_{12}\) preserves and
reflects precedence by composition.
\end{proof}

\begin{definition}[Ideal information]
\label{def:ideal-info}
An information model~\(\mathcal I\) is \emph{ideal} if
\(S_O\syn S_C\), witnessed by a coding isomorphism graph
\(G_{OC}\) with induced bijection \(\tau_{OC}\), and
\(\mathcal E(s_o)=\{\tau_{OC}(s_o)\}\) for every
\(s_o\in S_O\).
\end{definition}

\begin{assumption}[Common semantic universe]
\label{assump:semantic-universe}
The ambient set
\(\mathbb{S}_O\supseteq S_O\)
\textup{(}introduced in Section~\textup{\ref{subsec:proof-closure})}
is closed under \(\mathcal L_{\mathrm{sem}}\)-definable recodings:
if \(S\subseteq\mathbb{S}_O\) and \(S'\syn S\) with \(S'\) over the
same sorts, then \(S'\subseteq\mathbb{S}_O\).
Moreover, \(\mathbb{S}_O\) is effectively representable
\textup{(}the encoding~\(\enc_O\) and membership predicate extend to
all of~\(\mathbb{S}_O\)\textup{)}.
\end{assumption}

\begin{assumption}[Carrier representability]
\label{assump:carrier-rep}
There exists an \(\mathcal L_{\mathrm{sem}}\)-definable
\(S'_O\subseteq\mathbb{S}_O\) with \(S'_O\syn S_C\).
\end{assumption}

\begin{definition}[Noisy information]
\label{def:noisy-info}
Given~\(\mathcal I\) satisfying
Assumption~\textup{\ref{assump:carrier-rep}} and an
\(\mathcal L_{\mathrm{sem}}\)-definable
\(\tilde S_O\subseteq\mathbb{S}_O\) with \(\tilde S_O\syn S_C\)
\textup{(}witnessed by \(G\) with bijection
\(\tau_{OC}:\tilde S_O\to S_C\)\textup{)},
the \emph{noisy information} is
\(\tilde{\mathcal I}=
(\tilde S_O,S_C,G,\tau_{OC},\tau_{CO})\) with noise pair
\(S_O^-:=S_O\setminus\tilde S_O\) and
\(S_O^+:=\tilde S_O\setminus S_O\).
\end{definition}

\begin{proposition}[Existence of noisy information]
\label{prop:noisy-exist}
Under Assumption~\textup{\ref{assump:carrier-rep}}, for every
information model~\(\mathcal I\) there exists a noisy
information~\(\tilde{\mathcal I}\) with
\(\tilde S_O\syn S_C\).
If~\(\mathcal I\) is ideal, the perturbation is trivial.
\end{proposition}

\begin{proof}
Set \(\tilde S_O:=S'_O\) from
Assumption~\ref{assump:carrier-rep}.
If~\(\mathcal I\) is ideal, \(S_O\syn S_C\), so choosing
\(S'_O:=S_O\) gives \(S_O^-=S_O^+=\varnothing\).
\end{proof}

\begin{remark}[Computability of semantic invariants for noisy bases]
\label{rem:noisy-ideal-connection}
Since \(\tilde S_O\subseteq\mathbb{S}_O\) is finite (in bijection
with the finite set~\(S_C\)), the invariants
\(\mathsf{A}(\tilde{\mathcal I})=|\Atom(\tilde S_O)|\) and
\(\mathsf{D_d}(\tilde{\mathcal I})=
\max_{q\in\tilde S_O}\Dd(q\mid\Atom(\tilde S_O))\) are
well-defined, finite, and computable by the same reasoning as
Theorem~\textup{\ref{thm:semantic-invariants}(i)}.
\end{remark}

\bibliographystyle{IEEEtran}
\bibliography{ref}

\vfill
\end{document}